\documentclass[10pt]{article}
\pdfoutput=1
\usepackage{amssymb,amsmath,xcolor,empheq,physics}
\numberwithin{equation}{section}
\usepackage{epsfig,array}
\usepackage{epstopdf}
\usepackage{latexsym}
\usepackage{graphicx}
\usepackage{caption}

\usepackage{booktabs}
\usepackage{bbm}
\usepackage{enumitem}
\usepackage[numbers,compress]{natbib}
\usepackage{bm}
\usepackage[T1]{fontenc}
\usepackage[linesnumbered,ruled]{algorithm2e}
\usepackage{amsthm,complexity}
\usepackage{fancyhdr}

\newcommand{\identity}{\mathbb{I}}
\newcommand{\Hin}{H_{\text{in}}}
\newcommand{\CNOT}{\text{CNOT}}
\newcommand{\SWAP}{\text{SWAP}}
\newcommand{\homologyProb}{{\sc $\text{Homology}_{\mathcal{I}}(K,l)$}}
\newcommand{\PromisehomologyProb}{{\sc $\text{Promise-Homology}_{\mathcal{I}}(K,l)$}}
\newcommand{\Hbravyi}{H_{\text{Bravyi}}}
\usepackage{natbib}
\usepackage{fancybox}

\usepackage[margin=20pt]{caption}
\usepackage{subcaption}

\usepackage[toc]{appendix}

\usepackage{color}
\usepackage{datetime}
\usepackage[
      colorlinks=false,
      linkcolor=darkblue,  
      urlcolor=blue,    
      filecolor=blue,     
      citecolor=red,
linktocpage=true,
      pdfstartview=FitV,
      bookmarksopen=true,
	  hidelinks
      ]{hyperref}
      
\usepackage[nameinlink,capitalise]{cleveref}
\crefname{section}{Section}{Sections}
\crefname{subsection}{subsection}{subsections}
\crefname{theorem}{Theorem}{Theorems}
\crefname{corollary}{Corollary}{Corollaries}
\crefname{lemma}{Lemma}{Lemmas}
\crefname{appendix}{Appendix}{Appendices}
\crefname{definition}{Definition}{Definitions}
\crefname{equation}{eq.}{eqs.}

\ifpdf
\fi

\newcommand*{\boxedcolor}{black}
\makeatletter
\renewcommand{\boxed}[1]{\textcolor{\boxedcolor}{%
  \fbox{\normalcolor\m@th$\displaystyle#1$}}}
\makeatother

\usepackage{calligra}
\DeclareMathAlphabet{\mathcalligra}{T1}{calligra}{m}{n}

\definecolor{cardinal}{rgb}{0.6,0,0}
\definecolor{darkgreen}{rgb}{0,0.5,0}
\definecolor{golden}{rgb}{0.92, 0.7, 0}
\definecolor{midnight}{rgb}{0, 0, 0.5}
\definecolor{darkblue}{rgb}{0.2, 0, 0.8}

\definecolor{tk}{RGB}{246,76,246}

\def\cN{{\cal N}}
\def\cO{{\cal O}}
\def\cP{{\cal P}}
\def\cQ{{\cal Q}}

\newcommand{\Hclock}{H_\textrm{clock}}

\newcommand{\Hout}{H_\textrm{out}}
\newcommand{\Hprop}{H_\textrm{prop}}

\newcommand{\Hpropt}{H_\textrm{prop,t}}

\newcommand{\DQC}{\mathsf{DQC1}}

\newcommand{\be}{\begin{equation}} \newcommand{\ee}{\end{equation}}
\newcommand{\bea}{\begin{equation} \begin{aligned}} \newcommand{\eea}{\end{aligned} \end{equation}}
\newcommand{\bmu}{\begin{multline}} \newcommand{\emu}{\end{multline}}

\newcommand\equ[1] {\begin{equation}#1\end{equation}}

\newcommand\eqss[1] {\begin{align}\begin{split}#1\end{split}\end{align}}

\usepackage[cmtip,all]{xy}
\newcommand{\longsquiggly}{\xymatrix{{}\ar@{~>}[r]&{}}}

\def\cH{\mathcal H}

\newtheorem{theorem}{Theorem}

\newtheorem{definition}{Definition}

\newtheorem{lemma}{Lemma}
\newtheorem{corollary}{Corollary}

\newcommand{\yes}{\mathsf{YES}}
\newcommand{\no}{\mathsf{NO}}
\usepackage{tikz}
\usetikzlibrary{quantikz}

\newcounter{mycount}
\newcommand\myprob[3]{%
   \stepcounter{mycount}
 \par\noindent   {\bfseries Definition\  \themycount} (#1)\\
   {\bfseries Input}: #2\\
   {\bfseries Problem}: #3\par
}

\newcommand\mypromprob[4]{%
   \stepcounter{mycount}
 \par\noindent  {\bfseries Definition\  \themycount} (#1)\\
   {\bfseries Input}: #2\\
   {\bfseries Promise}: #3\\
   {\bfseries Problem}: #4\par
}

\begin{document}  

\begin{titlepage}

\begin{center} 

\vspace*{1.5cm}

{\huge Clique Homology is $\QMA_{1}$-hard}

\vspace{1.5cm}

{\large   Marcos Crichigno${}^{(1)}$ and Tamara Kohler${}^{(2)}$ }

\bigskip
\bigskip
${}^{(1)}$ Blackett Laboratory, Imperial College London  
\vskip 5mm


${}^{(2)}$ Instituto de Ciencias Matemáticas, Madrid
\vskip 5mm

\vskip 5mm

\texttt{~crichigno@proton.me,~tamara.kohler@icmat.es}\\

\end{center}

\vspace{0.5cm}

\noindent 

\abstract{
We tackle the long-standing question of the computational complexity of determining homology groups of simplicial complexes, a fundamental task in computational topology, posed by Kaibel and Pfetsch  20 years ago. We show that this decision problem is $\QMA_{1}$-hard. Moreover, we show that a version of the problem satisfying a suitable promise and certain constraints is contained in $\QMA$.
This suggests that the seemingly classical problem may in fact be quantum mechanical. 
In fact, we are able to significantly strengthen this by showing that the  problem remains $\QMA_{1}$-hard in the case of clique complexes, a family of simplicial complexes specified by a graph which is relevant to the problem of  topological data analysis. The proof combines a number of techniques  from Hamiltonian complexity and  homological algebra. We discuss potential implications for the problem of quantum advantage in topological data analysis.  }

\thispagestyle{fancy}
 \fancyhead[R]{Imperial/TP/2022/MC/01} 

\end{titlepage}



\setcounter{tocdepth}{2}
\tableofcontents

\newpage


\section{Introduction and Main Result}

Simplicial complexes are fundamental objects in topology.  They are constructed by gluing together discrete building blocks consisting of points, edges, triangles, tetrahedra, and higher-dimensional analogs called simplices. Simplicial complexes are often used as approximations to smooth spaces, with the advantage that due their discrete nature they lend themselves to computation. Indeed, simplicial complexes are the main objects studied in computational topology \cite{dey1999computational,edelsbrunner2010computational} and play a fundamental role in various areas of pure and applied mathematics and including data analysis \cite{carlsson2009topology}. 

\

A fundamental task in topology is to determine the topological properties of such spaces. Topological properties of space refer to properties which are unchanged under smooth deformations of the space. A prime example of such a property is the existence of ``holes'' in the space, which are formally captured by the notion of homology. Perhaps the most fundamental problem  in computational topology is the ``homology problem'' which is informally stated as:
\begin{center}
\it Given a simplicial complex of dimension $n$, does it have a $k$-dimensional hole?
\end{center}

 \noindent The study of homology   goes back to the early days of topology in the 1850s. 
 The question of its complexity as a {\it computational} problem  was  posed formally  as an open problem 20 years ago by Kaibel and Pfetsch in \cite{2002math......2204K} (Problem 33). Partial progress was made in   \cite{ADAMASZEK20168} where the problem was shown to be $\NP$-hard in, but its precise complexity remained open. 
 
 \

 In this paper, we take steps toward showing that this problem is intrinsically quantum mechanical. 
 More precisely, we  show that the simplicial homology problem (which is a decision problem) is $\QMA_1$-hard.
 Moreover, we show that under certain constraints a version of the problem with a suitable promise is contained in $\QMA$, with $\QMA$ the class of problems whose solution can be  efficiently checked by a quantum computer and $\QMA_1$ the one-sided error version of $\QMA$  (precise definitions are given below). It also follows as a simple corollary of our proof that the counting version of the problem is $\#\BQP$-hard, with $\#\BQP$ the quantum analog of the classical counting class $\#\P$.\footnote{It turns out that the class $\#\BQP$ is equivalent to $\#\P$ under weakly parsimonious reductions \cite{Brown_2011}. }

\

It may come as a surprise that the complexity of determining whether a simplicial complex has a $k$-dimensional hole (a seemingly classical problem in computational topology) could be captured by {\it quantum}  complexity classes. There is a good reason for this, however, as pointed out in \cite{Crichigno:2020vue,Cade:2021jhc}. It was shown by Witten in the seminal work \cite{Witten:1982im} that elements of homology are in 1-to-1 correspondence with ground states of ``supersymmetric'' (SUSY) systems,  a special class of quantum mechanical systems with a symmetry relating bosonic states and fermionic states.\footnote{The correspondence is usually states in terms of cohomology. Cohomology is a dual version of homology but for our purposes this distinction is not important. See \cite{Cade:2021jhc} for a discussion. } 

This connection was exploited in \cite{Cade:2021jhc} to define a general version of the homology problem (beyond the case of simplicial complexes but containing it as a special case) and shown to be  $\QMA_1$-hard and contained in $\QMA$. Here we  improve on this result, showing the homology problem remains $\QMA_1$-hard for the case of simplicial complexes originally posed in \cite{2002math......2204K}. In fact, we significantly improve on this by showing that the problem remains $\QMA_1$-hard for the special case of clique complexes $Cl(G)$, a special family  of simplicial complexes defined by a graph $G$ which are particularly relevant to topological data analysis (TDA). For an introduction to TDA and applications see  \cite{carlsson2009topology,reviewcarlsson,wasserman2016topological}. For applications of clique homology to computational neuroscience see \cite{petri2014homological,giusti2016two,reimann:hal-01706964}.

\begin{figure}[htbp] 
\begin{center}
\includegraphics[]{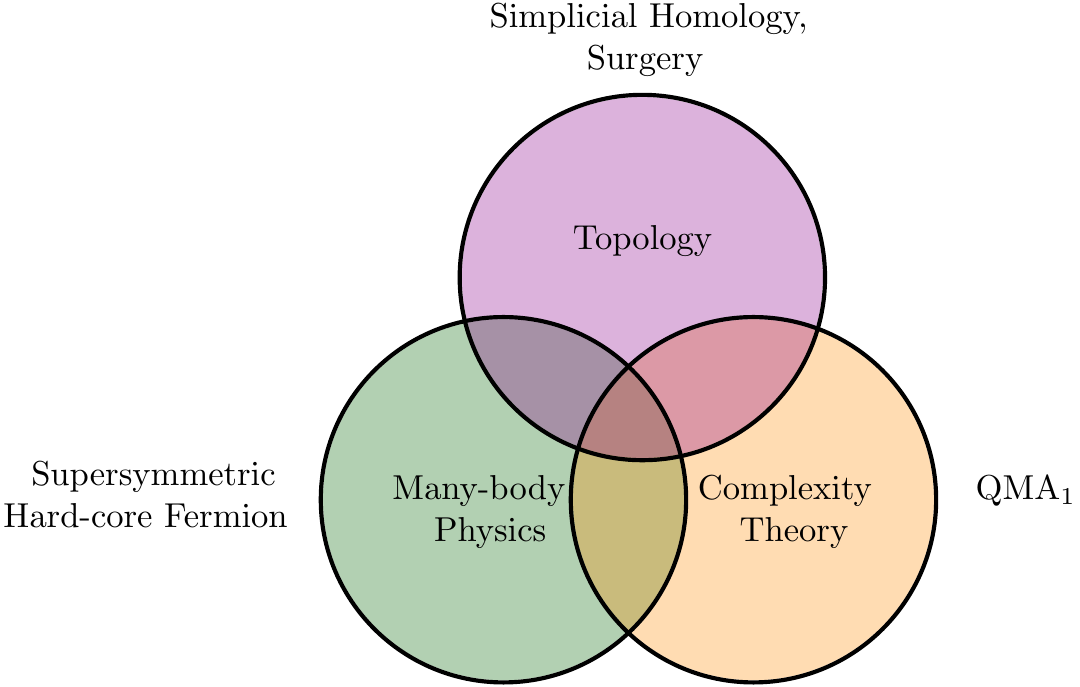}
\caption{The complexity of the clique homology problem, posed 20 years ago, lies at the intersection of (quantum) many-body physics, topology, and quantum complexity theory. Techniques from these three areas are combined to establish $\QMA_1$-hardness of the clique homology problem.  }
\label{venn}
\end{center}
\end{figure}

\ 

The relation between homology and supersymmetric ground states makes it clear that the homology problem can be contained in $\NP$ only if there is an efficient description of the ground states of supersymmetric systems, and if their energy can be computed efficiently. This is generally not the case for most many-body systems but one may wonder if it could be the case for {\it supersymmetric} many-body  systems. After all, what makes supersymmetric systems attractive to mathematical physicists is precisely that they can often be solved analytically. Our result can then be viewed as showing that there can be no efficient description of the ground states of supersymmetric many-body systems, unless $\QMA_1=\NP$.\footnote{For specific graphs with significant structure, this need not be the case. For instance, for the case of chordal graphs it was shown \cite{ADAMASZEK20168} that the homology problem for the independence complex is in $\NP$ and thus we expect an efficient description of the fermion hard-core model ground states in this case. } On the other hand, this connection also makes it clear that a suitable promise version of the homology problem with certain locality constraints can be placed inside $\QMA$, as finding ground states of supersymmetric Hamiltonians is a special case of finding ground states of generic, non-supersymmetric, Hamiltonians.

\ 

It is well known that another area where topology and quantum computation interrelate is in the study of knot invariants and topological quantum computation \cite{freedman2002simulation,freedman2002modular}. Computing the {\it exact} Jones polynomial is $\#\P$-hard but efficient {\it approximations} can obtained by various quantum algorithms, which are unlikely to be achieved by classical counterparts.  Indeed, a certain approximation was shown to be $\BQP$-complete in \cite{aharonov2009polynomial} and another approximation to be  $\DQC$-complete in \cite{shor2007estimating}.   We note that there is an interesting parallel between these two strands.   Once again, the connection between topology and quantum computation arises from the study of certain {\it physical} systems. In the former case via topological quantum field theories (TQFTs) and in the latter via supersymmetric (SUSY) many-body systems.

\

The relationship between simplicial homology and supersymmetric quantum mechanics similarly suggests that computational homology may be an area where quantum computers can lead to (provable) exponential speedups. Indeed, as Feynman originally suggested, quantum computers should lead to exponential speedups when simulating quantum mechanical systems. One would expect this to be the case regardless of whether the quantum mechanical system at hand is actually realized in Nature.  We thus view supersymmetry in this context as a ``spotlight'' in the space of computational problems, revealing that certain problems which are not {\it obviously} quantum mechanical are, in fact, secretly so and thus are promising targets for exponential quantum speedups. We consider the results in this paper as supporting this but further work is required.

\subsection{Precise statement}

\noindent  The main result of the paper is the following Theorem:
\begin{theorem}\label{thmQMA1}
 The homology problem is $\QMA_{1}$-hard for clique complexes $Cl(G)$, when the graph $G$ is given as input. For inputs satisfying a suitable promise, the homology problem is contained in $\QMA$.
\end{theorem}

\begin{proof}[Proof sketch]

We  give a schematic overview of the proof (the technical proof is the content of Section~\ref{sec:QMA1hard}). Although the proof is somewhat technical the main idea should be conceptually clear. The goal is to encode the computational history of a $\QMA_{1}$ verification procedure into a clique complex $\Sigma=Cl(G)$ of a graph $G$, such that that for every accepting witness there is a corresponding $(n-1)$-dimensional hole in $\Sigma$, and vice versa. To achieve this we go via Bravyi's "circuit-to-Hamiltonian" construction \cite{bravyi2011efficient} showing that the accepting witnesses of a $\QMA_{1}$ circuit are in 1-to-1 correspondence the set of solutions to  
the solutions to \equ{
\sum_{a}\Pi_a^{\text{Bravyi}}\ket{\psi}=0\,,
}
where the $\Pi_a^{\text{Bravyi}}$ are a $\text{poly}(n)$ collection of 4-local, projectors. These depend on the details of the universal gate set used (details in Section~\ref{sec:QMA1hard} and the Appendix).
\begin{figure}[]
\begin{center}
\includegraphics[scale=0.9]{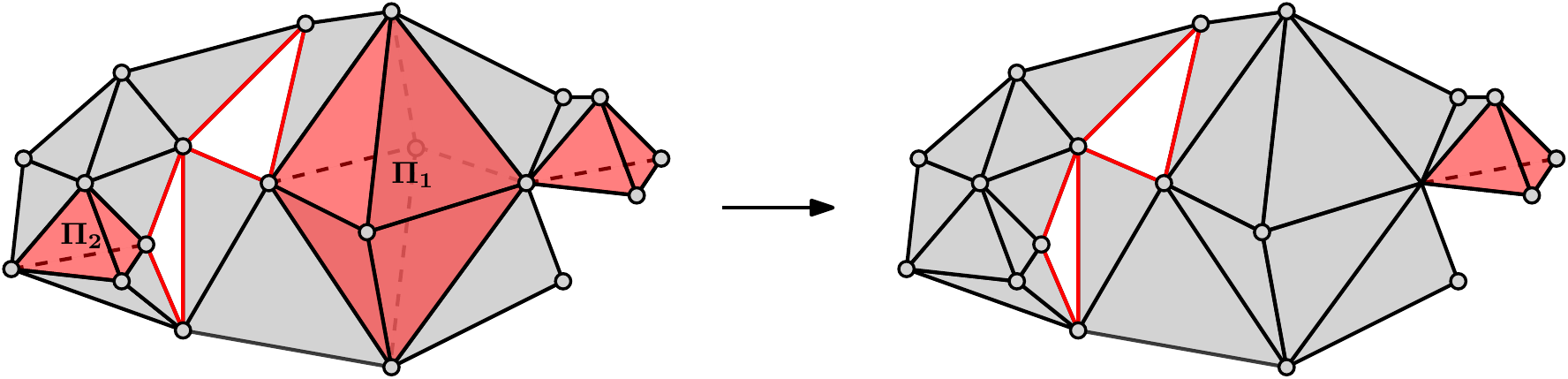}
\caption{Schematic representation of the reduction from quantum $k$-$\SAT$ to the homology problem. We first construct an $n$-dimensional simplicial complex with $2^n$ number of $(n-1)$-dimensional holes which are identified with the computational basis states in the Hilbert space of $n$ qubits. Given an $n$-qubit projector $\Pi = \sum_{A}\ket{\psi_A}\bra{\psi_A}$ we identify each state $\ket{\psi_A}$ with a hole (or ``linear combinations'' of holes) and fill it in, thus eliminating the state $\ket{\psi_A}$ from homology. Then, a given instance of quantum $k$-$\SAT$ has a satisfying assignment if and only if the final space has any $(n-1)$-dimensional holes left.}
\label{fig:simplicialcomplex}
\end{center}
\end{figure}

To construct $\Sigma=Cl(G)$ we proceed as follows. We first begin by constructing an $n$-dimensional clique complex $\Sigma^{(0)}=Cl(G^{(0)})$ with $2^n$ number of $(n-1)$-dimensional holes which are identified with the computational basis states in the Hilbert space of $n$ qubits. This simplicial complex corresponds to the space of solutions to a trivial instance of quantum $k$-$\SAT$ where there are no projectors. Then, we take a first projector $\Pi_1^{\text{Bravyi}}=\sum_{A=1}^{r_1} \ket{\psi_A}\bra{\psi_A}$, where  $r_{1}$ is the rank of $\Pi_1^{\text{Bravyi}}$. This first projector then excludes the states $\{\ket{\psi_A}\}$ from the space of satisfying assignments. Correspondingly, we map each $\ket{\psi_A}$ to a particular $(n-1)$-dimensional hole in $\Sigma^{(0)}$, which we ``fill in'' and is thus no longer a hole. This is accomplished by adding an appropriate number of simplices, obtaining a new simplicial complex  $\Sigma^{(1)}$ (see \cref{fig:simplicialcomplex}).    We repeat this process for each $\Pi_a^{\text{Bravyi}}$, making sure that at every stage that $\Sigma^{(i)}$ is the clique complex of a graph $G^{(i)}$. When the process is completed, we are left with a complex $\Sigma=Cl(G)$ whose $(n-1)$-dimensional holes are in 1-to-1 correspondence with the satisfying assignments $\sum_{a=1} \Pi_a^{\text{Bravyi}}\ket{\psi}=0$ and thus to the number of accepting witnesses of the $\QMA_{1}$ circuit $U$: 
\equ{
 \Pr[U(\ket{w}) = 1] = 1 \qquad \qquad \xleftrightarrow{\text{1-to-1}} \qquad \qquad \text{$\ket{w}$ is an $(n-1)$-dimensional hole of $Cl(G)$}
}
Since the reduction is parsimonious (i.e., it preserves the number of solutions), it also follows as a simple corollary that determining the Betti number of the clique complex is $\#\BQP$-hard. 
\end{proof}

A few comments are in order. It is important to note that projectors which are not diagonal in the computational basis are associated to ``linear combinations'' of holes in $\Sigma$ which are filled in. Indeed, it is the fact that a linear combination of holes is well defined in the theory of homology that makes this correspondence possible and makes the homology problem quantum mechanical. We give an introduction to the theory of homology in Section~\ref{sec:Background}.

\

Although conceptually simple, there are a number of technical challenges along the way. The first is that the states $\ket{\psi_A}$ above can generically have complex coefficients in the computational basis but our identification of these states with holes in $\Sigma$ requires that these coefficients are all integer (up to an overall normalization). This is achieved by choosing a universal gate set in the $\QMA_1$ circuit with only rational coefficients. Second, ensuring that at every stage of our construction one has a graph  $G^{(i)}$ whose clique complex precisely captures the accepting witnesses requires a procedure of cutting and gluing simplicial complexes, which we term "simplicial surgery". This leads to an unwieldy gadget shown in \cref{fig:ComplementPyth}. The construction of this gadget is the most technically challenging step in the reduction and  would be impossible to find without the aid of topology.

\

For clarity of presentation we have decided to avoid any reference to supersymmetry in the proof of Theorem~\ref{thmQMA1}  and relegate the discussion of supersymmetry  to Section~\ref{sec:connection to SUSY}. Several of the ideas for the reduction, however, were motivated by the perspective of Hamiltonian complexity of supersymmetric systems. Readers not interested in this connection can safely ignore  Section~\ref{sec:connection to SUSY}. On the other hand, readers who are interested in learning supersymmetry may find this a sufficient reason to do so and we provide a self-contained review in  Section~\ref{sec:connection to SUSY} (see also \cite{Cade:2021jhc} and references therein for more details).

\paragraph{Organization of the paper.} In Section~\ref{sec:Background} we give a review of basic elements of simplicial complexes and homology. In Section~\ref{sec:QMA1 and Quantum k SAT} we give our definition of the class $\QMA_1$ and review Bravyi's clock construction  showing $\QMA_1$-hardness of quantum $k$-$\SAT$. In Section~\ref{sec:QMA1hard} we present the proof of Theorem~\ref{thmQMA1}, the main result in the paper. In Section~\ref{sec:connection to SUSY} we discuss the relation to finding ground states in supersymmetric many-body systems. We conclude with some open questions in Section~\ref{sec:outlook}. Some technical details are collected in the Appendix.

\section{Review of Simplicial Homology}
\label{sec:Background}

We begin by giving a review of basic elements of simplicial homology and the clique/independence complex. Readers familiar with this may skip this section. 
\subsection{Simplicial Homology}

Simplices are the basic building blocks of simplicial complexes. A 0-simplex  is a point, a 1-simplex is a line, a 2-simplex is a triangle with its face included, a 3-simplex a solid tetrahedron,  etc. Consider a set of $(n+1)$ vertices $V=\{x_0,x_1,\cdots,x_n\}$. A $p$-simplex corresponds to a subset of these vertices with $p+1$ elements and is denoted  
\equ{
\sigma^{(p)}=[x_0\cdots x_p]\,.
}
We will consider always an ordered set $V$, which leads to oriented simplices. The orientation of a $p$-simplex is determined by the ordering of its vertices. An oriented 1-simplex $[x_0x_1]$, for instance, is  a directed line segment in the direction $x_0\to x_1$. The orientation of a 2-simplex $[x_0x_1x_2]$ is similarly induced by the ordering of its vertices, and similarly for higher dimensional simplices  (see Fig.~\ref{fig:simplices}). Permuting the vertices leads to an equivalent simplex, possibly up to an overall sign indicating whether the simplex has a positive or negative orientation. Precisely,  
\equ{
[x_{i_0} x_{i_2} \cdots x_{i_p}] = \text{sgn}(P)\,  [x_0x_1\cdots x_p]\,,
}
where $\text{sgn}(P)$ is the parity of the permutation from $\{01\cdots p\}$ to $\{i_0 i_2 \cdots i_p\}$.
\begin{figure}[]
\begin{center}
\includegraphics{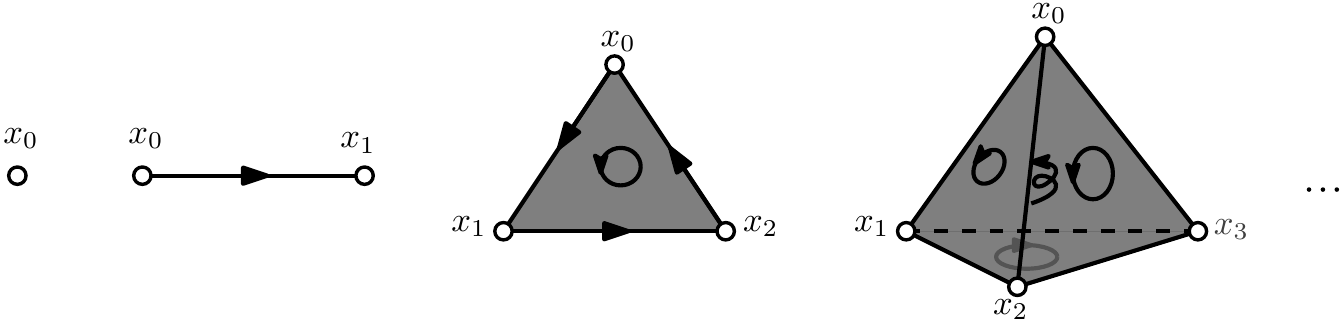}
\caption{Geometric representation of a  0-simplex, oriented 1-simplex, oriented 2-simplex, etc.}
\label{fig:simplices}
\end{center}
\end{figure}
Taking a $(q+1)$-subset of the vertices of a $p$-simplex defines a $q$-simplex called a $q$-face and denoted $\sigma^{(q)}\leq \sigma^{(p)}$. Similarly, one says that $\sigma^{(p)}$ is a coface of $\sigma^{(q)}$.

\noindent An (abstract) simplicial complex $K$ is  as a collection, or set, of simplices satisfying two requirements: \\

\begin{enumerate}
\item If a simplex is in $K$, then all of its faces are also in $K$ and; \label{cond1}
\item The intersection of any two simplices in $K$ is a face of each of them.\\
\end{enumerate}

\noindent  The first condition states that a simplicial complex must be a set which is closed under taking subsets. For instance, a simplicial complex corresponding to a filled in triangle is given by the set $K=\{[x_0],[x_1],[x_2],[x_0x_1],[x_1x_2],[x_2x_0],[x_0,x_1,x_2]\}$. The second condition indicates how simplices can be ``glued'' to each other implying, in particular, that the intersecting faces have the same dimension.\footnote{We emphasize that two simplices must intersect at a {\it single} face of each simplex. For instance, two vertices connected by an edge, $\{[x_0],[x_1],[x_0x_1]\}$, is a simplicial complex but two points connected by two edges $\{[x_0],[x_1],[x_0x_1],[x_0x_1]\}$ is not because the two 1-simplices intersect at two faces. }

A ``maximal simplex'' (or facet) of $K$ is a simplex that is not a face of another simplex in $K$. Note that due to the condition \ref{cond1}) above, it is sufficient to provide the list of all maximal simplices to fully determine the simplicial complex as taking the power set (the set of all subsets) of all maximal simplices generates the entire set of simplices in $K$. For the previous example $K=\langle[x_0x_1x_2] \rangle$, where $\langle\cdot\rangle$ denotes taking the power set. Without additional structure, providing the set of maximal faces is the most efficient way to communicate a simplicial complex.

Let  $f_p$ denote the number of $p$-simplices in $K$.  It is convenient to collect all these by defining the $f$-vector of a simplicial complex:
\equ{
f:=(f_0,f_1,\ldots, f_{\text{dim}(K)})\,,
}
where $\text{dim}(K)$ is the dimension of $K$,  defined to be the highest dimension of its simplices (note that $\text{dim}(K)\leq n$ if there are $n+1$ vertices in $K$).

We have described what is known as an {\it abstract} simplicial complex, defined combinatorially as a collection of sets without any reference to geometry or topology. However, given an abstract simplicial complex one can always associate a {\it geometric} simplicial complex by thinking of $ [x_0x_1\cdots x_p]$ as representing the convex hull of $(p+1)$ points in $\Bbb R^{d}$ with $d\geq \text{dim}(K)+1$. See \cite{Ossona:1999} for more details and formal definitions. It thus makes sense to ask whether a simplicial complex has a $k$-dimensional hole. Intuitively, a hole in a space is a closed cycle which is non-contractible. Homology is what allows one to define this formally.

\paragraph{The boundary operator.}

The main insight in homology theory, and a precursor to many ideas in modern mathematics, is that one can attach an Abelian group structure to a simplicial complex. Given a simplicial complex $K$ let $K_p$  denote the set of all $p$-simplices, $K_p:=\{\sigma^{(p)}_{0},\sigma^{(p)}_{1}, \cdots, \sigma^{(p)}_{f_p} \}$. These are taken to be the generators of a collection of Abelian groups. The  ``$p$-chain groups'' are then
\equ{
C_p(K) := (K_p,+)\,,
}
where $+$ denotes the formal ``addition'' operation in the group.  A general element of $C_p(K)$ is then given by 
\equ{\label{pchainF}
c^{(p)}= \alpha_0\, \sigma^{(p)}_{0}+\cdots+\alpha_{f_p} \,  \sigma^{(p)}_{f_p}\,,\qquad \alpha_i\in \Bbb F\,,
}
where the coefficients $\alpha_i$ take value in a field $\Bbb F$ called the ``coefficient field.'' The inverse element of $\sigma^{(p)}$, denoted formally by $-\sigma^{(p)}$, is obtained by reversing the orientation of $\sigma^{(p)}$ so that $\sigma^{(p)}+(-\sigma^{(p)})=0$.

One now introduces the ``boundary operators'' $\partial_p$, which formalizes the notion of boundary. Given a $p$-simplex $\sigma^{(p)}=[x_{0}\cdots x_{p}]$, the action of the boundary operator is, by definition,
\equ{\label{defdel}
\partial_p[x_{0}\cdots x_{p}]= \sum_{i=1}^{p}(-1)^{i}[x_{0}\cdots \hat x_{i}\cdots x_{p}]\,,
}
where the notation $[x_{0}\cdots \hat x_{i}\cdots x_{p}]$ means that the vertex $x_{i}$ is deleted and the sum is to be understood in the sense above. That is, the boundary map acting on $\sigma^{(p)}$ gives a formal sum of all the boundary components of $\sigma^{(p)}$; the signs appearing on the RHS of \eqref{defdel} are such that the orientation of the boundary simplices are consistent with those of $\sigma^{(p)}$.\footnote{For instance, for the solid triangle, $\partial[x_0x_1x_2]= [x_1x_2]-[x_0x_2]+[x_0x_1]=[x_0x_1]+[x_1x_2]+[x_2x_0]$, consistent with the orientations of each boundary edge--see Fig.~\ref{fig:simplices}.} It is easy to see that 
\equ{\label{dpdp0}
\partial_{p}\cdot \partial_{p+1}=0\,,
}
reflecting the intuitive notion that the ``boundary of a boundary'' vanishes.

One can extend the action of $\partial_p$ by linearity to a general $p$-chain \eqref{pchainF}. Thus, the operators $\partial_p$ are a collection of homomorphisms between chain groups:
\equ{
\partial_p: C_p(K)\to C_{p-1}(K)\,.
}
Given this (finite) collection of groups $C_p$ and the boundary maps, it is convenient to introduce the notion of a ``chain complex,'' which is a sequence of Abelian groups connected by the action of the boundary operator:\footnote{The empty group at the left of $C_n$ is included to emphasize the chain is finite. }
\equ{\label{complexN}
C(K): \qquad   0\rightarrow  C_{n}(K)\xrightarrow{\partial_{n}}C_{n-1}(K)\xrightarrow{\partial_{n-1}}  \cdots \xrightarrow{\partial_{2}}  C_{1}(K) \xrightarrow{\partial_{1}}  C_{0}(K)\xrightarrow{\partial_{0}}0\,.
}
Note that since the $C_p(K)$ come equipped with a basis $\sigma^{(p)}$, the operators $\partial_p$ can be expressed as matrices. In particular, one can define the ``coboundary'' operators 
\equ{\partial^\dagger_p:  C_p\to C_{p+1}
}
by hermitian conjugation of the matrix associated to  $\partial_p$. These can be used to define the a ``co-chain'' complex in which the direction of the arrows above is reversed.

\paragraph{Cycles, boundaries, and homology}

Two important notions are those of ``cycles'' and ``boundaries.'' Consider the set of $p$-chains which are in the kernel of the boundary operator $\partial_p$. These are known as $p$-cycles and denoted by $Z_p$:
\equ{
Z_p := \text{ker}(\partial_p) = \{c^{(p)}\in C_p\;|\; \partial_p c^{(p)}=0\}\,.
}
Intuitively, these correspond to closed cycles in the simplicial complex, as they have no boundary. The space of $p$-cycles in fact forms a subgroup of $C_p$, as can be easily checked.

Consider now the set of $p$-chains which are in the image of the boundary operator $\partial_{p+1}$. These are known as $p$-boundaries and denoted by $B_p$:
\equ{
B_p:=\text{im}(\partial_{p+1})=\{c^{(p)}\;|\; c^{(p)} =\partial c^{(p+1)}\}\,.
}
Intuitively, these correspond to chains which are boundaries of a higher-dimensional chain. These also form a subgroup of $C_p$, as can be easily checked. In fact, 
$B_p$ is a subgroup of $Z_p$,
\equ{
B_p\subseteq Z_p \,,
}
which follows directly from \eqref{dpdp0}. Again, this is simply the statement that boundaries do not themselves have a boundary.

With these elements one can now define the notion of homology. The $p$-th homology group is defined as the coset groups
\equ{
H_{p}(K,\Bbb F):=\frac{Z_p}{B_p}=\frac{\text{ker}(\partial_{p})}{\text{im} (\partial_{p+1})}\,.
}
That is, $H_{p}(K,\Bbb F)$ denotes the set of all $p$-cycles which are not $p$-boundaries. This  captures the intuitive notion of a $p$-dimensional hole in $K$. The number of independent $p$-dimensional holes is captured by the Betti numbers:
\equ{
\beta_p =\text{rank} \, H_{p}(K,\Bbb F)\,.
}
That is, $H_{p}(K,\Bbb F)=(\Bbb F)^{\beta_p}$ and the Betti numbers  specify the corresponding homology groups.\footnote{More generally, the full homology group takes the form $H_\ast(K,\Bbb F) = (\Bbb F)^{\beta_0}\oplus (\Bbb F)^{\beta_1}\oplus \cdots \oplus (\Bbb F)^{\beta_n} \oplus T(K,\Bbb F)$ where $T(K,\Bbb F)$ are known as the torsion coefficients. The Betti numbers capture the freely acting part of the homology group while the torsion coefficients capture a possible non-freely acting part. These depend on the field coefficient $\Bbb F$ and the torsion vanishes for $\Bbb F=\Bbb Q,\Bbb R, \Bbb C$. Since we will be mostly working with $\Bbb F=\Bbb C$ we ignore torsion and the homology is completely specified by the Betti numbers.  }

Recall that coset groups correspond to equivalence classes. The meaning of the coset above is that any two elements $c^{(p)}$ and $c^{(p)'}$ in $Z_p$ which differ by a boundary are taken to be equivalent or ``homologous'':
\equ{\label{equivhom}
c^{(p)}\sim c^{(p)'} \qquad \qquad \text{if} \qquad \qquad c^{(p)'}=c^{(p)} +\partial c^{(p+1)}\,.
}
The elements $c^{(p)}$ and $c^{(p)}$ are thus equivalent in homology and either one is a valid representative of the homology class.

Cycles which are boundaries 
are trivial elements of homology:
\equ{
c^{(p)}=\partial c^{(p+1)} \sim 0\,.
}
The most important property of homology groups is that they are topological invariants of $K$;  if two spaces $K$ and $M$ are homeomorphic, then $H_p(K,\Bbb F)=H_p(M,\Bbb F)$. In fact, homology groups are not only homeomorphic invariants but homotopic invariants \cite{hatcher2002algebraic}.

Note that determining whether a space has a $k$-dimensional is then a question in linear algebra, studying the kernel and image of the linear operator $\partial$. This can be answered by bringing the boundary operator into a Smith normal form and thus homology is algorithmically computable but can be extremely slow for large matrices.

\ 

A useful quantity is the Euler characteristic of $K$, given by the alternating sum of Betti numbers: 
\equ{
\chi(K)= \sum_p (-1)^p \beta_p\,.
}
Since the Betti numbers are topological invariants, so is the Euler characteristic. The Euler-Poincar\'e formula states that it can can also be written as the alternating sum of the number of simplices at each dimension:
\equ{
\chi(K):= \sum_p (-1)^p f_p\,.
}
Since no knowledge of the homology groups is required in this expression, this sometimes makes the Euler quantity an easier, albeit less refined, topological invariant  to compute.

\paragraph{Reduced homology.}

It is sometimes useful to consider a slight modification  called ``reduced homology.'' This is defined by taking the augmented chain complex
\equ{
0\rightarrow  C_{n}(K)\xrightarrow{\partial_{n}}C_{n-1}(K)\xrightarrow{\partial_{n-1}}  \cdots \xrightarrow{\partial_{2}}  C_{1}(K) \xrightarrow{\partial_{1}}  C_{0}(K)\xrightarrow{\epsilon} \Bbb F \rightarrow 0\,,
}
where $\epsilon(\sum \alpha_i \sigma_{0,i}):=\sum_i\alpha_i$. The corresponding homology groups $\tilde H_p(K,\Bbb F)$ are defined analogously and one has
\equ{
H_0(K,\Bbb F) = \tilde H_0(K,\Bbb F)\oplus \Bbb F\,,\qquad \qquad H_p(K,\Bbb F) = \tilde H_p(K,\Bbb F)\,,\qquad p\geq 1\,.
}
The corresponding Euler characteristic  $\tilde \chi$ is defined analogously. It is easy to see that it is related to the standard Euler characteristic by 
\equ{
\tilde \chi(K) = -\chi(K)+1\,.
}
As we discuss in Section~\ref{sec:connection to SUSY} reduced homology groups capture the zero energy states of supersymmetric many-body systems and the reduced Euler characteristic corresponds to the so-called Witten index. 

\ 

\noindent We will focus on homology with  coefficients in $\Bbb F=\Bbb C$ (and thus $T(K,\Bbb F)=0$)  and to simplify notation we write
\equ{
H_\ast(K):=H_\ast(K,\Bbb C)\,, \qquad \qquad \tilde H_\ast(K):=\tilde H_\ast(K,\Bbb C)\,.
}

\subsection{The Clique/Independence Complex} \label{subsec:cli/in}

Given a simple graph $G=(V,E)$, the ``clique complex'' $Cl(G)$ is a simplicial complex whose $k$-simplices are given by  $(k+1)$-cliques of $G$ (i.e., complete subgraphs of $G$ with $k+1$ vertices). Since any subgraph  of a clique is a clique, the set $Cl(G)$ is closed under taking subsets and satisfies the first requirement of a simplicial complex. It is also easy to see it satisfies the second requirement. The action of the boundary operator is that in \eqref{defdel}, sending $k$-cliques to $(k-1)$-cliques. Finding $k$-dimensional holes in clique complexes is an important task in TDA, with many applications including neuroscience  \cite{petri2014homological,giusti2016two,reimann:hal-01706964} and the study of protein networks.

Another simplicial complex which is defined naturally given a graph $G$ is the ``independence complex'' $I(G)$. The $k$-simplices of $I(G)$ are given by    $(k+1)$-independent sets of $G$ (i.e., a subset of $k+1$ disjoint vertices in $G$). Clearly, $k$-independent sets of $G$ are equivalent to $k$-cliques of the complement graph $\bar G$ and therefore
\equ{
 I(G) =Cl(\bar G)\,.
}
Since the graph complement can be computed in time polynomial in $\abs{V}$, the computational complexity of computing homology for the clique complex is equivalent to that of the independence complex. Throughout most of Section~\ref{sec:QMA1hard} we will focus on the independence complex, and in the final theorem we will translate our result to discuss the clique complex.

\section{$\QMA_1$ and Quantum $k$-$\SAT$}
\label{sec:QMA1 and Quantum k SAT}

In this section, we review the complexity class $\QMA_1$ and quantum $k$-$\SAT$, the canonical $\QMA_1$-complete problem \cite{bravyi2011efficient}. As we discuss below, the definition of this class depends on the choice of a universal gate set $\mathcal G$ and we introduce a choice which will be used for our hardness proof  in \cref{sec:QMA1hard}. 

\subsection{The complexity class $\QMA_1$}

The quantum complexity class $\QMA_1$ is defined as the the perfect completeness version of $\QMA$. Precisely, 
\begin{definition}[$\QMA_1$ \cite{bravyi2011efficient,Gosset_2013}] \label{defQMA1}
A promise problem $A = (A_{\text{yes}}, A_{\text{no}})$ is in $\QMA_1$ if there exists a polynomial-time quantum circuit (the ``verifier'') $V_x$ composed of gates from the gate set $\mathcal{G} $ for any $n$-bit input string $x$ such that
\begin{itemize}
	\item If $x \in A_{\text{yes}}$, there exists a $\poly(n)$-qubit witness state $\ket{w}$ such that $\Pr[V_x(\ket{w}) = 1] = 1$,
	\item If $x \in A_{\text{no}}$, then for any $\poly(n)$-qubit witness state $\ket{w}$, $\Pr[V_x(\ket{w}) = 1] \leq \frac{1}{3}$.
\end{itemize}
\end{definition}
There are two main differences with respect to the more standard class $\QMA$. The first is that in the case of $\yes$ instances the verifier is required to accept an appropriate witness with probability 1, while in $\QMA$ this is only required to be great than $1/2$ (and often taken to be 2/3). This is why this class is often referred to as perfect completeness or one-sided error version of its more common sibling $\QMA$. Importantly, also note that the definition  explicitly depends on the choice of a universal gate set $\mathcal G$. In the case of $\QMA$ the definition is independent of any choice of gate set, since all universal gate sets can approximate any unitary evolution. But for $\QMA_1$ it is no longer necessarily enough to be able to approximate a given unitary. 
And for this reason when specifying the class $\QMA_1$ it is necessary to define it relative to some particular choice of universal gate set. 

A frequently used choice is $\mathcal{G}'=\{\text{CNOT}, \hat{H}, T\}$ where
\equ{
\text{CNOT}= \begin{pmatrix}
1& 0 & 0 &0\\
0& 1 & 0 &0 \\
0& 0 & 0 &1 \\ 
0& 0 & 1 &0\\
\end{pmatrix}\,, \  \qquad \hat{H}=\frac{1}{\sqrt 2}\begin{pmatrix}
1& 1\\
1&-1
\end{pmatrix}\,, \ \qquad T = \begin{pmatrix}
1 & 0 \\
0 & e^{\frac{i\pi}{4}}
\end{pmatrix}
} 
However, for reasons to be discussed below we wish to have gates with  {\it rational} coefficients. As shown in \cite[Theorem 3.3]{Adleman:1997} (it also follows from \cite[Theorem 1.2]{shi2003both}), the gate $\text{CNOT}$ together with   a one-qubit real unitary $U$ is universal provided $U^{2}$ is basis-changing. 

Our definition of $\QMA_1$ is then based on the choice of universal gate set 
\equ{\label{unigateset}
\mathcal G=\{\text{CNOT},U_{\mathit Pyth.}\}\qquad \qquad U_{\mathit Pyth.}=\frac{1}{5}\begin{pmatrix}
3& 4\\
-4&3
\end{pmatrix}\,,
}
for which all gates have rational coefficients. We refer to $U_{\mathit Pyth.}$ as the ``Pythagorean'' gate, since any Pythagorean triple defines such a  matrix. It is relative to this definition of $\QMA_1$ that we will prove hardness of the clique homology problem.

\subsection{Review of quantum $k$-$\SAT$ and the clock Hamiltonian  }

The canonical $\QMA_1$-complete problem is Quantum $k$-$\SAT$:\\

\mypromprob{\label{defQSATp}{\sc Quantum $k$-$\SAT$} \cite{bravyi2011efficient}}{A set of $k$-local projectors $\Pi_S \in \mathcal{P}$ where $S$ are possible subsets of $\{1,\ldots,n\}$ of cardinality $k$, and $\mathcal{P}$ are a set of projectors obeying certain constraints.  }{Either there exists a state $\ket{\psi}$ such that $\Pi_S\ket{\psi}=0$ for all $S$ or otherwise $\sum_{S}\bra{\psi}\Pi_{S}\ket{\psi}\geq \epsilon$ for all $\ket{\psi}$, with $\epsilon>\frac{1}{\poly(n)}$. }{Output $\yes$ if the former and $\no$ if the latter.\\}

\noindent It was shown in \cite{bravyi2011efficient} that quantum $k$-$\SAT$ is $\QMA_1$-complete for $k\geq 4$ and that for $k=2$ there is an efficient classical algorithm to solve the problem. For  $k=3$ the problem was shown to be $\QMA_1$-complete in \cite{Gosset_2013}. We will construct a mapping from the quantum $4$-$\SAT$ circuit-to-Hamiltonian construction to the $k$-local homology problem. The general $\QMA_1$ circuit is encoded into a Hamiltonian of the form \cite{bravyi2011efficient}  
\equ{
H_{\text{Bravyi}}=\Hin  +\sum_{i=1}^{6}\Hclock^{(i)} +\sum_{t=1}^L  \left(\Hpropt + \Hpropt' \right)+\Hout \,,
}
The first terms in $H_{\text{Bravyi}}$ encode the history of some quantum verification circuit, while the final term, $\Hout$ penalises computations which don't accept with probability 1. 
In $\yes$ instances there exist "history states" which can be shown to have exactly zero energy with respect to $H_{\text{Bravyi}}$, and therefore act as accepting witnesses.
While it can be shown that in $\no$ instances the minimum energy of $H_{\text{Bravyi}}$ is bounded from below by an inverse polynomial in the problem size. 

A detailed overview of the construction is given in  \cref{sec:app-projectors}. We have summarized the rank-1 projectors which are required in order to construct $H_{\text{Bravyi}}$ for the universal gate set \eqref{unigateset} in \cref{table:4SAT}. In Section~\ref{sec:QMA1hard} we construct a reduction to the clique homology problem by implementing each of these projectors. 
\begin{table} 
\begin{center}
\begin{tabular}{ |c|c|c| } 
 \hline
 \emph{Term in} $H_{\text{Bravyi}}$ & \emph{Penalizes state} $\ket{\psi_S}$ \\ 
 \hline
$\Hpropt' $& $ \frac{1}{\sqrt{2}} \ket{10}\left(\ket{11}-\ket{00} \right) $     \\ 
 \hline
  $\Hpropt(\text{CNOT})$ & $\frac{1}{\sqrt{2}} \ket{01} \left(\ket{10}-\ket{01} \right)$  \\
 \hline
   $\Hpropt(\text{CNOT})$ & $\frac{1}{\sqrt{2}} \ket{00} \left(\ket{10}-\ket{01} \right)$  \\
  
\hline
$\Hpropt(U_{\mathit{Pyth.}})$ &$ \frac{1}{5\sqrt{2}} \left(-5\ket{011}+4\ket{100}+3\ket{101}\right)$   \\ 
 \hline
 $\Hpropt(U_{\mathit{Pyth.}})$ &$ \frac{1}{5\sqrt{2}}  \left(-5\ket{010}+3\ket{100}-4\ket{101}\right)$   \\ 
  \hline
   $\Hpropt(\text{CNOT})$ & $\frac{1}{\sqrt{2}} \ket{1} \left(\ket{101}-\ket{010} \right)$ \\
 \hline
 $\Hpropt(\text{CNOT})$ & $\frac{1}{\sqrt{2}} \ket{1}\left(\ket{011}-\ket{100} \right) $  \\
 \hline
$\Hclock^{(1)}$ & $\ket{00}$    \\ 
 \hline
 $\Hclock^{(2)}$ & $\ket{11} $  \\ 
 \hline
 $\Hin$, $\Hout$ &$ \ket{011}  $  \\ 
 \hline
$ \Hclock^{(6)}$, $\Hclock^{(4)}$, $\Hclock^{(5)}$, $\Hclock^{(3)}$ &$ \ket{1100}$ \\
\hline
$ \Hclock^{(4)}$ &$ \ket{0111}$  \\
\hline
$ \Hclock^{(5)}$ &$ \ket{0001}$  \\
\hline
\end{tabular}
\caption{Projectors needed for quantum $4$-$\SAT$ with universal gate set $\mathcal{G}$. Note we collated projectors which are the same up to re-ordering the qubits involved.}
\label{table:4SAT}
\end{center}
\end{table}

\paragraph{Comment on completeness.} We have been deliberately  vague about the definition of the set of allowable projectors $\mathcal{P}$ in the definition of quantum $k$-$\SAT$. There is no need to put any constraints on the set of projectors $\mathcal{P}$ for showing that the problem is $\QMA_1$-hard. However, for showing {\it containment} in $\QMA_1$ care has to be taken, and the choice of $\mathcal{P}$ which ensures containment will depend on the universal gate set chosen in the definition of $\QMA_1$. It is known that quantum $k$-$\SAT$ is contained in $\QMA_1$ as long as there exists an efficient algorithm which can be used to measure the eigenvalue of any $\Pi \in \mathcal{P}$ for arbitrary state $\ket{\psi}$ using the gate set $\mathcal{G}$ \cite{bravyi2011efficient,Gosset_2013}.
And clearly to show $\QMA_1$ completeness it is necessary to demonstrate that the set of projectors used in mapping from a $\QMA_1$ verification circuit to quantum $k$-$\SAT$ is contained in $\mathcal{P}$. Since throughout this work we are concerned with demonstrating $\QMA_1$-hardness, not containment or completeness,  we will defer any further discussion of the choice of $\mathcal{P}$ and allow it to contain arbitrary $k$-local projectors.

\section{Clique Homology is $\QMA_{1}$-hard}
\label{sec:QMA1hard}

In this Section, we prove our main result Theorem~\ref{thmQMA1}.  For purposes of presentation we will focus on the independence complex $I(G)$, rather than the clique complex. As noted in \cref{subsec:cli/in} these are related by taking the graph complement, which can be done efficiently in the size of the graph, and thus proving $\QMA_1$-hardness for the independence complex or the clique complex is equivalent. In \cref{sec:reduction} we will explicitly consider the mapping between the two problems. 

As outlined in the Introduction, the basic idea is to encode the computational history of a $\QMA_1$ verification circuit into the topology of an independence complex. We achieve this via the "circuit-to-Hamiltonian" construction reviewed in \cref{sec:QMA1 and Quantum k SAT}. The task at hand is then to map the data in  \cref{table:4SAT} to a graph $G$ such that the corresponding independence complex has a nontrivial homology if and only if the Hamiltonian $H_{\text{Bravyi}}$ has a zero energy groundstate.\footnote{In fact, we will require the reduction to be parsimonious so there is a 1-to-1 correspondence between the number of accepting witnesses to the $\QMA_1$ protocol and the rank of the homology group. This will allow us to show the counting version of the problem is $\#\BQP$-hard.}

\subsection{Basic gadgets}
\label{sec:gadgets}

\begin{figure}[]
\begin{center}
\includegraphics{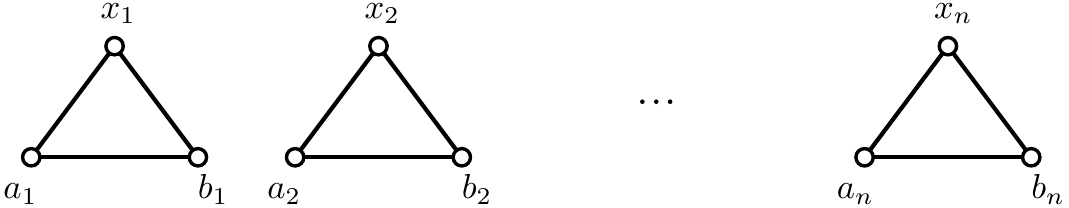}
\caption{A graph $G_{n}$ representing a system of $n$ disconnected triangles. }
\label{fig:Gn}
\end{center}
\end{figure}
Consider a graph $G_{n}$ consisting of $n$ disconnected triangles (see \cref{fig:Gn}) and denote the associated independence complex by $\Sigma_{n}:=I(G_{n})$. The total number of independent sets is $4^n$. The largest independent sets are of size $n$ and there are $3^n$ of them, corresponding to choosing one vertex from each triangle. We claim that the  homology groups  for all $n\geq 2$ are given by
\equ{\label{HSn}
  H_{0}(\Sigma_{n})\cong\Bbb C\,,\qquad   H_{n-1}(\Sigma_{n})\cong(\Bbb C^{2})^{\otimes n}\,,\qquad  H_{i\neq \{ 0, n-1\}}(\Sigma_{n})\cong0\,.
} 
Thus,  the $(n-1)^{\text{th}}$  homology group of $\Sigma_n$  is isomorphic to the Hilbert space of $n$ qubits. It is inside this vector space that we will encode the satisfying instances of quantum $k$-$\SAT$.

Consider as a simple example the case of two triangles, $n=2$. There are a total of 6 independent sets of size 1 and 9  independent set of size 2, corresponding to the vertices and edges of $\Sigma_2$, respectively (see \cref{fig:I(G2)}).
\begin{figure}[]
\begin{center}
\includegraphics{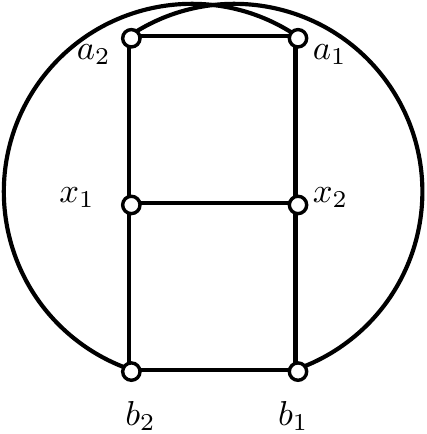}
\caption{The independence complex $I(G_2)$ associated to two disjoint triangles. }
\label{fig:I(G2)}
\end{center}
\end{figure}
Clearly, $\Sigma_2$ has one connected component, and thus $\beta_0=1$, and four independent 1-dimensional holes (1-dimensional non-contractible cycles), and thus $\beta_1=4$, while all other Betti numbers vanish, consistent with \eqref{HSn}. The general result for arbitrary $n$ is easy to show.\footnote{One way to see this is  by using the formula for reduced homology $\tilde H_{k}(K\cup L)=\sum_{p+q+1=k}H_{p}(K)\otimes H_{q}(L)$. It is straightforward to see that $\tilde H_0(\Sigma_1)=(\Bbb C)^2$ and $\tilde H_{i}(\Sigma_1)=0$ for all $i\geq 1$. The result above then follows from  $\Sigma_n=\Sigma_1\cup\cdots \cup \Sigma_1$ and using this expression recursively. }

To make contact with standard qubit notation, we label the basis for $H_{n-1}(\Sigma_{n})$ by the set of all $n$-bit strings. We note that such a canonical basis is obtained as follows. Let 
\begin{equation} \label{eq:basis}
\ket{0}_{i}:= [x_{i}]-[a_{i}]\,,\qquad \qquad  \ket{1}_{i}:= [x_{i}]-[b_{i}]\,,\qquad i\in \{1\cdots n\}\,.
\end{equation}
where $x_i,a_i,b_i$ label the vertices of each triangle (see \cref{fig:I(G2)}).
Note that each of these 0-chains is annihilated by the boundary operator,
\equ{
\partial \ket{0}_{i}=0\,,\qquad \qquad  \partial \ket{1}_{i}=0\,, \qquad \qquad i\in \{1\cdots n\}\,.
}
Let us introduce some  notation. Consider two simplices $\sigma^{(p)}=[x_0\cdots x_p]$ and  $\sigma^{(q)}=[y_0\cdots y_q]$. We define the wedge product $\wedge $ as 
\equ{
\sigma^{(p)}\wedge \sigma^{(q)}= \begin{cases}[x_0\cdots x_p y_0\cdots y_q]  &  x_i\neq y_j \\ 0 &\text{otherwise} \end{cases}
}
Then, a basis $B=\{\ket{e_{i}}, i=1\cdots 2^{n}\}$ for $H_{n-1}(\Sigma_{n})$ is obtained by  
\eqss{\label{compbasisn}
\ket{0\cdots 0}:=\,& ( [x_{1}]-[a_{1}])\wedge \cdots \wedge ( [x_{n}]-[a_{n}])\,,\\
\ket{0\cdots 1}:=\,& ( [x_{1}]-[a_{1}])\wedge \cdots \wedge ( [x_{n}]-[b_{n}])\,,\\
&\vdots   \\
\ket{1\cdots 1}:=\,& ( [x_{1}]-[b_{1}])\wedge \cdots \wedge ( [x_{n}]-[b_{n}])\,.
}
Indeed, each of the basis states $\ket{e_{i}}$ is an $(n-1)$-dimensional hole in $\Sigma_{n}$ since each $\partial \ket{e_{i}}=0$ and $\ket{e_{i}}\neq \partial \ket{f_{i}}$ since there are no $(n+1)$-independent sets $\ket{f_{i}}$ in $G_{n}$. Expanding the products, and $[ij]=-[ji]$ and  $[ij]=[i]\wedge [j]$ one obtains the $2^{n}$ linear combinations spanning  $H_{n-1}(\Sigma_{n})\cong(\Bbb C^{2})^{\otimes n}$.

\subsubsection{2-qubit projectors} \label{sec:2qubProj}

 To illustrate this, consider the case $n=2$. The corresponding simplicial complex $\Sigma_2$ consists of 6 vertices corresponding to the $6$ 1-independent sets and nine 1-simplices corresponding to the $9$ 2-independent sets of $G_2$ (see \cref{I2}). Expanding the terms in \eqref{compbasisn}, the computational basis for $H_{1}(\Sigma_{2})=\Bbb C^{4}$ is
\eqss{ \label{2qubitbasis}
\ket{00}=\,&[x_1\, x_2]+[x_2\, a_1]+[a_1\,a_2 ]+[a_2 \, x_1]\,,\\
\ket{01}=\,&  [x_1\, x_2]+[x_2\, a_1]+[a_1 \, b_2]+[b_2\, x_1] \,,\\
\ket{10}=\,&[x_1\, x_2]+[x_2\, b_1]+[b_1\,a_2]+[a_2 \, x_1]\,, \\
\ket{11}=\,&[x_1\, x_2]+[x_2\, b_1]+[b_1\,b_2]+[b_2\, x_1]\,.
}
These correspond to the four independent non-contractible 1-cycles in $\Sigma_{2}$ depicted in  \cref{I2}.  
\begin{figure}[]
\begin{center}
\includegraphics[scale=0.8]{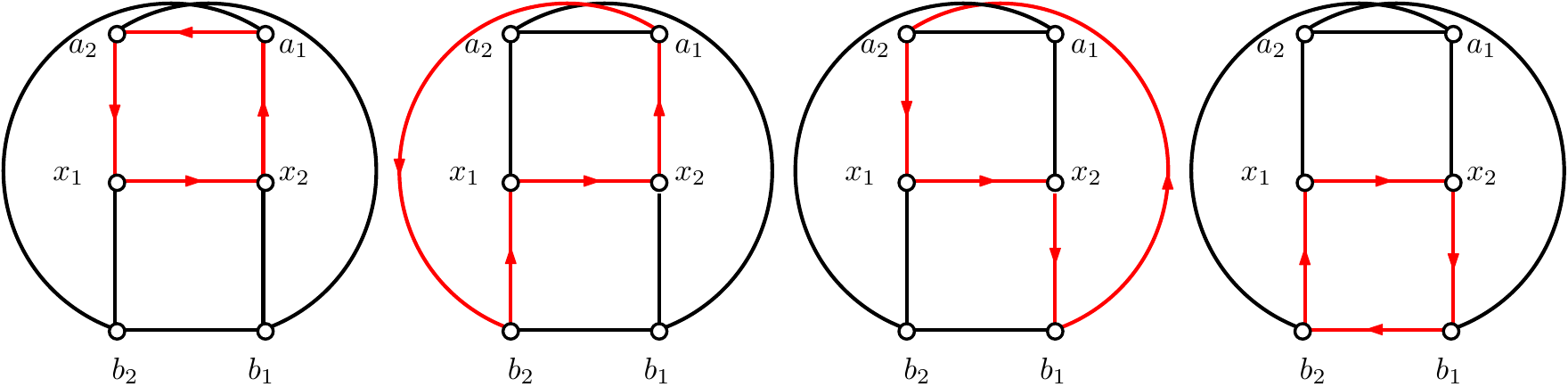}
\caption{The independence complex $\Sigma_{2}=I(G_{2})$ associated to $n=2$ non-interacting qubits. The simplicial complex has 4 independent 1-cycles, corresponding go $ H_{1}(\Sigma_{2})\cong\Bbb C^{4}$. The four basis states 
$\left\{\ket{00},\ket{01},\ket{10},\ket{11}\right\}$ are indicated in red.}
\label{I2}
\end{center}
\end{figure}

We now follow the general strategy outlined above. We begin with a simple example. Consider the rank-1 classical projector
\equ{
\Pi=\ket{00}\bra{00}\,.
}
To eliminate the 1-dimensional hole corresponding to the state $\ket{00}$, we identify the cycle corresponding to this state in $\Sigma_{2}$, shown in  the first diagram in \cref{I2}. In order to fill this face we add a new ``mediator'' vertex $m$ and the four 2-simplices $[a_{1}a_{2}m]$, $[a_{2}x_{1}m]$, $[x_{1}x_{2}m]$, and $[x_{2}a_{1}m]$. This results in filling in the 1-chain $\ket{00}$, which is thus no longer a hole in $\Sigma_{2}'$ (see \cref{fig:I2filled}). Indeed, we can verify this algebraically:
\equ{
\ket{00}= \partial \ket{\Psi}\,,
}
with
\equ{
\ket{\Psi}= [a_{1}a_{2}m]+[a_{2}x_{1}m]+[x_{1}x_{2}m]+[x_{2}a_{1}m] 
}
where $\partial=\partial_{a,b,x}+\partial_m$ is the boundary operator on $\Sigma_{2}'$. Finally, from  $\Sigma_{2}'$ we construct the graph $G_{2}'$ such that $\Sigma_{2}'=I(G_{2}')$. It is easy to see that this corresponds to adding the mediator vertex $m$ and connecting it to the vertices $b_{1}$ and $b_{2}$.

\begin{figure}[]
\begin{center}
\includegraphics[scale=1]{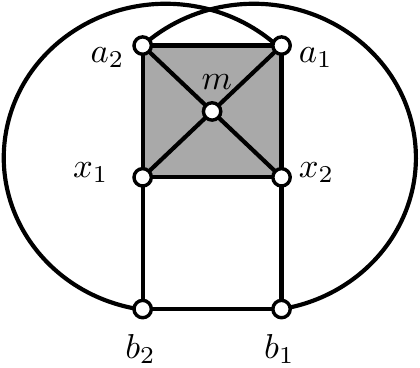}  \qquad \qquad \includegraphics[scale=1]{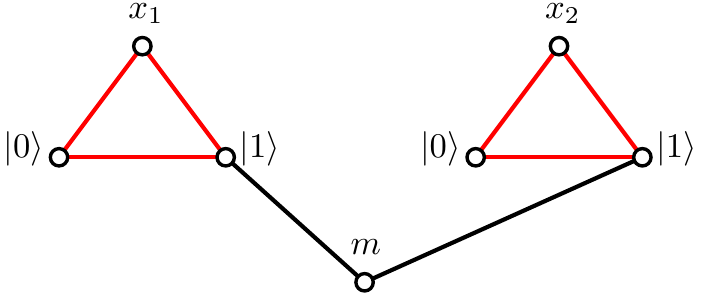}
\caption{Filling in the face $\ket{\Psi}$ bounded by the cycle $\ket{00}$ leads to the simplex $\Sigma_{2}'$, rendering the state $\ket{00}$ a trivial element of homology. This is  accomplished by adding the additional vertex $m$ and the 2-simplices $[a_{1}a_{2}m]$, $[a_{2}x_{1}m]$, $[x_{1}x_{2}m]$, and $[x_{2}a_{1}m]$. In the corresponding graph $G_{2}'$ shown on the right this amounts to adding a mediator vertex and connecting it as depicted.  }
\label{fig:I2filled}
\end{center}
\end{figure}

Projectors of higher rank are easily obtained. Consider for instance the rank-2 classical projector
\equ{
\Pi=\ket{00}\bra{00}+\ket{11}\bra{11}\,.
}
Since this lifts both states $\ket{00}$ and $\ket{11}$ we need to fill the faces of the first and last cycle in \cref{I2}, which is achieved by adding two mediator vertices $m_{1}$ and $m_{2}$ and eight 2-simplices. Explicitly, 
\equ{
\ket{00}=\partial \ket{\Psi_1}\,,\qquad \ket{11}=\partial \ket{\Psi_2}
}
with 
\eqss{
\ket{\psi_1} =\,& [a_{1}a_{2}m_1]+[a_{2}x_{1}m_1]+[x_{1}x_{2}m_1]+[x_{2}a_{1}m_1]\,,  \\
\ket{\psi_2} =\,&  [x_{1}b_{2}m_2]+[b_{2}x_{1}m_2]+[x_{1}x_{2}m_2]+[x_{1}b_{1}m_2]\,.
}
 The resulting graph $\Sigma_{2}'$ is shown in   \cref{table:gadgets2qubit}. Note that the total gadget is given by overlapping the corresponding graph gadgets on the common qubits and connecting the mediators. 
 In \cref{sec:combining_gadgets} we outline a general procedure for adding projectors which act on the same qubits.
 \begin{table}[htp]
\begin{center}
\begin{tabular}{ |>{\centering\arraybackslash}m{5cm}| >{\centering\arraybackslash} m{5.5cm}| >{\centering\arraybackslash} m{5.5cm} |}  \hline 
2-qubit Projector & Gadget $\Sigma_{2}'$ & Gadget $G_{2}'$  \\  \hline  \hline
$\ket{00}\bra{00}$ & \includegraphics[scale=0.8]{Figures/Sigma2p0.pdf}   & \includegraphics[scale=0.8]{Figures/00.pdf}  \\  \hline
$\ket{00}\bra{00}+\ket{11}\bra{11}$ &  \includegraphics[scale=0.8]{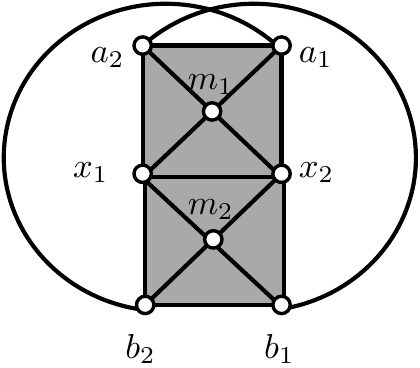}    & \includegraphics[scale=0.8]{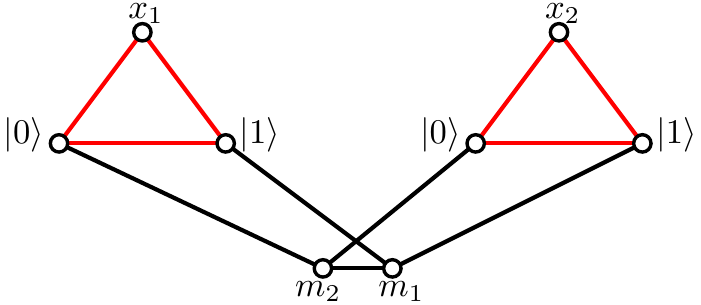}  \\  \hline
$(\ket{00}-\ket{11})(\bra{00}-\bra{11})$ 
& \includegraphics[scale=0.8]{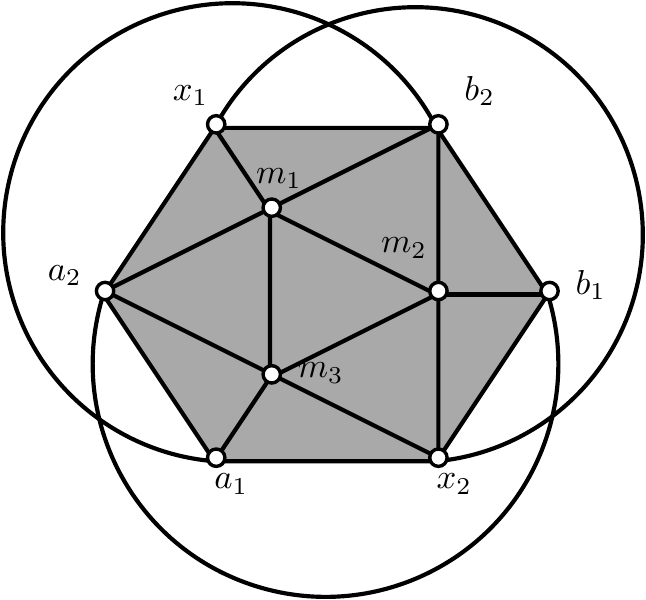} & \includegraphics[scale=0.8]{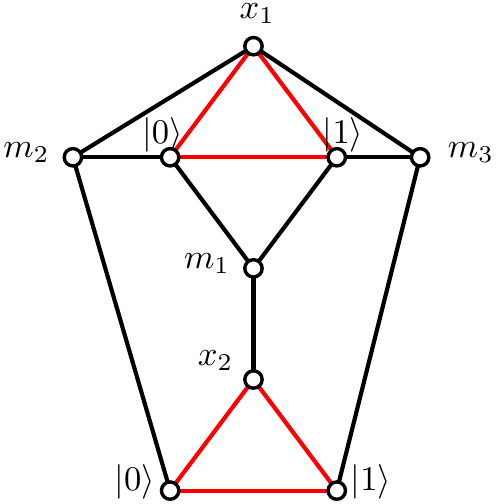}  \\  \hline
$(\ket{01}-\ket{10})(\bra{01}-\bra{10})$ & \includegraphics[scale=0.8]{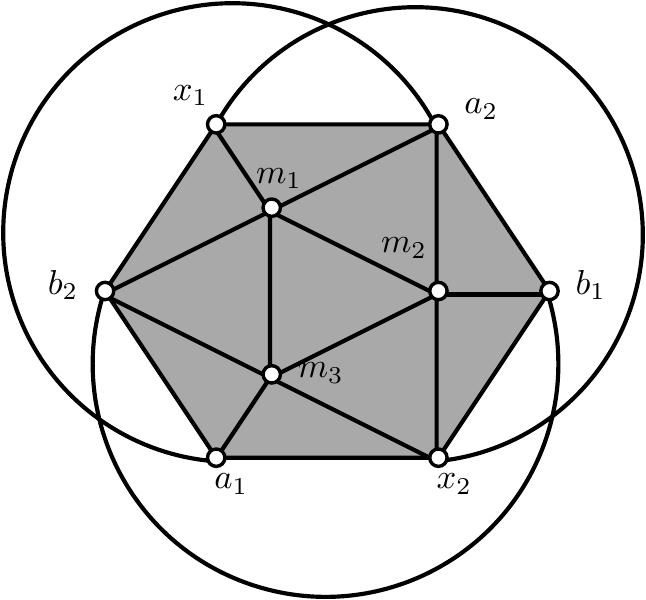}  &   \includegraphics[scale=0.8]{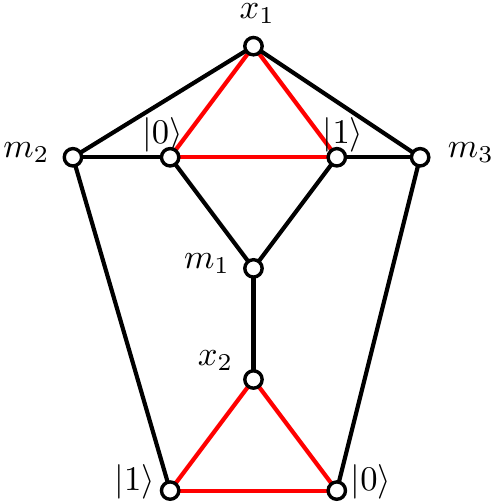} \\  \hline

\end{tabular}
\end{center}
\caption{Collection of gadgets for each 2-qubit projector in the clock Hamiltonian of quantum 4-$\SAT$. The red triangles denote the qubits. }
\label{table:gadgets2qubit}
\end{table}%

Projectors which are not diagonal in the computational basis require a little more work. Consider the projector
\equ{
\Pi=(\ket{00}-\ket{11})(\bra{00}-\bra{11})\,.
}
Note this is a rank-1 projector, lifting only the entangled state $\ket{00}-\ket{11}$, given by
\equ{
\ket{00}-\ket{11}=[x_2\, a_1]+[a_1\,a_2 ]+[a_2 \, x_1]+[x_1\, b_2]+[b_2\,b_1]+[b_1\, x_2]\,,
}
To fill in this (and {\it only} this) cycle by a triangulation of the enclosed face one requires three mediators, which is shown in  \cref{table:gadgets2qubit}. The corresponding graph $G_2'$ is also shown there.\footnote{Note that there is a choice on how to perform this triangulation and the choice we have made breaks the symmetry between the two qubits.} Again, one can algebraically check that 
\equ{
\ket{00}-\ket{11} = \partial \ket{\Psi}\,,
}
where 
\eqss{
\ket{\Psi}=\,&([x_2\, a_1]+[a_1\,a_2 ])\wedge [m_3]+([a_2 \, x_1]+[x_1\, b_2])\wedge[m_1]+([b_2\,b_1]+[b_1\, x_2])\wedge [m_2] \\ 
\,&+ [b_2m_2m_1]+[x_2m_3m_2]+[a_2m_1m_3]  +[m_1m_2m_3]\,.
}
Furthermore, one can check that this is the {\it only} 1-cycle rendered trivial by this triangulation, as desired (see \cref{app:mathematica}).  The projector $\Pi=(\ket{01}-\ket{10})(\bra{01}-\bra{10})$ is generated in the same way, up to relabelling (see \cref{table:gadgets2qubit}).

\subsubsection{3-qubit projectors}

We now consider $n=3$ and thus need to identify 2-dimensional voids in $\Sigma_{3}$ to be filled in. The easiest ones to identify are the classical projectors and we begin again with a simple example to illustrate the method. Consider the classical  3-qubit, rank-1, projector
\equ{
\Pi= \ket{000}\bra{000}\,.
}
The  cycle $\ket{000}$ can be visualized as a double pyramid made out of 8 faces, shown in \cref{table:gadgets3qubit}. Lifting the state $\ket{000}$ then corresponds to filling in this three-dimensional space. This can be achieved by adding a single mediator vertex $m$ connecting to all 6 vertices  and including all the 4-simplices obtained by taking the triangles bounding the double pyramid, and adjoining the mediator to form a tetrahedron. The corresponding graph $G_3'$ is simply three triangles all connected to the single mediator, shown in \cref{table:gadgets3qubit}). Just like in the case of the 2 qubits one can explicitly check that this procedure leads to $\ket{000}=\partial \ket{\Psi}$ thus rendering $\ket{000}$ a trivial element of homology. See \cref{app:algebraic} for the explicit expression; here we will focus on the geometric intuition.  Crucially, one can check that no other 3-qubit state can be written in this form and thus the only lifted state is $\ket{000}$ (see \cref{app:mathematica}). The corresponding graph gadget $G_{3}'$ is what one would expect and could have been guessed directly. However, this is no longer the case when looking for gadgets to uplift entangled states, which in general will contain an unwieldy number of mediator and connections.   
 \begin{table}[]
\begin{center}
\begin{tabular}{ |>{\centering\arraybackslash}m{5cm}| >{\centering\arraybackslash} m{4cm} | >{\centering\arraybackslash} m{8cm}|}  \hline 
3-qubit Projector & Gadget $\Sigma_{3}'$  & Gadget $G_{3}'$  \\  \hline  \hline
$\ket{000}\bra{000}$  & \includegraphics[scale=0.8]{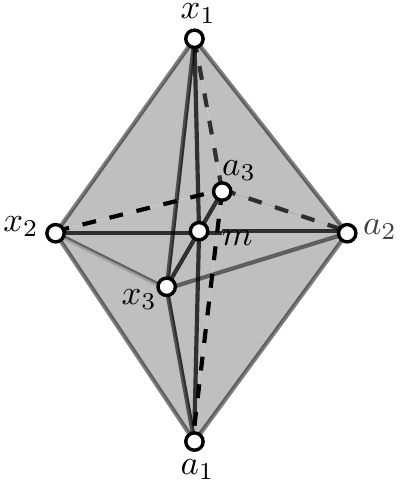}   & \includegraphics[scale=0.8]{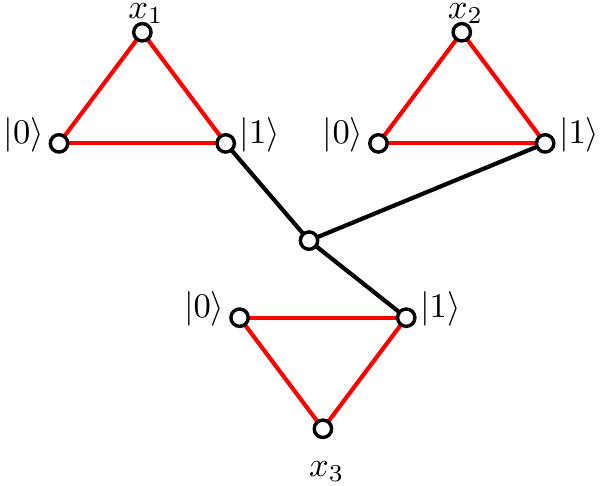}  \\  \hline
$\ket{001}\bra{001}$  & \includegraphics[scale=0.8]{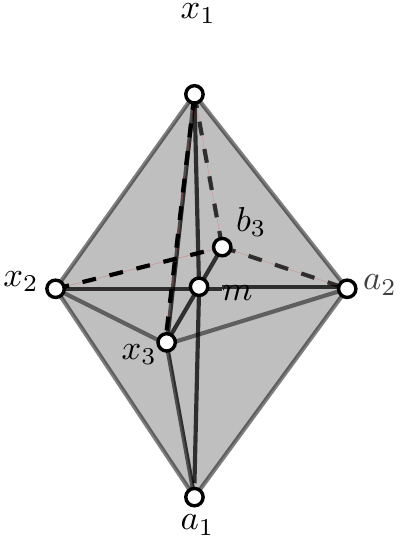}  & \includegraphics[scale=0.8]{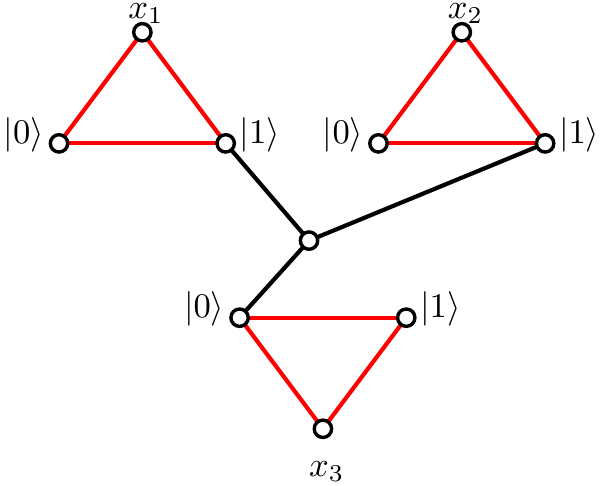}   \\  \hline 
$(\ket{101}-\ket{010}) (\bra{101}-\bra{010})$  &
\begin{minipage}[b]{0.45\linewidth}
    \includegraphics[width=1\linewidth]{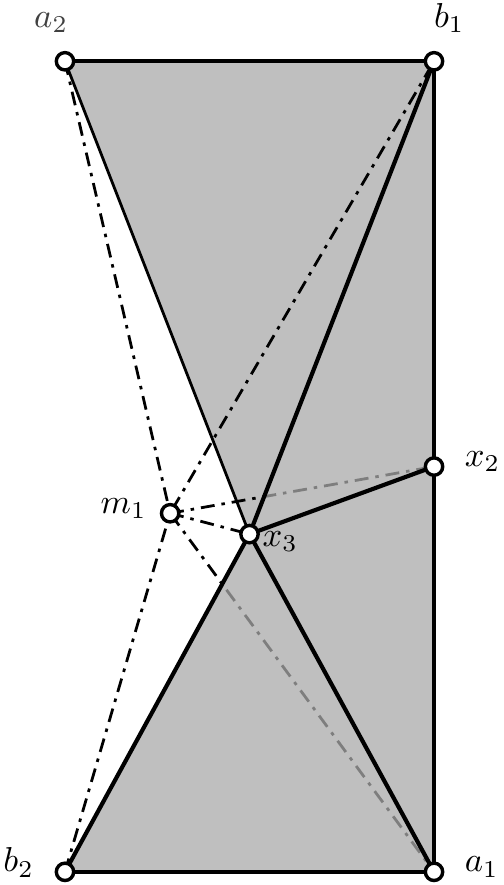} 
  \end{minipage} 
  \begin{minipage}[b]{0.45\linewidth}
    \includegraphics[width=0.75\linewidth]{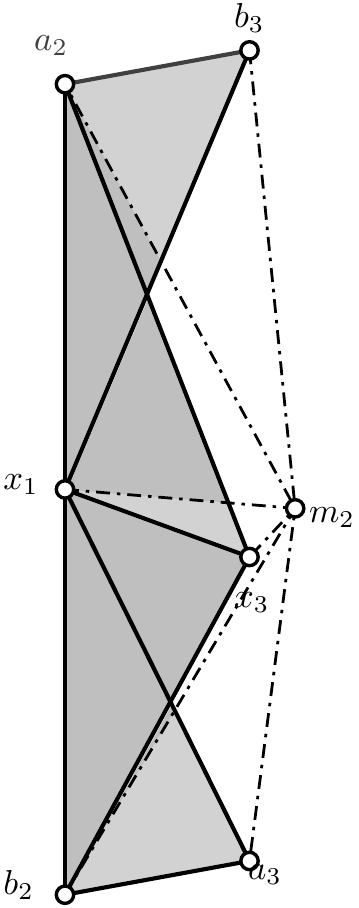} 
  \end{minipage} 
  \begin{minipage}[b]{0.45\linewidth}
    \includegraphics[width=1\linewidth]{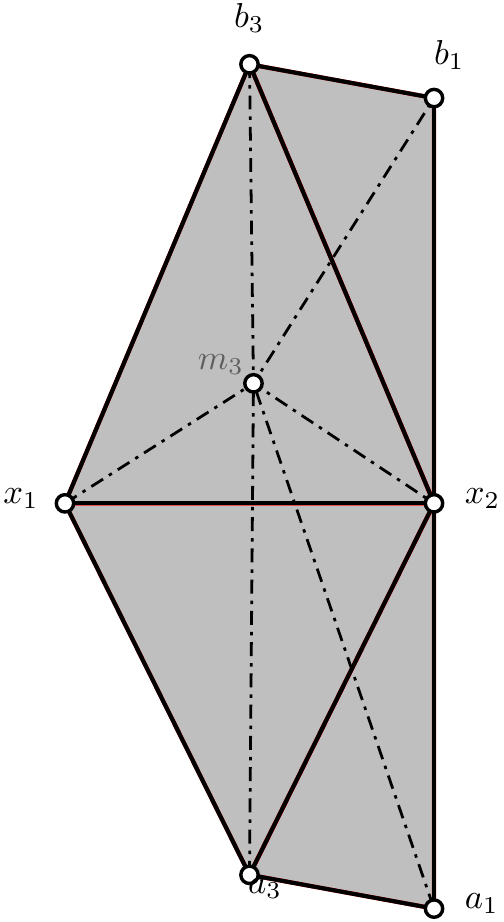} 
  \end{minipage}
  \begin{minipage}[b]{0.45\linewidth}
    \includegraphics[width=1\linewidth]{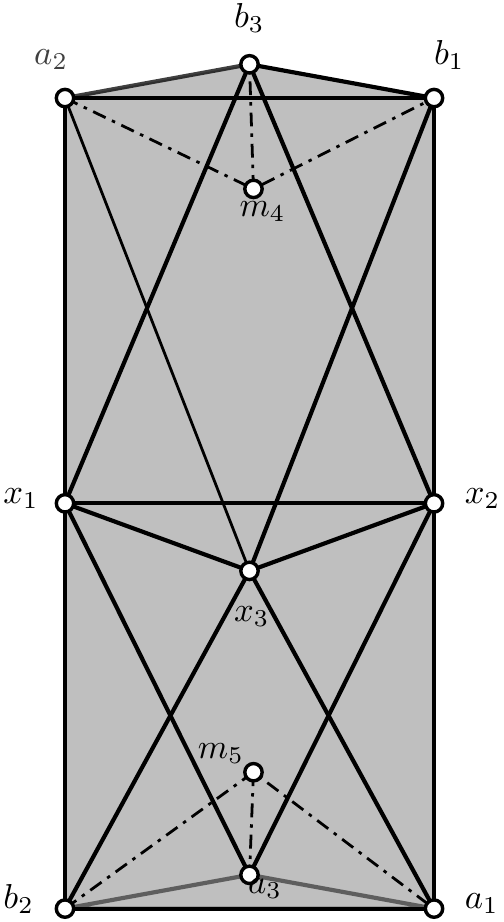}
  \end{minipage} 

&  \includegraphics[scale=0.75]{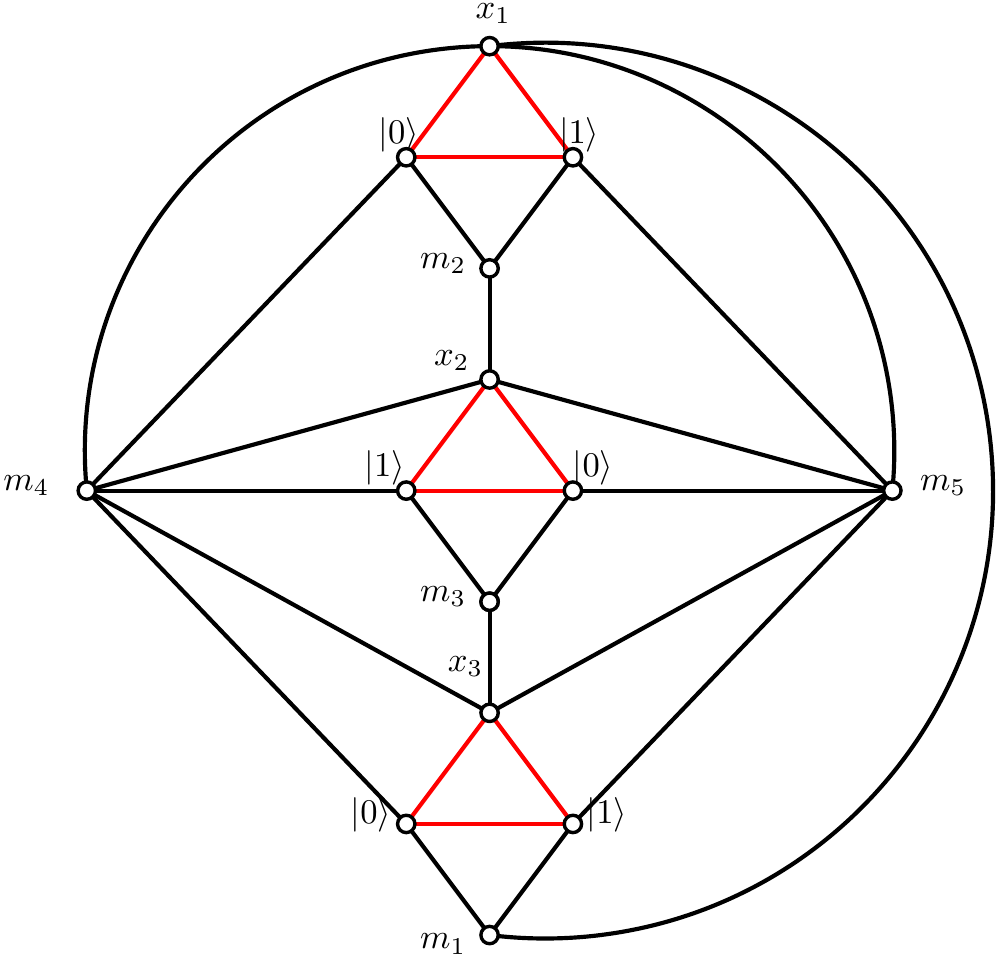} \\  \hline 
\end{tabular}
\end{center}
\caption{Collection of gadgets for the 3-qubit projectors in the clock Hamiltonian of quantum 4-$\SAT$.  
In order to aid readability for the entangled state we have shown the simplices connected to $m_1$, $m_2$ and $m_3$ on separate graphs, and shown how $m_4$ and $m_5$ are connected to the entire complex. We have also omitted the edges between mediators - these form a complete graph between the mediator vertices.
The gadget which lifts the state $\ket{011}-\ket{100}$ is realised by permuting the order of the first and second triangles in the gadget which lifts the state $\ket{101}-\ket{010}$.  The gadgets for propagation under the Pythagorean gate are more complicated, and are dealt with separately in \cref{sec:pythag}.
}
\label{table:gadgets3qubit}
\end{table}%

\subsubsection{The 3-qubit CNOT gadget}
\label{sec:cnot}

We now consider the projector
\equ{
\Pi= (\ket{101}-\ket{010})(\bra{101}-\bra{010})\,,
}
which arises from the CNOT gate (see \cref{table:4SAT}), lifting the entangled state $\ket{101}-\ket{010}$.  This cycle can be thought of as adding $\ket{101}$ and $\ket{010}$ with opposite orientations (see \cref{fig:sigma101010}).
\begin{figure}[]
\begin{center}
\includegraphics[scale=1]{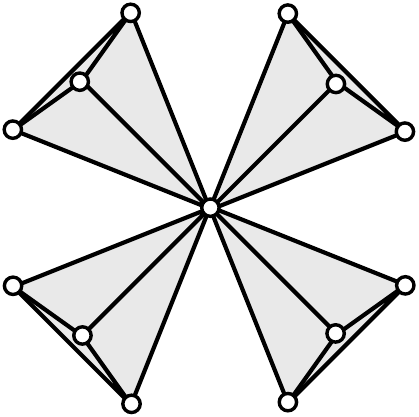} 
\caption{Four solid tetrahedra connected at a single mediator vertex $m$. The base of each tetrahedron is glued to a face in the 2-cycle to be filled in. }
\label{fig:X}
\end{center}
\end{figure}
To fill in this cycle consider four solid tetrahedra connected at a common vertex $m$, as shown in Fig.~\ref{fig:X}. The base of each tetrahedron can be glued to to one of the faces in $\Sigma_3$. There are a total of 14 faces so we introduce 3 mediators $m_{1,2,3}$ and glue these to 12 different faces. This leaves two faces unfilled, which we fill by simply introducing two additional mediators $m_{4}$ and $m_5$ forming two tetrahedra with those two faces. This partially fills in the three-dimensional void, leaving some empty wedges between the 14 tetrahedra. To fill these spaces we simply connect all the mediators $m_1,\cdots, m_5$ by edges and fill the wedges in between by tetrahedra.  The resulting simplicial complex consists of 14 0-simplices, 61 1-simplices, 97 2-simplices, 43 3-simplices, and 1 4-simplex (given by $[m_1...m_5]$). In this way the three-dimensional void has been completely filled in. As a simple consistency check we see that the Euler characteristic is $\chi=1-14+61-97+43-1=-7$, consistent with $H_2=(\Bbb C)^7$. The corresponding graph is shown in \cref{table:gadgets3qubit}. 

We have  algebraically checked that this procedure correctly fills in the desired void by showing that 
\equ{
\ket{101}-\ket{010} = \partial \ket{\Psi} \,.
}
See  \cref{app:algebraic} for an explicit expression for $\ket{\Psi}$.  Moreover, in \cref{app:mathematica} we check that (up to an irrelevant global phase) the state $\ket{101}-\ket{010}$ is the  only state which can be written in this way and is thus the {\it only} element that has been eliminated from $H_2$. 
\begin{figure}[]
\begin{center}
\includegraphics[scale=0.9]{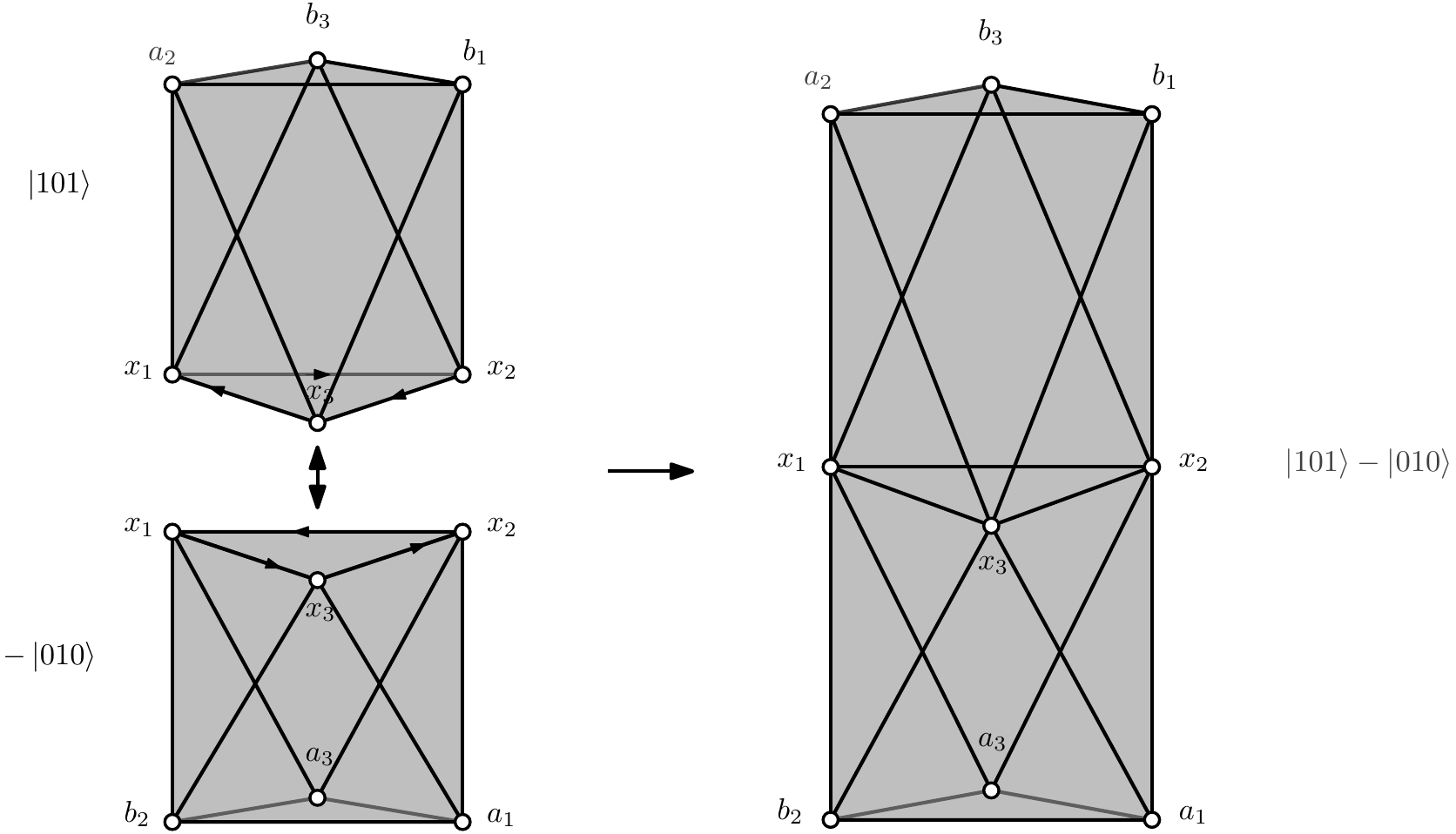}
\caption{Combining the two 2-dimensional voids $\ket{101}$ and $-\ket{010}$ into  the  2-dimensional  void bounded by the  entangled state $\ket{101}-\ket{010}$. Note that the face $[x_1x_2x_3]$ is common to both cycles but comes with opposite orientations and thus cancels in the superposition, leaving only {\it one} void to be filled. Additional vertices and edges which are not relevant have been omitted for clarity.  }
\label{fig:sigma101010}
\end{center}
\end{figure}

\subsubsection{General approach}

\begin{figure}[]
\begin{center}
\includegraphics[scale=0.9]{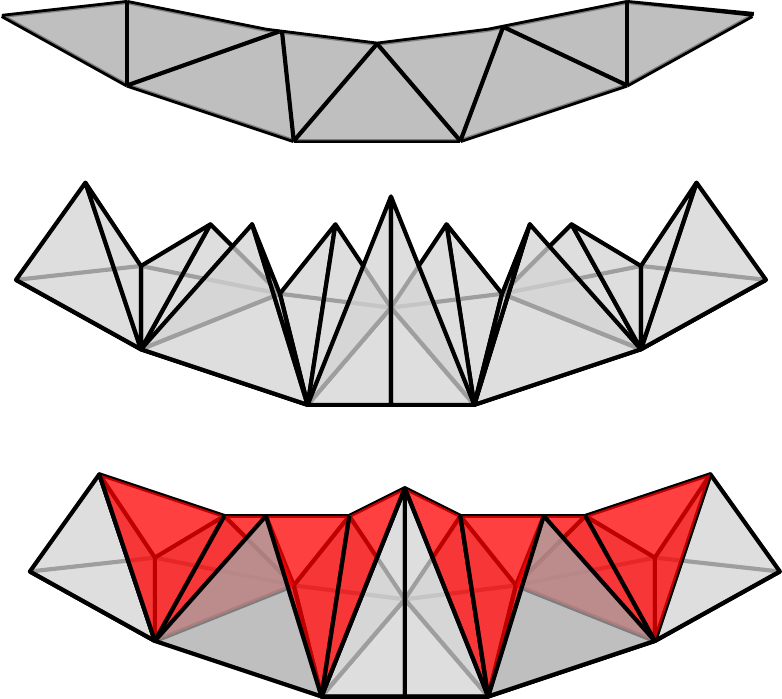} 
\caption{General approach to filling a 2-dimensional cycle. The first picture is a part of a 2-dimensional closed surface. For each face $[f_i]$ in the surface we add include a mediator $m_i$ and connect the mediator to all the vertices of $[f_i]$ to form a solid tetrahedron $[f_im_i]$. This leaves some empty voids in between the tetrahedra, which are filled by connecting mediators whose corresponding faces share an edge.This leads to a closed shell with an inner void. Finally, choosing one of the mediators to connect to all others fills the void and we obtain a three-dimensional space whose boundary is $\sum_i [f_i]$, as desired. }
\label{fig:filling}
\end{center}
\end{figure}

The procedures described above depend on the details of the cycle to be filled in. We now describe a general  procedure. The procedure we describe is far from optimal, in that it will generically required a greater number of mediators and edges that are minimally required, but it has the advantage of being of general applicability and no particular insights into the details of the cycle being filled in are required. 

This is the procedure we will use to fill in the cycles generated by propagation under the Pythagorean gate in the next section.
These will be the most technically involved gadgets required for this proof, and it is useful to understand the high level process before delving into the details of the specific gadget.  

Consider an $n$-dimensional cycle $\ket{s}$ defined by a set of $n$-simplices $[f_i]$, 
\equ{
\ket{s}=  \sum_i[f_i]\,.
}
To fill in this (and only this) cycle  we proceed as follows. For concreteness, consider a 2-cycle. Then, for every face $[f_i]$ we introduce a mediator $m_i$ and add the solid tetrahedron  $[f_i m_i]$ to the simplicial complex. This partially fills in the inner space but leaves some empty 3-dimensional ``wedges'' between the tetrahedra corresponding to neighboring faces (second picture in Fig.~\ref{fig:filling}). We can fill these  in by connecting the  mediators $m_i$ and $m_j$  with the edge $e_{ij}$ that is common among faces $[f_i]$ and $[f_j]$. This generates a ``shell'' whose outer boundary is the surface we wish to fill (see the third picture in Fig.~\ref{fig:filling}), leaving only an inner empty void. To fill in this inner void we add one final mediator vertex, which is connected to all other mediator vertices. Algebraically, this amounts to the statement that there exist a 3-chain $\ket{V}$ such that 
\equ{
 \ket{s}= \partial \ket{V}\sim 0\,.
}
Furthermore $\ket{s}$ is the {\it only} cycle rendered trivial in homology by this procedure. (If one wishes to fill in several cycles the procedure must be carried out for each one).

\subsubsection{The 3-qubit Pythagorean gadget}
\label{sec:pythag}
This is the most technically challenging gate. The states to be lifted are   
\equ{\label{pythstates}
\ket{\psi_{\mathit Pyth.}}=  \frac{1}{5\sqrt{2}} \left(-5\ket{011}+4\ket{100}+3\ket{101}\right)\,,\qquad 
\ket{\psi_{\mathit Pyth.}}= \frac{1}{5\sqrt{2}}  \left(-5\ket{010}+3\ket{100}-4\ket{101}\right)\,.
}
The complication stems mainly from the fact that there are three distinct computational cycles entering each entering  3, 4, or 5 times (up until now we only had computational cycles appearing with coefficients $\pm 1$). For concreteness we focus on the first state in \eqref{pythstates}. The three computational cycles are explicitly, 
\eqss{
A = \ket{011}=\,&  [x_{1}x_{2}x_{3}]+ [x_{1}b_{3}x_{2}]+ [b_{2}x_{1}x_{3}]+ [x_{1}b_{2}b_{3}]+ [a_{1}x_{3}x_{2}]+ [a_{1}x_{2}b_{3}]+ [a_{1}b_{2}x_{3}]+ [b_{2}a_{1}b_{3}]\,, \\ 
B = \ket{100} =\,& [x_{1}x_{2}x_{3}]+ [x_{1}a_{3}x_{2}]+ [a_{2}x_{1}x_{3}]+ [x_{1}a_{2}a_{3}]+ [b_{1}x_{3}x_{2}]+ [b_{1}x_{2}a_{3}]+ [b_{1}a_{2}x_{3}]+ [a_{2}b_{1}a_{3}]\,, \\
C = \ket{101} =\,& [x_{1}x_{2}x_{3}]+ [x_{1}b_{3}x_{2}]+ [a_{2}x_{1}x_{3}]+ [x_{1}a_{2}b_{3}]+ [b_{1}x_{3}x_{2}]+ [b_{1}x_{2}b_{3}]+ [b_{1}a_{2}x_{3}]+ [a_{2}b_{1}b_{3}]\,.
}
The strategy is then to consider copies of each $-A$, $B$, and $C$ cycles separately and glue them appropriately to obtain the cycle $\ket{\psi_{\mathit Pyth.}}$. To do so we describe a procedure which we term ``simplicial surgery,'' in which we cut each of these cycles in a certain way (described below) and glue them along the cuts to form the closed surface $\ket{s}=\ket{\psi_{\mathit Pyth.}}$. Once $\ket{\psi_{\mathit Pyth.}}$ is obtained this way, we fill it in by tetrahedra by the general method described in \cref{sec:pythag}. The procedure for the second state in  \eqref{pythstates} is similar and we provide the details in  \cref{app:pythag2}.

To understand how to precisely  glue these cycles,  note that  $B$ and $C$  have the two edges $[b_1x_2]$ and $[b_1a_2]$ in common. We will glue together the $B$ and $C$ cycles by cutting open the cycles along this shared edge, and gluing together the cycles along the open 1-cycle this creates. Take four copies of $B$  and three copies $C$. Then, in each cycle create a duplicate vertex $b_1'$ and include the edges $[b_1'x_2]$ and $[b_1'a_2]$. This leads to cycles with slits cut into them (see \cref{fig:slits}). Now we glue together the cycles by identifying the $a_2$ and $x_2$ vertices from the different cycles, and identifying $b_1'$ from cycle $i$ with $b_1$ from cycle $i+1$. 
 
We then attach to this the five $A$ cycles. This step is more straightforward since the $A$ cycles are added with the opposite orientation to the $B$ and $C$ cycles.  This means that when we are adding the cycles the common face $[x_1x_2x_3]$ cancels out and we can glue the cycles along this face.
When adding an $A$ cycle to a $B$ cycle this simply amounts to identifying the $x_1$, $x_2$ and $x_3$ vertices from the $A$ cycle with that of the $B$ cycle, and adding the edges of the $A$ cycle in the appropriate orientation. When adding the $A$ cycles to the $C$ cycles they also share the common face ${x_1b_3x_2}$, so the $x_1$, $x_2$, $x_3$ and $b_3$ faces from the $A$ cycle are identified with those of the $C$ cycle, and again the edges of the $A$ cycle are added in the appropriate orientation. 
\begin{figure}[]
\begin{center}
\includegraphics{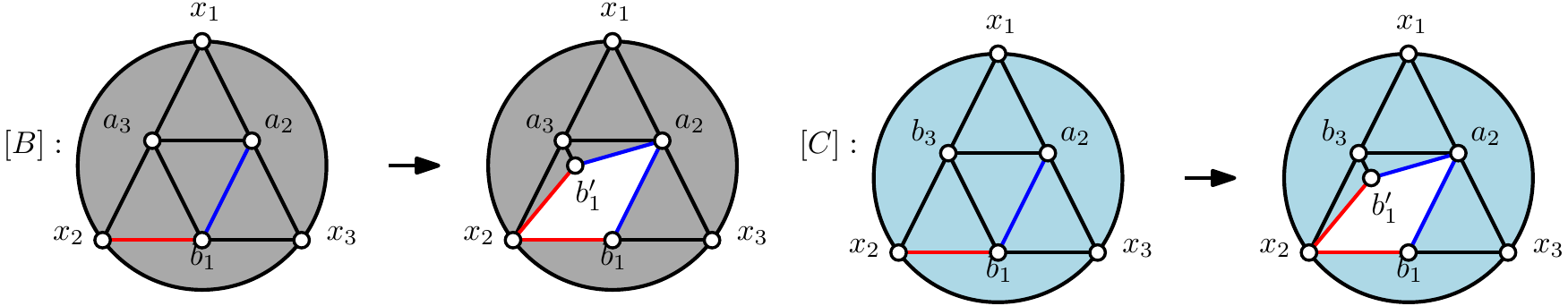}
\caption{The $[B]$ and $[C]$ cycles are cut open by cutting a slit along $[x_2b_1]+[b_1 a_2]$ by duplicating the middle point $b_1$. Here we have omitted vertices and edges to keep the pictures clear but one should bear in mind that before performing the cut each is a closed 2-cycle, homeomorphic to $S^2$.  Then, the open cycles are  connected with each other by gluing through the common one-cycles  $[x_2b_1]+[b_1a_2]+[a_2b_1']+[b_1' x_2]$, ensuring the orientation is consistent, and attaching the $[A]$-cycles by deleting the common face $[x_1x_2x_3]$. The resulting space, $A_5B_4C_3$ is a closed surface composed of the $[B]$ and $[C]$ cycles with the five $[A]$-cycles attached, as shown in Fig.~\ref{fig:PythCycle}. As a whole, the space $A_5B_4C_3$ is homeomorphic to $S^2$ and it is this topological $S^2$ which we wish to fill. }
\label{fig:slits}
\end{center}
\end{figure}
 This leads to the closed surface $A_5B_4C_3$ shown in \cref{fig:PythCycle}, which is bounded by $8 \times 2 +7\times 4 +6 \times 6=80$ faces, which we denote by $[f_i]$.\footnote{The two lots of 8 faces come from the $B$ cycles which aren't attached to any $A$ cycle. The four lots of 7 faces come from the two sets of $A$ and $B$ cycles that are attached to each other, since the common face $[x_1x_2x_3]$ has cancelled out. The six lots of six faces comes from the 3 sets of $A$ and $C$ cycles that are attached together since the common faces $[x_1x_2x_3]$ and $[x_1b_3x_2]$ have cancelled out.}
\begin{figure}[]
\begin{center}
\includegraphics[scale=0.8]{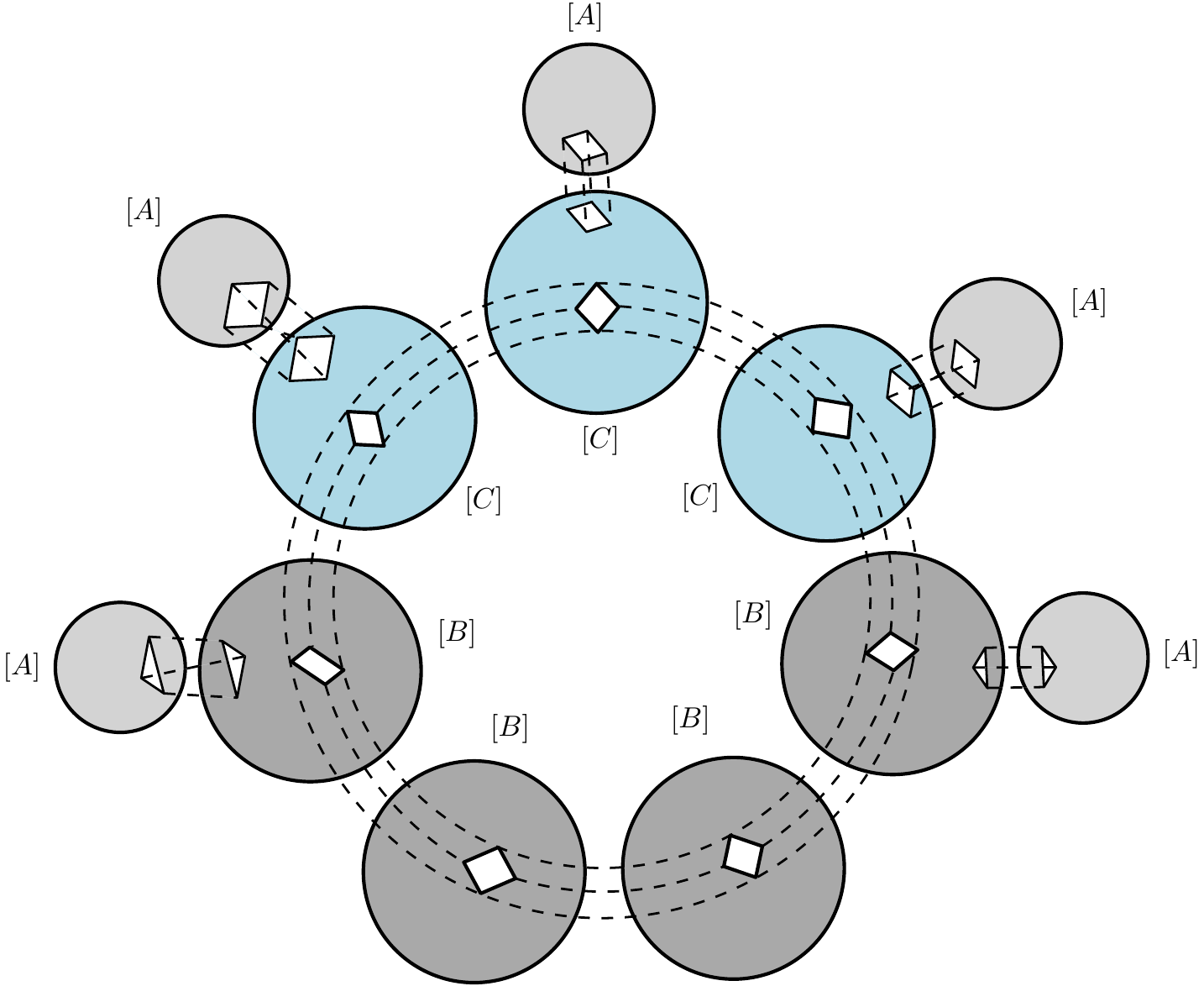}
\caption{The Pythagorean cycle to be filled in is obtained by attaching five copies of $[A]$ which corresponds to the cycle $\ket{011}$ with that face $[x_1x_2x_3]$ missing.  }
\label{fig:PythCycle}
\end{center}
\end{figure}
The resulting surface is homeomorphic to $S^2$. To fill this surface in we follow the general procedure from the previous section. 
For each face composing the  surface we introduce a mediator vertex $m_i$ and connect each to all the vertices of the corresponding $[f_i]$, forming a solid tetrahedron $[f_i\, m_i]$. 

This introduces 80 mediator vertices which partially fill in the cycle defined by $\ket{\psi_{\mathit Pyth.}}$, but leave some gaps. 
We then introduce an edge between two mediators if the faces they are connected to share an edge or a vertex, which starts filling the gaps. 
Finally we add a last mediator vertex, and connect it to the other 80. 
This fills in all gaps in the surface, so that the cycle defined by $\ket{\psi_{\mathit Pyth.}}$ is closed.

We then identify all the copies of the original vertices with themselves, leaving the connections between the original vertices and the mediators unchanged.
This leads to a simplicial complex $\Sigma_3'$ consisting of $90$ 0-simplices, $581$ 1-simplices, $1049$ 2-simplices, $597$ 3-simplices, and $29$ 4-simplices. The Euler characteristic is $1-90+581-1049+579-29=-7$, consistent with $H_2(\Sigma_3')=(\Bbb C)^7$. 

The graph gadget $G_3'$ is shown in \cref{fig:ComplementPyth}. The resulting graph has $9+81=90$ vertices, $3424$ edges and maximum vertex degree  $\delta=81$. The explicit procedure for constructing this gadget is carried out step-by-step in the attached Mathematica file \cite{mathematica}.
Details of the calculations carried out to show that this gadget closes the cycle associated to $\ket{\psi_{Pyth.}}$ and only this cycle are given in \cref{app:mathematica}.  A similar process works for the other Pythagorean gadget, details can be found in \cref{app:pythag2}.

We point out that this technique of gluing faces together can be applied in general to any state of the form
\equ{ \label{eq:psin}
\ket{\psi}= \sum_{i}n_{i} \ket{e_{i}}\,,
}
where the $\ket{e_{i}}$ are the computational cycles and the  $n_{i}\in \Bbb Z$ are {\it integer} coefficients. The method does not apply, however, if the coefficients in \eqref{eq:psin} are general complex numbers. Although such states are well defined as abstract elements of the chain complex $C(\Sigma)$, they cannot be interpreted geometrically as cycles in a {\it simplicial} complex and our procedure for lifting such states in terms of adding discrete data such as mediator vertices and edges is not available. This is the reason we consider the non-standard universal gate set in our definition of $\QMA_{1}$, which involves only rational coefficients and requires to lift the with states with integer coefficients  \eqref{pythstates}. As long as the coefficients in the universal gate set are rational, however, this method carries through and therefore the homology problem is $\QMA_1$-hard for all such definitions of $\QMA_1$.

\begin{figure}[]
\begin{center}
\includegraphics[scale=0.5]{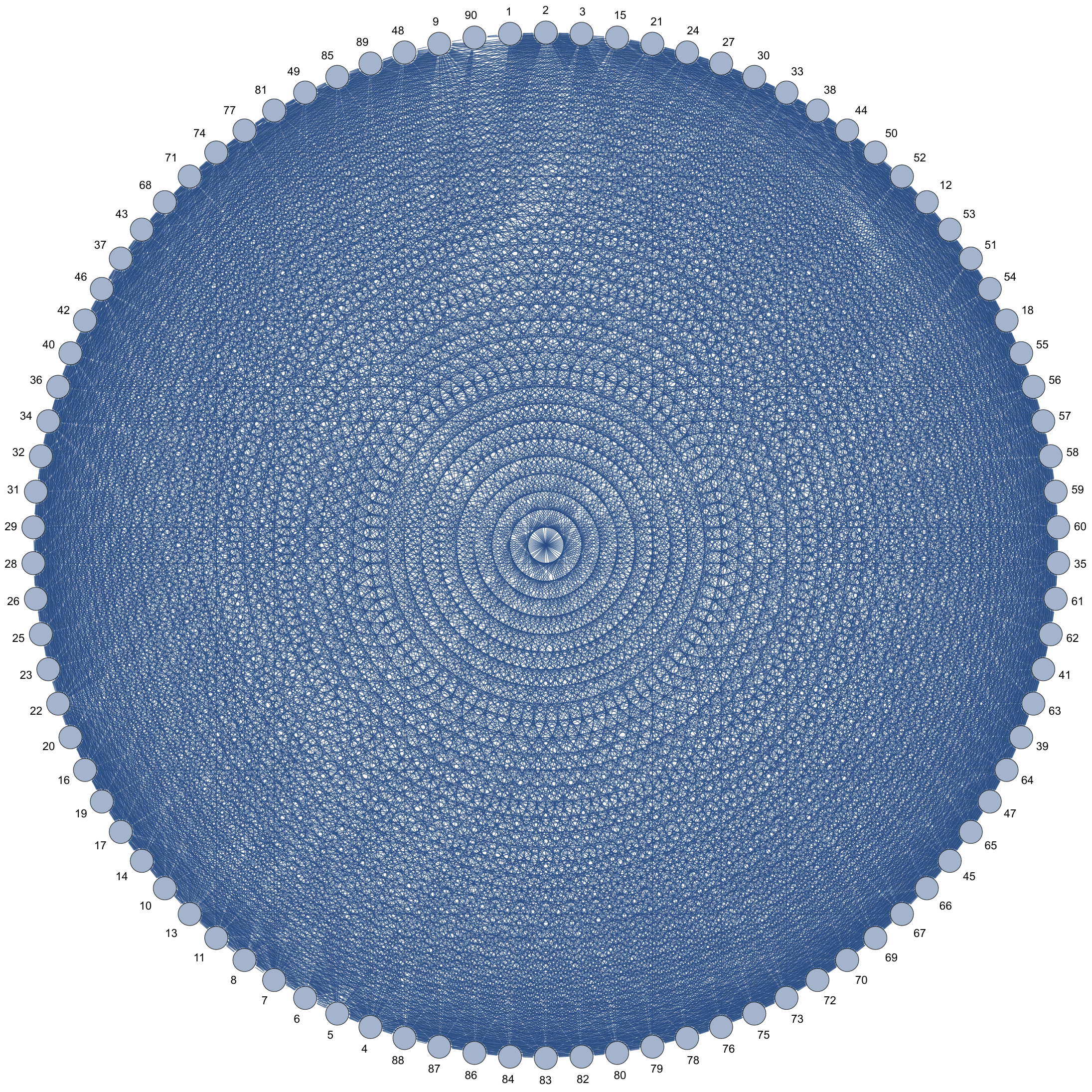}
\caption{The Pythagorean gadget graph $G_{\mathit Pyth.}$ The maximum degree of $G_{\mathit Pyth.}$ is $\delta=81$.  
}
\label{fig:ComplementPyth}
\end{center}
\end{figure}

\subsubsection{Combining gadgets together}  \label{sec:combining_gadgets}

The projectors needed for the various Hamiltonian terms in Table~\ref{table:4SAT} involve the tensor product of the projectors just discussed. Furthermore, in order to construct the full Hamiltonian we will need to take the sum of all these terms.  Thus, the final ingredient needed for the reduction from quantum 4-$\SAT$ to the homology problem is to understand how to take tensor products and sums of the gadgets just described. Luckily, this is straightforward. 

\paragraph{Adding projectors.}

Adding two projectors simply projects onto the common null space. Thus, given two projectors $\Pi_{1}$ and $\Pi_{2}$ the gadget for the sum
\equ{
\Pi=\Pi_{1}+\Pi_{2}
}
is obtained by independently filling in the corresponding cycles in $\Sigma_{n}$ for each projector. 
If any qubits are simultaneously acted upon by both projectors then we need to overlap the gadgets, while taking care not to render any other states equivalent in the homology.
The general procedure for doing this for a pair of projectors $\Pi_1$ (implemented by the graph $G_1$) and $\Pi_2$ (implented by the graph  $G_2$) which act on a total of $n$ qubits is:
\begin{enumerate}
    \item Start constructing the graph $G_{12}$ with $n$ triangles.
    \item If there are any mediators in the graphs $G_1$ and $G_2$ which connect to exactly the same set of qubit vertices in the two graphs, add this mediator and the corresponding connections to the graph $G_{12}$ \emph{once}. Label this set of mediator vertices by $\{m_{12}\}$.
    \item Take the remaining mediators from the graph $G_1$ and add them, and their corresponding connections, to the graph $G_{12}$, and label this set of mediator vertices by $\{m_{1}\}$. Do the same for $G_2$, and label this set of mediator vertices by $\{m_{2}\}$.
    \item Add connections from every vertex in the set $\{m_{1}\}$ to every vertex in the set $\{m_{2}\}$.
\end{enumerate}

The final step in this procedure may seem superfluous.
Indeed, constructing the gadgets without this step does lead to the desired states being lifted, but it can also render some ground states equivalent in the homology.
If this occurs then the homology group of the corresponding graph is no longer isomorphic to the ground state of the desired projector.
Adding connections between the mediators arising from the different projectors prevents this from happening, preserving the isomorphism between the homology group and the  ground state.
This point is explained in more detail in \cref{app:adding_projectors}.

\paragraph{Tensoring projectors.} 
\begin{figure}[]
\begin{center}
\includegraphics{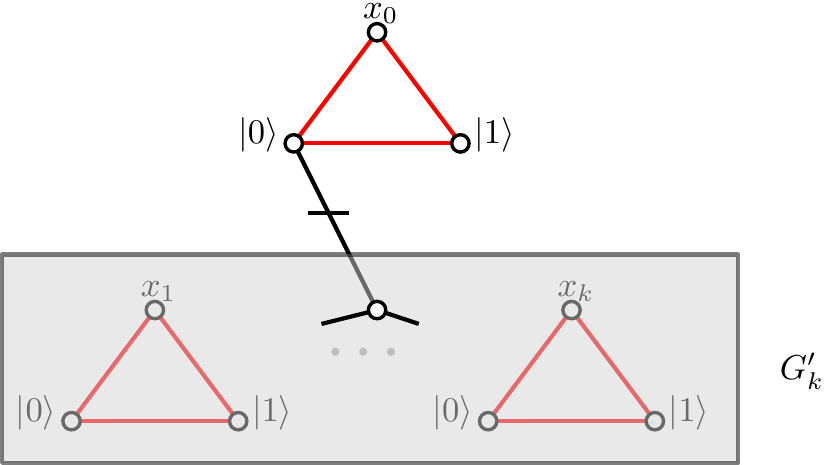}
\caption{The qubit 0 at the top realizes the 1-qubit projector $\ket{1}\bra{1}$ and the bottom gadget $G_{k}'$ implements $\Pi^{(k)}$. Connecting them as in the figure implements  the $(k+1)$-qubit projector $\Pi^{(k+1)}=\ket{1}\bra{1}\otimes \Pi^{(k)}$.   }
\label{fig:1controlG}
\end{center}
\end{figure}

A number of tensor products between entangling projectors and classical projectors appear in \cref{table:4SAT}. Consider a $ (k+1)$-qubit projector of the form
\equ{
\Pi^{(k+1)}= \ket{1}\bra{1}\otimes \Pi^{(k)}\,,
}
where $\Pi^{(k)}$ is a $k$-qubit projector and the identity operator is understood to be acting on the remaining qubits. If the first qubit above is in the state $\ket{0}$, any configuration of the remaining $k$ qubits is in the kernel of $\Pi^{(k+1)}$. If instead it is in the state $\ket{1}$ this forces the $k$ qubits to be in the  kernel of $\Pi^{(k)}$. We wish to understand how to reproduce this in terms of the graph gadgets for $\ket{1}\bra{1}$ and $\Pi^{(k)}$.  For clarity we  consider an explicit example, but the logic can  be applied to any such gadget. 
The case we will consider is the gadget for lifting the state $\ket{1}\left(\ket{01}-\ket{10}\right) $.
The cycle that must be filled in order to lift the state $\ket{01}-\ket{10}$ is given by:
\equ{ \label{eq:tensoring}
\ket{01}-\ket{10} = [x_1a_2]+[a_2b_1] + [b_1x_2] + [x_2a_1] + [a_1b_2]+[b_2x_1] \,.
}
Tensoring on the state $\ket{1}=[x_3]-[b_3]$ is equivalent to taking the wedge product $(\ket{01}-\ket{10}) \ket{1}=([x_1a_2]+[a_2b_1] + [b_1x_2] + [x_2a_1] + [a_1b_2]+[b_2x_1])\wedge ([x_3]-[b_3])$. Recall that  when we lifted the state  $\ket{01}-\ket{10}$ we filled the 1-cycle defined by the edges given in \eqref{eq:tensoring}. To lift $\ket{1}\left(\ket{01}-\ket{10}\right)$ we clearly need to fill in a 2-cycle, and we can see the form of that 2-cycle by considering the wedge product above.  For each 1-simplex in \eqref{eq:tensoring} the tensor product will give rise to two 2-simplices - one including the original two vertices and $x_3$, the other with the original two vertices and $b_3$ (with opposite orientations). Thus, the cycle for the state $\left(\ket{01}-\ket{10}\right)\ket{1} $ will involves 12 2-simplices which bound the original 1-cycle (see \cref{fig:3Dcycle}).

It is straightforward to see that taking the independence complex which filled the 1-cycle for $\ket{01}-\ket{10}$, and adding edges between each mediator and $x_3$ and $b_3$ will fill the required 2-cycle.
From this independence complex, $\sum_t$, we can construct the graph $G_t$ - see \cref{fig:newG}.

\begin{figure}[]
\begin{subfigure}{0.5\textwidth}
\centering
\includegraphics{Figures/0101.pdf}
\caption{Graph $G$ to fill in the state $\ket{01}-\ket{10}$}
\end{subfigure}
\begin{subfigure}{0.5\textwidth}
\centering
\includegraphics{Figures/Sigma2ppp.pdf}
\caption{Independence complex $\sum$ to lift the state $\ket{01}-\ket{10}$.}
\end{subfigure}
\begin{subfigure}{0.5\textwidth}
\centering
\includegraphics{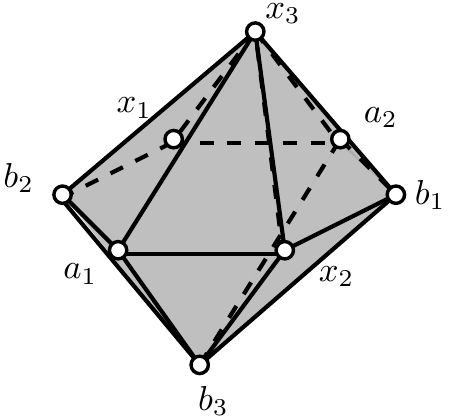}
\caption{To lift the state $\ket{1}\left(\ket{01}-\ket{10}\right)$ we need to fill in the 3-dimensional void shown in this figure. Irrelevant vertices and edges have been omitted for clarity. } \label{fig:3Dcycle}
\end{subfigure}
\begin{subfigure}{0.5\textwidth}
\centering
\includegraphics[scale=0.75]{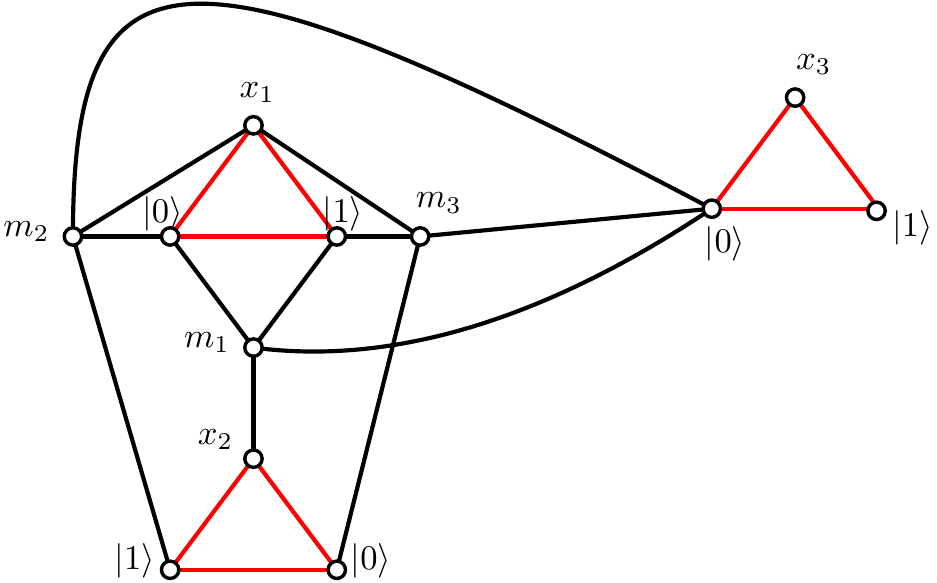}
\caption{Graph $G_t$ to lift the state $\ket{1}\left(\ket{01}-\ket{10}\right)$.} \label{fig:newG}
\end{subfigure}
\caption{Tensoring on a classical projector to an entangled one is equivalent to turning the $k$-simplices bounding the cycle from the original projector into $k+1$-simplices which bound the cycle we need to fill in for the new projector. This new void can be filled in by adding edges between the additional vertices bounding the cycle and all mediators from the original projector in the independence complex. In the graph this is equivalent to adding edges between the vertex \emph{not} bounding the cycle and all the mediators.} \label{fig:tensorProducts}
\end{figure}

The same logic carries over to the other examples where we have an entangled projector tensored with a classical projector.
In each case, the $k$-simplices bounding the original cycle become $k+1$-simplices bounding the new cycle, and adding edges in the independence complex between the existing mediators and the additional vertices bounding the new cycle is sufficient to fill the necessary void.  In the graph this is equivalent to adding an edge between the vertex that does \emph{not} bound the new cycle (i.e. the vertex $a$ if the classical projector is $\ket{1}\bra{1}$ and the vertex $b$ if the classical projector is $\ket{0}\bra{0}$).
See \cref{fig:1controlG}.\footnote{For readers who have already covered the connection to supersymmetry in \cref{sec:connection to SUSY} the form of these tensor product gadgets can be checked very quickly. Clearly if the qubit with the classical projector acting on it is in the kernel of the classical projector then the mediators are forced to be empty by the  hard core condition, and the homology is non-trivial regardless of the state of the remaining qubits. If the classical qubit is not in the kernel of the classical projector, then the remaining qubits are forced to be in the kernel of the entangled projector otherwise the mediators will have trivial homology, and by the tic-tac-toe lemma so will the overall state.}

\paragraph{Tensoring on the identity} If we were to apply the logic for adding and tensoring projectors to the case of tensoring on the identity it would appear that this is a complex task.
We would need to first construct the `identity projection' onto some set of qubits by adding together a number of classical projectors, then tensor this sum of projectors onto the relevant non-trivial projection.

Fortunately, we can take a short cut that avoids this complication. Consider a projector $\Pi_i$ acting on some set of qubits $\{i_k\}$ with gadget given by the graph $G_i$.
If we want to tensor this with the identity on qubit $j$ in the graph we simply take the graph $G_i$ and add a triangle to represent the qubit $j$ which is disjoint from the rest of the graph. 

To see why this is equivalent to tensoring on the identity acting on qubit $j$ consider \cref{fig:tensorProducts} again.
In that figure we demonstrate how to tensor on the classical projector $\ket{1}\bra{1}$ to the entangled projector $(\ket{0}-\ket{1})(\bra{0}-\bra{1})$. 
To modify this gadget so that we were instead tensoring the entangled projector with the identity we need to also lift the state $\ket{0}(\ket{0}-\ket{1})$.
Examining \cref{fig:newG} and \cref{fig:3Dcycle} we see that removing the connections between the $\ket{0}$ vertex of qubit $j$ and the mediators from the graph $G_i$ will achieve this, since it will give rise to another filled cycle of the same form as that in \cref{fig:3Dcycle}, while it won't change the original filled cycle.\footnote{Again, readers who have already covered the material in \cref{sec:connection to SUSY} will immediately see that this is the correct procedure - in the Hamiltonian picture we are simply stating that acting with the identity on a set of qubits is equivalent to doing nothing.}

\subsection{$\QMA_1$-hardness}
\label{sec:reduction}

The homology problem for simplicial complexes, parameterised by an input type $\mathcal{I}$ was defined in \cite{ADAMASZEK20168}:
\\
\myprob{{\sc $\text{Homology}_{\mathcal{I}}(K,l)$}}{A simplicial complex $K$ represented as $\mathcal{I}$ and an integer $l$.}{Output $\yes$ if the $l^{\text{th}}$ homology group of $K$ is non-trivial and $\no$ otherwise.\\} 

\noindent In this section, we will prove that this problem is  $\QMA_1$-hard when the simplicial complex, $K$, is a clique complex $Cl(G)$ represented by $G$. 
As a stepping stone we consider the closely related problem when the complex is an independence complex $I(G)$ represented by $G$:

\begin{lemma}\label{lem:ind}
{\sc $\text{Homology}_{G}(I(G),l)$} is $\QMA_1$-hard.
\end{lemma} 
\begin{proof}
We will demonstrate hardness by a reduction from an arbitrary problem in $\QMA_1$ via the family of Hamiltonians used to show $\QMA_1$-hardness of quantum 4-$\SAT$ in \cite{bravyi2011efficient}.
That is, for any Hamiltonian $\Hbravyi$ which encodes a $\QMA_1$-verification circuit via the techniques of \cite{bravyi2011efficient}, we will construct a graph $G$ such that a particular homology group $H_l(d)$ of the independence complex $I(G)$ is non trivial iff there exists a satisfying assignment to the corresponding instance of quantum $k$-$\SAT$.
And therefore iff there exists a witness such that the $\QMA_1$-verification circuit accepts with probability 1. 

The ingredients for the construction are clear.
Given an arbitrary problem in $\QMA_1$ with input string $x$ we construct a $4$-local Hamiltonian $\Hbravyi$ which encodes the problem and acts on $n = \poly(|x|)$ qubits.
We then construct a graph composed of $n$ triangles (i.e. consisting of $3n$ vertices and $3n$ edges)
and add mediators connecting   triangles if they are jointly acted on by a projector in $\Hbravyi$.
Where the form of the connections between triangles is given by the gadgets derived in the previous sections.
We will demonstrate that the independence complex to the graph has non-trivial homology iff the encoded problem is a $\yes$ instance.

\paragraph{The input to the problem:}
$\Hbravyi$ is made up of at most $\poly(n)$ rank-1 projectors.
The gadgets we have constructed add $O(1)$ vertices to the graph for each projector.
Therefore the final graph has $V = \poly(n)$, and the input size to the problem is $|V(G)|^2 = \poly(n) = \poly(x)$. The final part of the input to the problem is the integer $l$. Recall from the problem definition that the quantity of interest is the $l^{\textrm{th}}$ homology group. 
For any fixed $l$ this problem can be efficiently solved, however if $l$ is allowed to grow with the input to the problem this is no longer the case.
For our purposes, we will take $l=n-1$.

\paragraph{Completeness:}
Given a $\yes$ instance of an arbitrary problem in $\QMA_1$ there is a corresponding satisfying assignment to the instance of quantum 4-$\SAT$ that encodes the verification of this problem. 
Therefore there exists a state $\ket{\Psi}$ satisfying $\Hbravyi\ket{\Psi} =0 $.
This state will correspond to an $(n-1)$-cycle, $c$, in the independence complex which is not a boundary - i.e. we have $\partial_{n-1}c = 0$ but $\nexists d$ such that $\partial_n d = c$.
This follows since in the initial graph - consisting simply of $n$ triangles - every element of the Hilbert space $(\mathbb{C}^{2})^{\otimes n}$ is an $(n-1)$-cycle which is not a boundary.
Adding the gadgets to build up a graph encoding the Hamiltonian $\Hbravyi$ does not change the first equation, $\partial_{n-1}c=0$ because no connections are added or removed between the initial $3n$ vertices by the gadgets, so the expression $\partial_{n-1}c$ is unchanged.
Moreover, the gadgets fill in cycles iff they correspond to states that violate one of the projectors in $\Hbravyi$.
But $\ket{\Psi}$ cannot violate any of the projectors in $\Hbravyi$, and as such the cycle corresponding to $\ket{\Psi}$ has not been filled in, and therefore we cannot write $c$ as the boundary of some $n$-cycle.
Therefore for $\yes$ instances the $(n-1)^{\text{th}}$ homology group of  $I(G)$ is non-trivial. 

\paragraph{Soundness:}
Given a $\no$ instance of an arbitrary problem in $\QMA_1$ the instance of quantum 4-$\SAT$ that encodes the problem has no satisfying assignments.
Therefore every cycle which was not a boundary in the trivial graph (consisting of $n$ disconnected triangles) has been filled in by the inclusion of the gadgets, and therefore the $(n-1)^{\text{th}}$ homology group of  $I(G)$ is trivial.

Therefore, the $(n-1)^{\text{th}}$ homology group of  $I(G)$ is non-trivial iff $G$ encodes a $\yes$ instance of a $\QMA_1$ problem.
\end{proof}

\begin{theorem}
{\sc $\text{Homology}_{G}(Cl(G),l)$} is $\QMA_1$-hard.
\end{theorem}
\begin{proof}
As noted in \cref{subsec:cli/in} $Cl(G) = I(\overline{G})$.
Moreover, given $G$ the complement graph $\overline{G}$ can be computed in $\P$.
Therefore an arbitrary problem in $\QMA_1$ can be reduced to {\sc $\text{Homology}_{G}(Cl(G),l)$} via the reduction to {\sc $\text{Homology}_{G}(I(G),l)$} given in  \cref{lem:ind}.
\end{proof}

Furthermore, we note that the homology problem for clique complexes of graphs can be reduced to the homology problem for simplicial complexes given by vertices and minimal non-faces \cite{ADAMASZEK20168}, and therefore we can derive the corollary:

\begin{corollary}
{\sc $\text{Homology}_{\mathcal{I}}(K,l)$} is $\QMA_1$-hard when $\mathcal{I}$ is the vertices and minimal non-faces of $K$. 
\end{corollary}

\subsection{Containment in $\QMA$ with a suitable promise}

To further pin down the complexity of the problem we need to consider the issue of a suitable promise.
\homologyProb\ was presented as a decision problem - i.e. a promise problem with a vacuous promise.
While it is known that there are problems in quantum complexity classes with vacuous promises (e.g. the \emph{group non-membership} problem is known to be in $\QMA$ \cite{Watrous:2000}) it is unlikely that a problem with a vacuous promise could be \emph{complete} for $\QMA_1$ \cite{Watrous:2008}. Therefore we do not expect \homologyProb\ as stated to be contained in $\QMA_1$. 

From the statement of \homologyProb\ above it is not immediately obvious that it is possible to present it with a non-vacuous promise.
A simplicial complex either has holes or doesn't, it appears to be a completely binary decision problem.
However, it can be shown that the $p^{\text{th}}$ homology group, $H_p(K)$ of a simplicial complex $K$ is non-trivial iff the kernel of the combinatorial Laplacian $\mathcal{L}_p = \partial^\dagger_p\partial_p +\partial_{p-1}\partial^\dagger_{p-1}$ is non-empty (in fact, the $p^{\text{th}}$ Betti number is given by the dimension of $\text{ker}(\mathcal{L}_p)$, see e.g. \cite{HORAK2013303}).
This suggests a reformulation of \homologyProb\ with a suitable promise (this version first appeared in \cite{Cade:2021jhc} for the general case of chain complexes, here for clarity we state it for simplicial complexes): \\

\mypromprob{{\sc $\text{Promise-Homology}_{\mathcal{I}}(K,p)$}}{A simplicial complex $K$ represented as $\mathcal{I}$ and an integer $p$. }
{Either there exists a cycle $c \in C_p(K)$ such that $\mathcal{L}_p(c)=0$ and hence $H_p(K)\neq 0$, or the smallest eigenvalue $\lambda_p$ of $\mathcal{L}_p$ satisfies $\lambda_p \geq \epsilon$ with $\epsilon = \frac{1}{\poly(|\mathcal{I}|)}$}
{Output $\yes$ if the former and $\no$ if the latter. \\}

With this we can demonstrate that the homology problem for clique complexes (with a suitable promise) satisfying certain constraints lies in $\QMA$:

\begin{theorem}\label{thm:containment}
{\sc Promise-Homology} for clique complexes $Cl(G)$ where the graph $G$ is given as input is contained in $\QMA$ provided that either $G$ is clique dense, or that the boundary operator $\partial$ of $Cl(G)$ is local.
\end{theorem} 

\begin{proof}
This follows from \cite[Theorem 4]{Cade:2021jhc} and \cite[Proposition 1]{Cade:2021jhc}, where the statement is shown for $k$-local or sparse boundary operators acting on chain complexes, with the operator $\partial$ given as input.\footnote{The proof in \cite{Cade:2021jhc} uses quantum phase estimation on $(\partial \pm \partial^\dagger)$ with respect to the witness state to check whether it is a state with low eigenvalue to within polynomial accuracy in polynomial time.} 
The $\partial$ operator for a clique complex is always sparse \cite{gyurik:2020}, and restricting the set of inputs cannot take the problem outside of $\QMA$.
So provided it is possible to efficiently obtain a description of $\partial$ from the graph $G$ the theorem holds.

We first consider the $k$-local case.
The operator $\partial$ can be expressed as a sum over the vertices of $\overline{G}$ where at each vertex we construct an operator acting only on adjacent vertices (see \cref{sec:connection to SUSY} for details of this description).
We can construct $\overline{G}$ efficiently given access to $G$.
Moreover, if $\overline{G}$ has constant degree (which it must do if $\partial$ is local) then we can construct the $\partial$ operator efficiently from a description of $\overline{G}$. If the boundary operator resulting from $G$ is not local, but $G$ is clique dense then sparse access to $\partial$ can be obtained efficiently from a description of $G$ via the prescription given in \cite{gyurik:2020}.

\end{proof}

Note that it is possible to modify the reductions from the previous section so that they result in simplicial complexes with local boundary operators (see \cref{app:locality}) and therefore satisfy the second requirement in \cref{thm:containment}. Taken with the results in the previous sections, this suggests that the problem of determining whether the homology group of a clique complex is non-trivial may lie between $\QMA_1$ and $\QMA$. 
However care must be taken -- we have not shown that the promise version of the homology problem is $\QMA_1$-hard, or that the decision version of the problem is in $\QMA$.
So while we believe that a result along these lines is likely to hold, we cannot claim it yet.
In \cref{sec:outlook} we discuss this open question in more detail. 

\subsection{The Betti number problem}

The problem of interest in topological data analysis is typically calculating how many holes a clique complex has at dimension $l$ - or determining the $l^{\textrm{th}}$ \emph{Betti number}. 
An immediate consequence of our $\QMA_{1}$-hardness  of the homology problem is the following: 
\begin{theorem}[]\label{BettiExact} 
 Given a graph $G$ on $n$ vertices the problem of computing the $l^{\textrm{th}}$ Betti number of its clique complex is $\#\P$-hard for $l=\Omega(n)$.
\end{theorem}

\begin{proof} 
To prove \cref{thmQMA1} we demonstrated how to construct a graph $G$ such that the Betti number $\beta_n$ is equal to the number of satisfying assignments to a quantum $4$-$\SAT$ instance.
The problem of counting satsifying assignments for quantum $k$-$\SAT$ is known to be $\#\P$-hard \cite{Cade:2021jhc,Brown_2011}. 

\end{proof}

\section{Simplicial Complexes as SUSY many-body systems}
\label{sec:connection to SUSY}

We have shown that the complexity of the homology problem, a canonical problem in classical topology, is captured by quantum complexity classes. The technical reduction from \cref{sec:QMA1hard}, however, does not explain {\it why} this is the case at a deeper level. As already mentioned in the Introduction, the reason underlying this was suggested in \cite{Crichigno:2020vue,Cade:2021jhc}. Indeed, as pointed out in these references, there is a deep reason why one should expect the  problem of homology to be quantum mechanical,  which dates back to the pioneering work of Witten \cite{Witten:1982im}. In this paper Witten considered quantum mechanical systems with supersymmetry, a special class of systems with a symmetry relating bosonic states to fermionic states. It was then shown that states with $E=0$ in such systems are in 1-to-1 correspondence with elements of homology. Although Witten focused on {\it continuous} quantum mechanical systems, this relation also holds for discrete quantum mechanical systems of interest here.  Here we spell out this relation in detail for the case of general simplicial complexes and discuss the special case of the clique/independence complex.

\subsection{SUSY quantum mechanics}

The Hilbert space $\cH$ of {\it any} quantum mechanical system can be decomposed as $\cH=\cH_B\oplus \cH_F$, where $\cH_B$ is the space of bosonic states $\ket{\psi_B}$ and  $\cH_F$ the space of fermionic states $\ket{\psi_F}$. These are formally distinguished by the action of the ``parity operator'' $\cP$, acting as
\equ{
\cP\ket{\psi} =\begin{cases} +\ket{\psi} & \text{if $\ket{\psi}\in \cH_B$} \\ -\ket{\psi} & \text{if $\ket{\psi}\in \cH_F$}  \end{cases}
}
Bosonic and fermionic states correspond to states which are symmetric or antisymmetric  under the exchange of particles, respectively.

The parity operator can also be used to define  the notion of bosonic and fermionic {\it operators}. By definition, an operator $\cO_B$ is bosonic if it commutes with the parity operator,  $[\cP,\cO_B]=0$. An operator $\cO_F$ is fermionic if it anticommutes with the parity operator, $\{\cP,\cO_F\}=0$.

By definition, a quantum mechanical system is supersymmetric if there exists an operator $\cQ$ sending bosonic states into fermionic states and vice versa, and squaring to zero:\footnote{Technically, this is called an $\cN=2$ SUSY quantum mechanics (see \cite{Witten:1982im} and e.g. \cite{Cade:2021jhc} for more details). It is one of the simplest examples of supersymmetry, which can also be defined in the context of quantum field theories and (quantum) gravity. A deeper way to understand supersymmetry in these more general contexts is by formulating it as a fermionic symmetry of spacetime. }
\equ{
\cQ: (\cH_B,\cH_F)\to (\cH_F,\cH_B)\qquad \qquad \cQ^2=0\,.
}
The operator $\cQ$ is known as the ``supercharge'' of the system. By consistency of the definitions above the supercharge must anticommute with the parity operator and is thus a fermionic operator. 
\begin{figure}
\begin{center}
\includegraphics[scale=1.6]{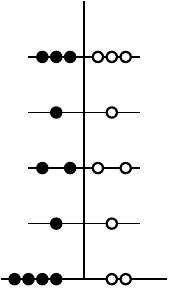}
\put(-130,50){$E>0$}
\put(-130,0){$E=0$}
\caption{The spectrum of supersymmetric quantum mechanics. Black dots represent bosonic states and white ones fermionic states. All states with $E>0$ come in boson/fermion pairs.  Supersymmetric ground states ($E=0$), if there are any, are not necessarily paired.}
\label{fig:spectrum}
\end{center}
\end{figure}
By definition, the Hamiltonian of the system is given by 
\equ{
H:=\{\cQ,\cQ^\dagger\} = \cQ \cQ^\dagger +\cQ^\dagger \cQ\,.
}
Note that $H  $ sends bosonic states into bosonic states and fermionic states into fermionic states and is therefore a bosonic operator.  Note that due to the nilpotency condition
\begin{equation} \label{eq:nilpotency}
[\cQ,H]=[\cQ^\dagger,H]=0
\end{equation}
and thus $\cQ$ and $\cQ^\dagger$ are symmetries of the system. They are not traditional symmetries since they are generated by {\it fermionic} operators, rather than the more standard symmetries such as translational or rotational symmetry which are generated by bosonic operators.

A few remarkable properties of supersymmetric Hamiltonians follow directly from the definitions above. The first, is that the Hamiltonian is positive semidefinite:
\equ{
\bra{\psi} H \ket{\psi} = \bra{\psi} ( \cQ \cQ^\dagger +\cQ^\dagger \cQ) \ket{\psi} = \abs{\cQ\ket{\psi}}^2\abs{\cQ^\dagger\ket{\psi}}^2 \geq 0\,.
}
Thus, the lowest possible energy is $E=0$. Given a particular supersymmetric system, an important question is whether the ground state of the system has $E=0$ or $E>0$. If the ground state energy is positive, the system is said to break supersymmetry ``spontaneously.'' It is in general a difficult question to determine whether a SUSY systems preserves or breaks supersymmetry. Note that  a state $\ket{\Omega}$ has $E=0$ if and if
\equ{
E=0: \qquad \cQ\ket{\Omega}=\cQ^\dagger \ket{\Omega}=0\,.
}
Such states, if they exist  are called SUSY ground states. 

Consider states with $E>0$. Then, it is easy to see that for every bosonic state with energy $E$, there is a corresponding fermionic state with the same energy. Indeed, let $H\ket{\psi_B}=E \ket{\psi_B}$, then clearly one of $Q\ket{\psi_B}$ and $Q^\dagger\ket{\psi_B}$ is non-zero, and by \cref{eq:nilpotency} will have the same energy, $E$.

\subsection{Simplicial complexes as SUSY many-body systems}

Consider a simplicial complex $K$ with $n+1$ vertices. Recall an oriented $p$-simplex  is an ordered $(p+1)$-subset of the vertices, $\sigma^{(p)}=[x_0,x_1, \cdots, x_p]$. One may want to denote such a simplex as the $n$-dimensional array $(z_0,z_1,\cdots,z_n)$, where each $z_i\in \{0,1\}$ indicates if the 0-simplex $[x_i]$ is not present or present in $\sigma^{(p)}$, respectively, and $\sum_{i=0}^n z_i=p+1$. However, this representation does not carry information about the orientation of the original simplex. To remedy this we instead  identify simplices  with  states in the Fock space of $n+1$ fermions, as
\equ{\label{mapsigmaf}
\sigma^{(p)} \leftrightarrow  (a_0^\dagger)^{z_0}(a_1^\dagger)^{z_1}\cdots (a_n^\dagger)^{z_n} \ket{0}
}
where  $\ket{0}$ is the state with no fermions and the operators $a_i$ and their conjugates $a_i^\dagger$ are the standard set of operators satisfying the anticommutation relations
\equ{
\{a_i,a_j^{\dagger}\}=\delta_{ij}\,, \qquad \qquad \{a_i,a_j\}=\{a_i^{\dagger},a_j^{\dagger}\}=0\,.
}
Permuting the vertices the LHS of \eqref{mapsigmaf} is correctly accounted for by the exchange of the anticommuting operators on the RHS.  The fermion number operator is given by $F=\sum_{i=0}^n a_i^\dagger a_i$. The annihilation operators $a_i$ decrease the fermion number by 1 and the creation operators $a_i^\dagger$ increase it by 1. Note that the states in \eqref{mapsigmaf} have $F=\sum_{i=0}^n z_i=p+1$. A $p$-simplex is then identified with a state with fermion number $F=p+1$.

The simplicial boundary operator \eqref{defdel} is represented very easily in fermionic Fock space:
\equ{
\partial=\sum_{i=0}^n a_i\,.
}
The signs appearing in \eqref{defdel} are automatically accounted for by the anticommuting nature of the fermionic operators as well as $\partial^2=\sum_{i<j} \{a_i,a_j\}=0$. It would be tempting to identify $\partial^\dagger$ with $\sum_{i=0}^n a_i^\dagger$ but this would not be generally correct, as some $a_i^\dagger$  may create simplex not in $K$. We thus define 
\equ{
\hat \partial = P \partial P\,,\qquad  \qquad \hat \partial^\dagger = P \partial^\dagger P\,,
}
where $P$ is the projector onto elements in $K$, which we assume commutes with the fermion number operator $[F,P]=0$ so that $\hat \partial: C_p\to C_{p-1}$ and  $(\hat \partial)^2=(\hat \partial^\dagger)^2=0$.\footnote{The projectors $P$ are strictly speaking not required in $\partial$ as long as the operator is acting on the space of allowed simplices, as the boundary of a simplex is always contained in the simplicial complex. However, we include these so that $\partial$ and $\partial^\dagger$ are manifestly Hermitian conjugates of each other. }   The supercharges are then identified as\footnote{It is customary to define the supercharge as having fermion number $+11$ and its Hermitian conjugate $-1$, consistent with the identification below,  but the choice $Q=\partial$ would be equally valid.}
\begin{equation}\label{eq:supercharges}
Q= P \partial^\dagger P\,,\qquad \qquad Q^\dagger= P \partial P\,,
\end{equation}
and the SUSY Hamiltonian is  
\equ{
H=\sum_{i,j} P a_i^\dagger P \ a_j P+ P a_i P a_j^\dagger P\,.
}
Note that the locality of $H$ depends on the projectors $P$. Although these are diagonal in the computational basis they could generically be non-local in the fermionic modes.  Let $\cH_{E=0}$ be the SUSY groundspace of the system. Since the Hilbert space is graded by the fermion number one has 
\equ{
\cH_{E=0}= \bigoplus_{p} \cH_{E=0}^{(p)}\,.
}
where $\cH_{E=0}^{(p)}$ are the subspaces with fixed fermion number $F=p$.  Each subspace is in 1-to-1 correspondence with the reduced homology groups
\equ{
\cH_{E=0}^{(p)}\cong \tilde H_{p-1}(K,\Bbb C)\,,
}
where we have set $\Bbb F=\Bbb C$ as general elements in the Hilbert space will have complex coefficients. We illustrate this in the case of the independence complex next. 

\subsection{Independence homology and the fermion hard-core model}

As discussed, the independence complex is a particular class of simplicial complexes specified by a graph $G$. The corresponding many-body systems models are known as the fermion hard-core model introduced in \cite{Fendley_2003}. The Hilbert space is given by all configurations of fermions on a graph $G$ subject to the condition that no two fermions can occupy neighboring vertices (the hard-core condition). The Hilbert space therefore corresponds to the independent sets of $G$. The supercharges are given by
\equ{
Q= \sum_i  P_i  a_i^\dagger\,,\qquad \qquad Q^\dagger= \sum_i  \bar a_i P_i\,,\qquad \qquad P_i = \prod_{j|(i,j)\in E}(1-\hat n_j)\,,
}
where $\hat n_i=a_i^\dagger a_i$ and the $P_i$ ensure that the operators act within the space of independent sets. These are a special case of the supercharges in \cref{eq:supercharges}, where one uses the fact that for this class of models we have $a_iP =P_ia_i$.

By some simple manipulations the Hamiltonian can be brought into the form
\equ{
H=\sum_{(i,j)\in E}P_ia_i^\dagger a_j P_j +\sum_{i\in V} P_i\,.
}
Since the Hilbert space of the theory is spanned by independent sets of $G$ and the supercharge $Q^\dagger$ acts on this space as the boundary operator $\partial$, the space of supersymmetric ground states is isomorphic to the homology of the independence complex, as observed in \cite{Jonsson:2010,Huijse:2008}.  This relation has been used in a number of beautiful works both in the condensed matter physics literature \cite{beccaria2005exact,fendley2003lattice,fendley2003lattice,fendley2005exact,huijse2008superfrustration,huijse2008charge} as well in the math literature \cite{adamaszek2013special,huijse2010supersymmetry,jonsson2009hard,jonsson2006hard,baxter2011hard,bousquet2008independence,engstrom2009upper,fendley2004hard,csorba2009subdivision}, in many cases leading to exact determination of the SUSY groundspace for simple graphs $G$, such as the square or other regular lattices. The results of \cref{sec:reduction}, however, indicate that there can be no efficient algorithm to determine the groundspace of the fermion hard-core model on a generic graph $G$ (even when restricted to graphs of constant degree -- see \cref{app:locality}.

Let us illustrate this explicitly for the graphs constructed in the reduction of \cref{sec:reduction}. Consider a single triangle $G_1$. Due to the hard-core condition, at most one fermion can occupy the triangle and the Hilbert space is spanned by the four states:
\equ{
\ket{b}=\ket{000}\,,\qquad \ket{f_1}=\ket{100}\,, \qquad \ket{f_2}=\ket{010}\,, \qquad \ket{f_3}=\ket{001}\,,
}
where $\ket{b}$ is a bosonic state with no fermions and the $\ket{f_i}$ are three fermionic states, with a fermion at site $i=0,1,2$. The supercharge is given by
\equ{
Q_1= \sum_{i=0}^2 a_i^\dagger P_i \,.
}
Note that the vacuum state $\ket{b}$ cannot be a supersymmetric ground state since acting with the supercharge gives $Q_1\ket{b}=\ket{f_0}+\ket{f_1}+\ket{f_2}=:\ket{f}\neq 0$ and thus $(\ket{b},\ket{f})$ form a supersymmetric doublet. Indeed, one can easily calculate that the energy of the multiplet is $E=1$. 
This leaves two fermionic states unpaired, which must therefore have $E=0$ and span a two-dimensional SUSY groundspace. 
These correspond to the two non-trivial homology classes outlined in \cref{sec:gadgets}.

\begin{figure}
\begin{center}
\includegraphics[scale=1]{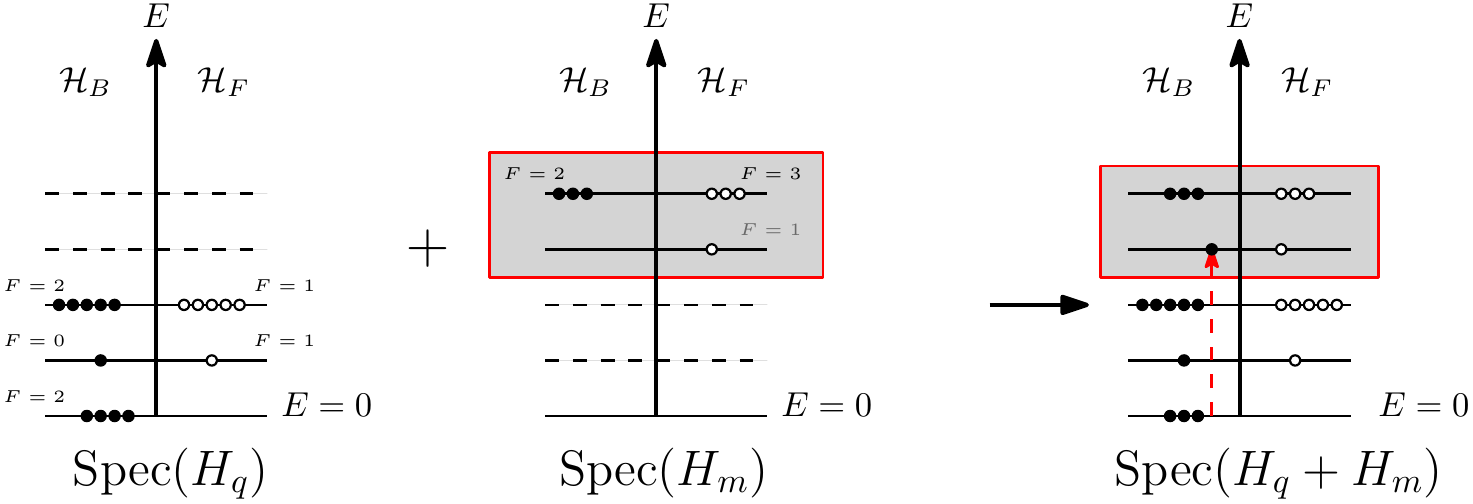}
\caption{Lifting a single SUSY ground state by introducing a mediator vertex connecting to two triangles as in \cref{eq:H}. $H_q$ is a supersymmetric Hamiltonian, so all non-zero eigenavlues are paired up. The Hamiltonian $H_m$ is not supersymmetric, and includes unpaired non-zero eigenvalues. When the two Hamiltonians are added together an unpaired eigenvalue from $H_m$ can pair up with a zero eigenvalue from $H_q$, lifting one state from the SUSY groundspace, as desired. }
\label{fig:lifting}
\end{center}
\end{figure}
Consider now two triangles $G_2$ and supercharge $Q_2= \sum_{i=0}^5 a_i^\dagger P_i$. The system has a total of $1+6+9$ states with $F=0,1,2$, respectively and two SUSY ground states at $F=2$. Now, consider adding a mediator vertex $m$ with a staggering parameter $\lambda_m$. The introduction of the new vertex incorporates $9=1+4+4$ new states with $F=1,2,3$, respectively,  into the Hilbert space.  The supercharge is given by
\equ{
Q= Q_2+ Q_m\,,
}
where $Q_m=a_m^\dagger P_m$
and the full Hamiltonian is now given by:
\begin{equation}\label{eq:H}
H = H_q + H_m
\end{equation}
where
\equ{H_q = \{Q_2,Q_2^\dagger\} \,, \qquad H_{m}= \{Q_2,Q_m^\dagger\} + \{Q_m,Q_2^\dagger\} + \{Q_m,Q_m^\dagger\} 
}
The first term in \cref{eq:H} is the (manifestly supersymmetric) Hamiltonian corresponding to the triangles with no mediators present.
Therefore, as outlined above, all of its eigenvalues are either zero, or come in supersymmetric pairs.
The additional term $H_m$ by itself does not correspond to a supersymmetric Hamiltonian. 
Therefore it can have unpaired non-zero eigenvalues. 
When adding the two Hamiltonians together the resultant Hamiltonian is non-zero. 
In order to lift one of the zero eigenvalues of $H_q$ to a non-zero eigenvalue it needs to pair up, and it can do this because $H_m$ has unpaired non-zero eigenvalues (see \cref{fig:lifting} for an illustration).

Now consider $n$ disconnected triangles $G_n$. The full system is in a SUSY ground state only if each triangle is in a SUSY ground state and thus the only SUSY ground states are found at $F=n$ and there are $2^n$ of them, consistent with 
\equ{
\tilde H_{n-1}(\Sigma_n)= (\Bbb C^2)^{\otimes n}\,.
}
In order to implement the projectors needed for $\Hbravyi$ we introduce mediator vertices to mediate interactions between the triangles.
Certain states then acquire a non-zero energy and are thus pushed out of the SUSY groundspace. 

Note we do not claim that the quantum $k$-$\SAT$ Hamiltonian $H=\sum_a \Pi_a$ can be interpreted as a supersymemtric Hamiltonian, only that there is a supersymmetric Hamiltonian with the same dimensional ground space.
The excited spectrum will generically be different.

\section{Outlook}
\label{sec:outlook} 

It may appear surprising that the problem of simplicial homology, a classical problem in topology, could be characterised in terms of quantum complexity classes. 
Indeed, this may go some way to explaining why the complexity of this fundamental problem in computational topology has remained an open question for so long. 
However, as outlined in \cref{sec:connection to SUSY} this problem can be seen to be quantum mechanical when considering the perspective of supersymmetric quantum many-body systems. 

\

We argue that the results reported above, and the more general relationship between simplicial homology problem and quantum mechanics, suggests that computational homology is a natural place to look for provable quantum speed-ups with potentially myriad applications. There is work to do to further support this claim. 
Closing the gap between the $\QMA_1$-hardness result we have shown for \homologyProb\ and containment in $\QMA$ for \PromisehomologyProb\ is an obvious area for further research.
As already stated, it is unlikely that a decision problem will be complete for $\QMA_1$, so a more promising direction would be demonstrating that the promise problem is hard for $\QMA_1$ \cite{Watrous:2008}.
While it may be possible to show that the hard-instances we constructed in the proof of \cref{thmQMA1} inherit the promise of $\QMA_1$ in the form of a promise gap of the corresponding combinatorial Laplacian, it is not obvious that this will be the case.
This is because in our reduction we have not encoded the $\QMA_1$ circuits directly into the combinatorial Laplacian, but instead into the homology of the boundary operator. 
While there is a 1-to-1 correspondence between elements in the kernel of the Laplacian and elements of homology, they are not exactly the same; the latter have to be harmonic representatives of the homology class.\footnote{This means that the element of homology of $\partial$ should also be an element of cohomology of $\partial^\dagger$, which fixes the representative.  } As a consequence, the correspondence does not straightforwardly imply the necessary promise gap.\footnote{We don't rule out the possibility of proving that the construction presented here does inherit the promise gap from quantum $k$-$\SAT$. It may be possible using additional techniques from homological algebra (such as spectral sequences). However the lack of a direct map between the combinatorial Laplacian and the Hamiltonian from quantum $k$-$\SAT$ does raise the possibility that the construction outlined here does not satisfy the promise gap.}
As such, the next step will likely require new reductions.
However, these new reductions will almost certainly use the techniques we have introduced here.

\

While the homology problem is $\QMA_1$-hard, and thus intractable even for quantum computers, this result suggests the possibility of quantum advantage in certain {\it approximations} to  homology. Indeed, one of the motivations for this work is the problem of estimating normalized ``Betti'' numbers for clique complexes introduced  \cite{lloyd2016quantum} where it was also shown that, under certain assumptions, the problem is contained in $\BQP$. It was shown that the version of this problem for {\it general} chain complexes is $\DQC$-hard \cite{Cade:2021jhc} and thus cannot be dequantized $\DQC \subseteq \BPP$. Whether this remains the case for the case of clique complexes remains as an open problem. Although our results here do not directly address this, we hope that some of the techniques developed here could shed light on this question. It is worth noting that it is not clear at the moment if certain assumptions made in \cite{lloyd2016quantum} are realistic and further work is needed to establish practical utility of this algorithm. See \cite{gyurik:2020,akhalwaya2022representation} for relevant work. For recent experimental implementations of this quantum algorithm to study Betti numbers of the CMB see \cite{IBMQ}.

\

Taking a broader perspective, the main message we hope to convey with this work is that the intersection of topology, many body physics and complexity theory is a rich arena in which  problems with potential quantum advantage can be identified, and concrete tools from Hamiltonian complexity can be harnessed to study this.  The results reported here illustrate the confluence of ideas and technical tools necessary to make further progress and we believe there is much work ahead. 

\

Even more broadly, we suggest that a promising strategy for identifying new computational problems with {\it provable} quantum speedups and wider applications, outside of chemistry or material science, is to focus on many-body Hamiltonians with special ``mathematical structures.'' Supersymmetric Hamiltonians, in particular, reveal that seemingly classical problems in computational topology are intrinsically quantum mechanical. It is possible that Hamiltonians with other rich mathematical structures may lead to similar insights in other areas of applied mathematics.

\section*{Acknowledgements}
We are grateful to Chris Cade, Toby Cubitt, Sev Gharibian, David Gosset, Casper Gyurik, Johnny Nicholson, Stephen Piddock, Anupam Prakash, Kareljan Schoutens,  and David Tennyson for  useful discussions and especially Irina Kostitsyna for her geometric insights and Jeff Carlson and Eric Samperton for enlightening discussions on homology. MC is supported by the European Union Horizon 2020 Research Council grant 724659 Massive-Cosmo ERC-2016-COG, by the Simons Foundation award ID 555326 under the Simons Foundation Origins of the Universe initiative, Cosmology Beyond Einstein’s Theory, and the STFC grant ST/T000791/1. MC thanks QC Ware Corp., Palo Alto, for hospitality while this work was completed. TK is supported by the Spanish Ministry of Science and Innovation through the ``Severo Ochoa Program for Centres of Excellence in R\&D'' (CEX2019-00904-S) and PID2020-113523GB-I0 and started the work while being supported by the ESPRC through the Centre for Doctoral Training in Delivering Quantum Technologies (grant EP/L015242/1).


\appendix 

\section{Quantum $4$-$\SAT$ technical details} 
\label{sec:app-projectors}

\subsection{Mapping from quantum $4$-$\SAT$ to a local Hamiltonian}
The construction in \cite{bravyi2011efficient} maps a quantum circuit, $U=U_L\cdots U_2U_1$, $U_j \in \mathcal{G}$ operating on $N$ qubits into the ground state of a local Hamiltonian $H(U)$ acting on Hilbert space $\mathcal{H}_{\text{clock}} \otimes \mathcal{H}_{\text{comp}}$ where $\mathcal{H}_{\text{clock}}=(\mathbb{C}^4)^{\otimes L}$ and $\mathcal{H}_{\text{comp}}=(\mathbb{C}^2)^{\otimes N }$. 
The ground state of $H(U)$ is a `history state' \cite{Kitaev2002} of the form:
\equ{
\ket{\Omega(\psi_{wit})} = \frac{1}{T}\sum_{t=1}^L \ket{t}\ket{\psi_t}
}
where $\ket{\psi_t} = \Pi_{i=1}^t U_i \ket{\psi_{wit}}$ and ${\ket{t}}$ is an orthonormal basis for $\mathcal{H}_{\text{clock}}$.

A single clock particle in \cite{bravyi2011efficient} has four possible basis states - $\ket{u} = \ket{00}$ - unborn, $\ket{a_1} = \ket{01}$ - active phase one, $\ket{a_2} = \ket{10}$ - active phase two and $\ket{d} = \ket{11}$ - dead.  
As time `passes' each clock particle evolves from the unborn state, through the two active stages, and eventually to the dead state. 
`Legal' clock states are defined as those obeying certain constraints \cite{bravyi2011efficient}:
\begin{enumerate}
    \item The first clock particle is either active or dead
    \item The last clock particle is either unborn or active.
    \item There is at most one active clock particle 
    \item If clock particle $j$ is dead then all particles $k$ satisfying $1\leq k \leq j$ are also dead
\end{enumerate}

In \cite{bravyi2011efficient} it is shown that the ground space of the Hamiltonian $\Hclock = \sum_{j=1}^6 \Hclock^{(j)}$ is spanned by legal clock states, where:
\begin{equation}
\Hclock^{(1)} = \ket{u}\bra{u}_1
\end{equation}
\begin{equation}
\Hclock^{(2)} = \ket{d}\bra{d}_L
\end{equation}
\begin{equation}
\Hclock^3 = \sum_{1\leq j \leq k \leq L} \left(\ket{a_1}\bra{a_1} + \ket{a_2}\bra{a_2} \right)_j \otimes \left(\ket{a_1}\bra{a_1} + \ket{a_2}\bra{a_2} \right)_k
\end{equation}
\begin{equation}
\Hclock^{(4)} =\sum_{1\leq j \leq k \leq L} \left(\ket{a_1}\bra{a_1} + \ket{a_2}\bra{a_2}+ \ket{u}\bra{u} \right)_j \otimes \left(\ket{d}\bra{d}  \right)_k
\end{equation}
\begin{equation}
\Hclock^{(5)} = \sum_{1\leq j \leq k \leq L} \left(\ket{u}\bra{u} \right)_j \otimes \left( \ket{a_1}\bra{a_1} + \ket{a_2}\bra{a_2}+  \ket{d}\bra{d}  \right)_k
\end{equation}
\begin{equation}
\Hclock^{(6)} = \sum_{1 \leq j \leq L-1} \ket{d}\bra{d}_j \otimes \ket{u}\bra{u}_{j+1}
\end{equation}

The computational qubits can be divided into input data and witness registers, $R_{in}$ and $R_{wit}$ such that $|R_{in}| + |R_{wit}| = N$. The input qubits are initialised into the all zero state at time zero via:
\equ{
\Hin = \ket{a_1}\bra{a_1}_1 \otimes \left( \sum_{b \in R_{in}} \ket{1}\bra{1}_b \right)
}
The `propagation' Hamiltonian is defined as:
\begin{equation}
\Hprop = \sum_{t=1}^L  \left(\Hpropt + \Hpropt' \right)
\end{equation}
where:

\begin{equation}
\Hpropt = \frac{1}{2}\left[\left(\ket{a_1}\bra{a_1} + \ket{a_2}\bra{a_2} \right)\otimes \identity - \ket{a_2}\bra{a_1} \otimes U_t - \ket{a_1}\bra{a_2} \otimes U_t^\dagger\right]
\end{equation}
\begin{equation}
\Hpropt' = \frac{1}{2}\left(\ket{a_2,u}\bra{a_2,u} + \ket{d,a_1}\bra{d,a_1} -\ket{d,a_1}\bra{a_2,u} - \ket{a_2,u}\bra{d,a_1} \right)
\end{equation}

It can easily be checked that the zero energy ground space of:
\equ{
H(U) =\Hin + \Hclock + \Hprop }
is spanned by computational history states of the form:
\equ{
\ket{\Omega(\psi_{wit})} = \sum_{t=1}^L\left(  \ket{d}^{\otimes t-1} \ket{a_1}_t \ket{u}^{\otimes L-t} \otimes \ket{Q_{t-1}} +\ket{d}^{\otimes t-1} \ket{a_2}_t \ket{u}^{\otimes L-t} \otimes \ket{Q_{t}}    \right)
}

where $\ket{Q_0} = \ket{0}^{\otimes R_{in}} \otimes \ket{\psi_{wit}}$, and $\ket{Q_t} = U_t \ket{Q_{t-1}}$, $t \in [1,L]$. 
Therefore the ground states of $H(U)$ encode the computational history of the circuit $U = U_L...U_1$.

The final step is to penalise computational histories where the circuit rejects the witness, this is achieved using the projector:
\begin{equation}
\Hout = \ket{a_2}\bra{a_2}_L \otimes \left( \sum_{b \in R_{out}} \ket{1}\bra{1}_b \right)
\end{equation}
This gives energy to computational history states $\ket{\Omega(\psi_{wit})}$ if the circuit $U$ rejects the witness $\ket{\psi_{wit}}$ with non-zero probability. 

So the final Hamiltonian:
\equ{
H = \Hin + \Hclock + \Hprop + \Hout
}
has a zero energy ground state iff there exists a witness $\ket{\psi_{wit}}$ such that circuit $U$ accepts with probability one. 

\subsection{Projectors needed for quantum $4$-$\SAT$} 

Here we outline the projectors needed to implement the clock construction of \cite{bravyi2011efficient} using our universal gate set $\mathcal{G}$.

For each $\Hin$ only a single projector is needed.
\begin{equation}
\Hin = \ket{011}\bra{011}
\end{equation}
This is a three qubit projector of rank 1.
$\Hout$ is equivalent, up to permuting the qubits involved in the interaction:
\begin{equation}
\Hout = \ket{101}\bra{101}
\end{equation}

Implementing $\Hclock$ requires six projectors.
$\Hclock^{(1)}$ and $\Hclock^{(2)}$ are a two qubit projectors with rank 1:
\begin{equation}
\Hclock^{(1)} = \ket{00}\bra{00}
\end{equation}

\begin{equation}
\Hclock^{(2)} = \ket{11}\bra{11}
\end{equation}

The remaining terms in $\Hclock$ each act on 4 qubits. $\Hclock^{(3)}$ and $\Hclock^{(5)}$ each have rank 4, $\Hclock^{(4)}$ has rank 3, and $\Hclock^{(6)}$ ha rank 1:
\begin{equation}
\Hclock^{(3)} = \ket{0101}\bra{0101} + \ket{0110}\bra{0110} + \ket{1001}\bra{1001} +\ket{1010}\bra{1010}
\end{equation}

\begin{equation}
\Hclock^{(4)} = \ket{0111}\bra{0111} + \ket{1011}\bra{1011} + \ket{0011}\bra{0011}
\end{equation}

\begin{equation}
\Hclock^{(5)} = \ket{0001}\bra{0001} + \ket{0010}\bra{0010} + \ket{0011}\bra{0011}
\end{equation}

\begin{equation}
\Hclock^{(6)} = \ket{1100}\bra{1100}
\end{equation}

For $\Hprop$ we have to consider two projectors for $\Hpropt$ and one for $\Hpropt'$.
The projector for propagation under the pythagorean gate is a three qubit projector of rank 2 given by:

\begin{equation}
\begin{split}
\Hpropt(U_{\mathit Pyth.})=\frac{1}{2}[\left(\ket{01}\bra{01} + \ket{10}\bra{10} \right)\otimes \identity \\
- \frac{1}{5\sqrt{2}} \left( 3\ket{100}\bra{010}-4\ket{100}\bra{011}+4\ket{101}\bra{010}+3\ket{101}\bra{011}  \right) \\
- \frac{1}{5\sqrt{2}}\left( 3\ket{010}\bra{100} +4\ket{010}\bra{101} -4\ket{011}\bra{100} +3\ket{011}\bra{101} \right) ]
\end{split}
\end{equation}

While the projector for evolution under $\CNOT$ is a four qubit projector of rank 4:

\begin{equation}
\begin{split}
\Hpropt(CNOT)= \frac{1}{2}[\left(\ket{01}\bra{01} + \ket{10}\bra{10} \right)\otimes \identity \\
- \left( \ket{1000}\bra{0100}+\ket{1001}\bra{0101} +\ket{1010}\bra{0111} +\ket{1011}\bra{0110}   \right)\\
- \left( \ket{0100}\bra{1000}+ \ket{0101}\bra{1001}+ \ket{0110}\bra{1011}+ \ket{0111}\bra{1010}\right)]
\end{split}
\end{equation}

And finally $\Hpropt'$ is a four qubit projector with rank 1:
\begin{equation}
\Hpropt' = \frac{1}{2}\left(\ket{1000}\bra{1000} + \ket{1101}\bra{1101} -\ket{1101}\bra{1000} - \ket{1000}\bra{1101} \right)
\end{equation}

\section{Technical details of gadget constructions} \label{app:gadgets}

\subsection{Algebraic details of the 3-qubit constructions} \label{app:algebraic}
\subsubsection{Classical three qubit projector}

Consider the classical  3-qubit, rank-1, projector
\equ{
\Pi= \ket{000}\bra{000}\,.
}
The 2-dimensional void to fill is then
\equ{
\ket{000}= [x_{1}x_{2}x_{3}]+ [x_{1}a_{3}x_{2}]+ [a_{2}x_{1}x_{3}]+ [x_{1}a_{2}a_{3}]+ [a_{1}x_{3}x_{2}]+ [a_{1}x_{2}a_{3}]+ [a_{1}a_{2}x_{3}]+ [a_{2}a_{1}a_{3}]\,.
}
As alluded to in the main text, this can be visualised as a double pyramid with vertices $\{x_{1,2,3},{a_{1,2,3}}\}$ and whose 2-faces are given by the eight terms above. 
Lifting the state $\ket{000}$ then corresponds to filling in the three-dimensional space bounded by these faces, which can be achieved by adding a single mediator vertex $m$ connecting to all 6 vertices (see \cref{table:gadgets3qubit}), and including the eight 4-simplices:
\begin{equation} \label{eq:class-3-qub}
[x_{1}x_{2}x_{3}m], [x_{1}a_{3}x_{2}m], [a_{2}x_{1}x_{3}m], [x_{1}a_{2}a_{3}m], [a_{1}x_{3}x_{2}m], [a_{1}x_{2}a_{3}m], [a_{1}a_{2}x_{3}m], [a_{2}a_{1}a_{3}m]
\end{equation}
Indeed, one can easily algebraically check again that 
\equ{
\ket{000}=\partial' \ket{\Psi}
}
where $\ket{\Psi}=\ket{000}\wedge [m]$ is the the equal superposition of the eight 4-simplices given in \cref{eq:class-3-qub} and $\partial'=\partial_{x,a,b}+\partial_{m}$.

\subsubsection{The 3 qubit $\CNOT$ gadget}

Consider the entangled state 
\begin{equation}\label{101101s}
\begin{split}
\ket{101}-\ket{010}&= [a_1b_2a_3] + [a_1a_3x_2] + [a_1x_3b_2] + [a_1x_2x_3] + [a_2b_1b_3] + [b_1a_2x_3]+[a_2b_3x_1] \\ & +[a_2x_1x_3]+[a_3b_2x_1] + [a_3x_1x_2] +[b_1x_2b_3] + [b_1x_3x_2] + [x_1b_2x_3] + [x_1b_3x_2]
\end{split}
\end{equation}

As outlined in the main text, we will fill in this cycle using sets of four solid tetrahedra touching at a common mediator vertex.
The idea is to glue the base of each tetrahedron to a face in $\Sigma_3$. We have a total of 14 faces and so consider we 3 mediators $m_{1,2,3}$ and glue these to 12 different faces, as follows: 
\eqss{ \label{eq:simplices}
([b_1a_2x_3]+[b_1x_2x_3]+[a_1b_2x_3]+[a_1x_2x_3])\wedge [m_1]\\ 
([x_1a_2b_3]+[x_1a_2x_3]+[x_1b_2x_3]+[x_1b_2a_3])\wedge [m_2]\\ 
([x_1x_2b_3]+[b_1x_2b_3]+[x_1x_2a_3]+[a_1x_2a_3])\wedge [m_3]\,.
}

Finally, we introduce two additional mediators $m_{4}$ and $m_5$ to cover the  two missing faces:
\equ{\label{eq:last2simplices}
[b_1a_2b_3m_4]\,,\qquad  [a_1b_2a_3m_5]\,.
}
This covers all 14 faces of $\Sigma_3$. 
To fill in the spaces between the tetrahedra we simply connect all the mediators by edges. This defines additional tetrahedra which fill the void. 

One can algebraically check that 
\equ{
\ket{101}-\ket{010} = \partial \ket{\Psi} \,,
}
where $\ket{\Psi}$ is the equal weight superposition of the 14 3-simplices presented in \cref{eq:simplices} and \cref{eq:last2simplices}, along with the following 3-simplices involving 2 mediators:

\begin{equation}
    \begin{gathered}
\left([a_2  x_3]+ [ x_3 b_2] \right) \wedge [m_1m_2] \\
      \left([b_3  x_1]+[x_1  a_3] \right) \wedge [m_2m_3]  \\
   \left([ x_2 b_1]+[a_1  x_2] \right) \wedge [m_1m_3]  \\
  \left([b_1  a_2  m_1] +[a_2  b_3  m_2]+   [b_3  b_1  m_3]\right)  \wedge [m_4]         \\
   \left([ b_2 a_1  m_1]+  [ a_3  b_2 m_2]+  [a_1  a_3  m_3]\right) \wedge[m_5] 
 \end{gathered}
 \end{equation}

the following 3-simplices involving 3 mediators:
\begin{equation}
\begin{gathered}
 \left([a_2  m_2  m_1]+ 
 [b_3  m_3  m_2]+ 
     [b_1  m_1  m_3]\right) \wedge [m_4]\\
       \left( [a_3  m_2  m_3]+
         [b_2  m_1  m_2]+
            [a_1  m_3  m_1]\right) \wedge [m_5]
            \end{gathered}
           \end{equation}

and finally, two 3-simplices which just involve the mediators:
\begin{equation}
[m_2  m_1  m_3  m_5]+
           [m_1  m_2  m_3  m_4]
\end{equation}

Checking that the state given above satisfies $\ket{101}-\ket{010} = \partial \Psi$ is a straightforward (if rather tedious) calculation.
However actually finding the state for this gadget is more involved.
In the case of the 2 qubit, and classical 3 qubit, gadgets, we could find $\Phi$ satisfying $c = \partial(\Phi)$ simply by looking at the simplicial complex - this was possible because the mediators introduced in those gadgets didn't introduce any higher dimensional simplices into the complex, and the filled in cycle was easy to visualise.
As the cycle, $c$, to fill in increases in size and complexity this becomes increasingly difficult, and instead we solve a set of equations to demonstrate both that there exists a $\Phi$ satisfying $c = \partial(\Phi)$, and that $c$ is the only cycle that can be written in this way. 
This is the same set of equations we use to demonstrate that all the gadgets we use fill in only the desired state, and not any extra ones.
Details of the calculations are given in \cref{app:mathematica}.

\subsection{Demonstrating that the gadgets fill in exactly one state} \label{app:mathematica}

As outlined in \cref{sec:gadgets} the homology group of the clique complex of a graph consisting of $n$ disconnected triangles is isomorphic to $\mathbb{C}^{\otimes n}$.
That means that given a basis $v^{(n)}_i$ for $\mathbb{C}^{\otimes n}$ we will find that for all $i$: $\partial_n v^{(n)}_i = 0$ and there does not exist any state $c^{(n+1)}$ satisfying:
\begin{equation} \label{eq:appMath}
\partial_{n+1} c^{(n+1)} = v^{(n)}_i
\end{equation}

Analytically, `filling in cycles' in the clique complex corresponds to breaking the second of these conditions.
The states we are lifting are still cycles, so the first condition still holds.
But by `filling in the cycle' we are constructing a state which the cycle is a boundary of, and hence there now does exist a state $c^{(n+1)}$ satisfying \cref{eq:appMath}.

Therefore, to check that for each of the gadgets we have constructed we lift the desired state, and only that state, we need to check that the only state for which \cref{eq:appMath} has a solution is the state that we are intending to lift. 

So, in Mathematica for a gadget acting on $n$ qubits we solve the equation:
\begin{equation}\label{eq:mathematica}
 \sum_{i,j} \left( k_i \ket{e^{(n)}_i} - b_j \partial_{n+1} \ket{e^{(n+1)}_j}\right) = 0
\end{equation}
for $k_i$ where $\ket{e^{(x)}_j}$ denotes the $j^{th}$ independent set of size $x$ in the graph.
For all the gadgets presented in the paper, we find that (up to an irrelevant global phase) the only state satisfying \cref{eq:mathematica} is the intended state.

The calculations for all the gadgets are presented in the attached Mathematica file \cite{mathematica}.

\subsection{Geometric description of second Pythagorean gadget construction} \label{app:pythag2}

The second Pythagorean gadget, to lift the state:    
\equ{\label{pyth2state}
\ket{\psi_{\mathit Pyth.}}= \frac{1}{5\sqrt{2}}  \left(-5\ket{010}+3\ket{100}-4\ket{101}\right)\,.
}

is designed in a similar way to the first Pythagorean gadget (outlined in \cref{sec:pythag}) but there are some modifications required.

Writing each computational cycle involved explicitly we have
\eqss{
B = \ket{100} =\,& [x_{1}x_{2}x_{3}]+ [x_{1}a_{3}x_{2}]+ [a_{2}x_{1}x_{3}]+ [x_{1}a_{2}a_{3}]+ [b_{1}x_{3}x_{2}]+ [b_{1}x_{2}a_{3}]+ [b_{1}a_{2}x_{3}]+ [a_{2}b_{1}a_{3}]\,, \\
C = \ket{101} =\,& [x_{1}x_{2}x_{3}]+ [x_{1}b_{3}x_{2}]+ [a_{2}x_{1}x_{3}]+ [x_{1}a_{2}b_{3}]+ [b_{1}x_{3}x_{2}]+ [b_{1}x_{2}b_{3}]+ [b_{1}a_{2}x_{3}]+ [a_{2}b_{1}b_{3}]\,. \\
D = \ket{010}=\,&  [x_{1}x_{2}x_{3}]+ [x_{1}a_{3}x_{2}]+ [b_{2}x_{1}x_{3}]+ [x_{1}b_{2}a_{3}]+ [a_{1}x_{3}x_{2}]+ [a_{1}x_{2}a_{3}]+ [a_{1}b_{2}x_{3}]+ [b_{2}a_{1}a_{3}]\,, \\ 
}
We now need to carry out ``simplicial surgery" on these cycles to generate $\ket{\psi_{\mathit Pyth.}}$. 

Applying exactly the same process as in \cref{sec:pythag} quickly runs into trouble.
The $C$ and $D$ cycles only have three edges in common - $[x_1x_2]$, $[x_2x_3]$ and $[x_3x_1]$.
But trying to glue along any two of these edges cannot work to construct the full cycle, since the cycle $C-B$ does not contain the edges $[x_3x_1]$ or $[x_2x_3]$ and the $D-B$ cycle does not include the edge $[x_1x_2]$.
So if we glue together the $C$ and $D$ cycles and then try to subtract of the $B$ cycles you will inevitably subtract off an edge that has been used to glue the cycles together.
   
To get around this we instead glue together three copies of the cycle $[C+D-B]$, one copy of the cycle $C$ and two copies of the cycle $D$.
Clearly this amounts to adding together the right number of copies of each cycle, but crucially these cycles all have the simplex $[x_1x_2x_3]$ in common.
So we can glue along this simplex using the technique outlined in \cref{fig:slits_two}.
We cut open the simplex by duplicating the vertices $x_2$ and $x_3$, and connecting up the vertices within cycles and between cycles to generate the necessary void. 

 \begin{figure}[]
\begin{center}
\centering
\includegraphics[width=0.8\linewidth]{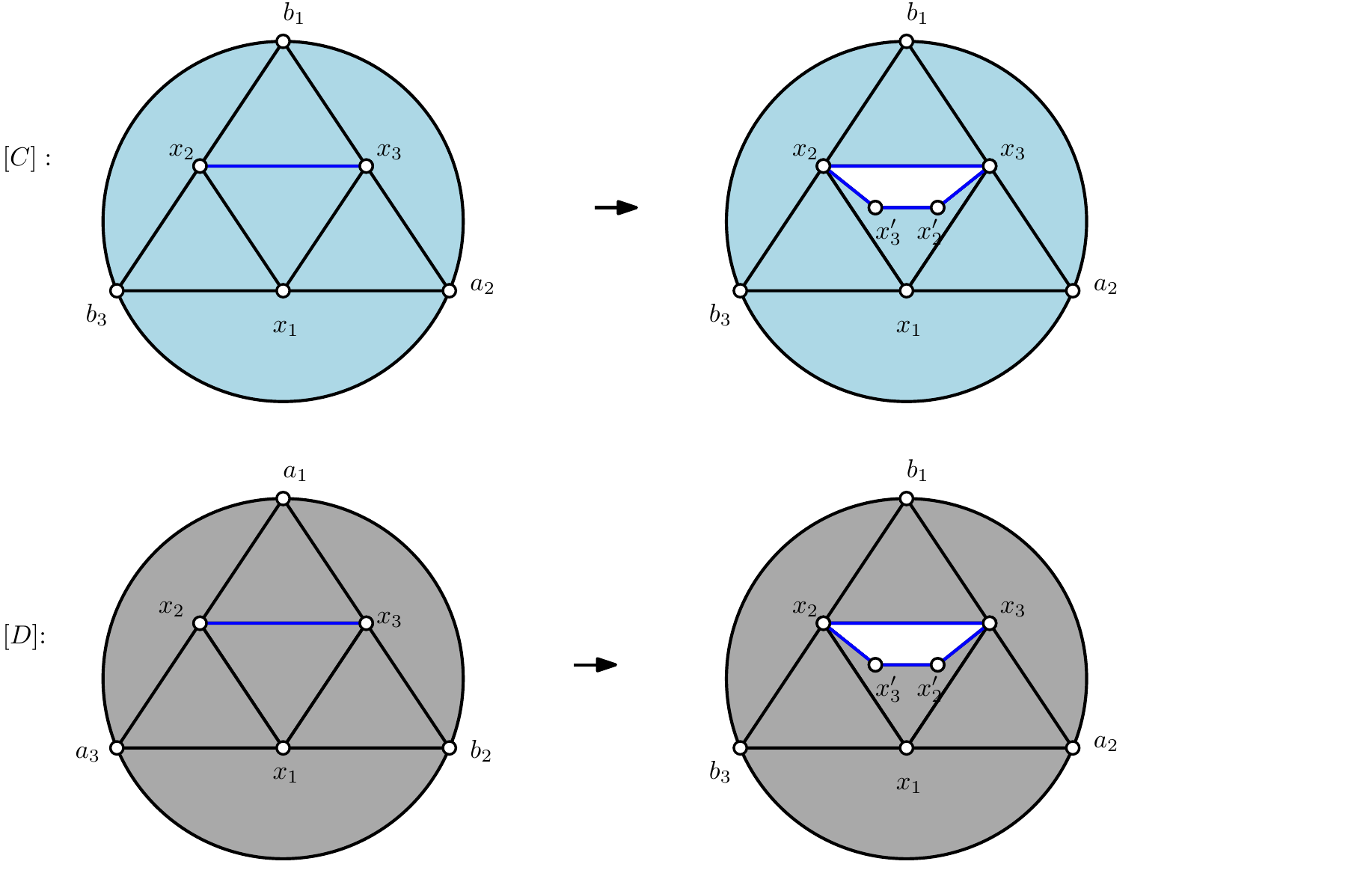}
\caption{The $[C]$, $[D]$ and $[C+D-B]$ cycles are cut open by cutting a slit in the $[x_1x_2x_3]$ triangle by duplicating the points $x_2$ and $x_3$. Then, the open cycles are  connected with each other by gluing through the common one-cycles  $[x_2x_3]+[x_3x_2']+[x_2'x_3']+[x_3'x_2]$, ensuring the orientation is consistent. 
The $[C+D-B]$ cycle is cut open and glued together in the same way - the figure for that cycle is omitted as it is significantly more cluttered than the original cycles.}
\label{fig:slits_two}
\end{center}
\end{figure}

This leads to a closed surface which is bounded by $67$ faces. 
The surface is homeomorphic to $S^2$ and we fill it in using exactly the same procedure outlined in \cref{sec:pythag}. 

This introduces 68 mediator vertices and leads to a simplicial complex $\Sigma_3'$ consisting of $77$ 0-simplices, $489$ 1-simplices, $869$ 2-simplices, $462$ 3-simplices, and $13$ 4-simplices. The Euler characteristic is $1-77+489-869+462-13=-7$, consistent with $H_2(\Sigma_3')=(\Bbb C)^7$. 

As in \cref{sec:pythag}, to construct the graph we take the graph complement of the 1-skeleton of the resulting simplicial complex. 
This gives a graph on $77$ vertices, with $2437$ edges and maximum vertex degree  $\delta=77$.

The explicit procedure for constructing this gadget is carried out step-by-step in the attached Mathematica file \cite{mathematica}.
Details of the calculations carried out to show that this gadget closes the cycle associated to $\ket{\psi_{Pyth.}}$ and only this cycle are given in \cref{app:mathematica}.

\subsection{Technical details of adding projectors} \label{app:adding_projectors}

This appendix outlines why it is in general necessary to connect the mediators from the set $\{m_1\}$ with the mediators from the set $\{m_2\}$ arising from the process outlined in \cref{sec:combining_gadgets} for adding projectors together.
We will use the gadgets for lifting $\ket{\phi_1} = \ket{00}-\ket{11}$ and $\ket{\phi_2}=\ket{01}-\ket{10}$ as an illustrative example.

Implementing the projector $\ket{\phi_1}\bra{\phi_1}+\ket{\phi_2}\bra{\phi_2}$ requires filling in the cycles:
\begin{equation}
\phi_1=[x_2\, a_1]+[a_1\,a_2 ]+[a_2 \, x_1]+[x_1\, b_2]+[b_2\,b_1]+[b_1\, x_2]\,,
\end{equation}
\begin{equation}
\phi_2=[x_2\, a_1]+[a_1\,b_2 ]+[b_2 \, x_1]+[x_1\, a_2]+[a_2\,b_1]+[b_1\, x_2]
\end{equation}

It can be straightforwardly checked that the graph resulting from applying steps 1-3 from the process for adding projectors in \cref{sec:combining_gadgets} to the graphs for the individual projectors from \cref{table:gadgets2qubit} does fill in the required cycles.
However, a simple calculation demonstrates that this graph has Euler Characteristic 1 - this implies that the homology group $H_2$ only contains one element.
It therefore cannot be isomorphic to the ground state of $\ket{\phi_1}\bra{\phi_1}+\ket{\phi_2}\bra{\phi_2}$, which is dimension two. 

This issue arises not because we have accidentally filled in an unintended cycle, but because the mediators added to the graph have caused two previously independent homology classes to be homologous to one another. 

To see why that has happened we will consider the states $\ket{\phi_3} = \ket{00}+\ket{11}$ and $\ket{\phi_4}=\ket{01}+\ket{10}$ - clearly these states span the ground state of the projector that we want to implement, and in the original gadgets they form distinct members of the homology.

To see that they are distinct members of the homology in the original projectors we note that:

\begin{equation}
\phi_3=[x_2\, a_1]+[a_1\,a_2 ]+[a_2 \, x_1]+[x_1\, x_2] + [x_2\, b_1] +[b_1\,b_2]+[b_2\, x_1]+ [x_1\, x_2]\,,
\end{equation}
\begin{equation}
\phi_4=[x_2\, a_1]+[a_1\,b_2 ]+[b_2 \, x_1]+[x_1\, x_2]+[x_2\,b_1]+[b_1\,a_2]+[a_2\, x_1]+[x_1\, x_2]
\end{equation}

And therefore:
\begin{equation}
    \phi_3-\phi_4 = [a_1\,a_2]+[a_2\,b_1]+[b_1\,b_2]+[b_2\,a_1]
\end{equation}

It can be checked that in the original gadgets there is no 3-cycle for which the cycle $\phi_3-\phi_4$ is a boundary, and therefore the two cycles are not homologous to one another.
However, it can be seen from \cref{fig:adding_projectors} that if we try to combine the projectors $\ket{\phi_1}\bra{\phi_1}$ and $\ket{\phi_2}\bra{\phi_2}$ without connecting the mediators in the graph, then in the independence complex we have constructed a 3-cycle which $\phi_3-\phi_4$ is a boundary of, and so the two cycles are homologous to each other.

By connecting the mediators from the set $\{m_1\}$ with the mediators from the set $\{m_2\}$ in the graph, we remove the connections between the mediators in the independence complex - and the two states $\phi_3$ and $\phi_4$ are again members of distinct homology classes, as required. 
\begin{figure}[]
\begin{subfigure}{0.5\textwidth}
\centering
\includegraphics{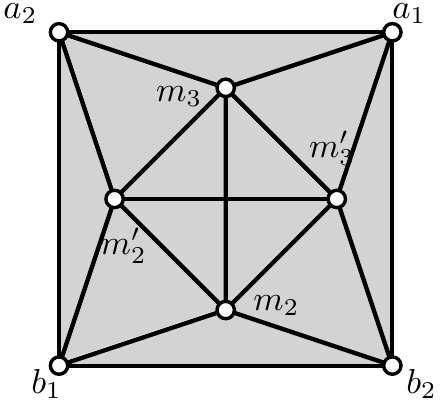}
\caption{Independence complex arising from trying to add projectors without connecting the mediators together in the graph.}\label{fig:wrong}
\end{subfigure}
\begin{subfigure}{0.5\textwidth}
\centering
\includegraphics{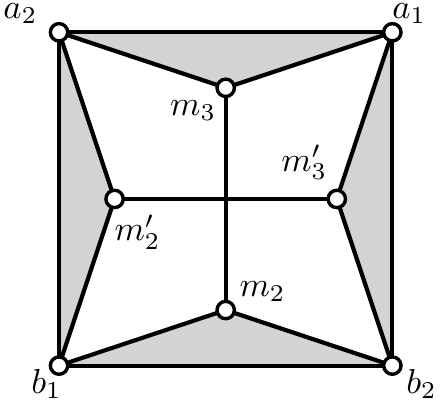}
\caption{Independence complex arising from correctly adding projectors (connecting the mediators together in the graph.)} \label{fig:right}
\end{subfigure}
\caption{The independence complexes that arise from adding projectors together without (\cref{fig:wrong}) and with (\cref{fig:right})  connections between the mediators from different gadgets in the graph. In both figures the un-primed mediator vertices are those arising from the gadget for lifting $\phi_1$, while the primed mediator vertices are those arising from the gadget for lifting $\phi_2$. In both cases superfluous vertices and edges have been omitted from the figure for clarity.}  \label{fig:adding_projectors}
\end{figure}

It is worth noting that connecting the mediators isn't always necessary - for particularly simple gadgets (e.g. two classical gadgets) this issue doesn't arise.
Moreover, in cases where the gadgets overlap on some, but not all, qubits it is generally possible to construct the correct gadget by connecting some, but not all, mediators from the different gadgets.
However, connecting all the mediators from the different gadgets will always work.
So, for clarity in the main text we assume that for all gadgets on overlapping sets of qubits we connect all the mediator vertices arising from one gadget with all the mediator vertices arising from the other gadget. 

\section{Locality of the boundary operator} \label{app:locality}

\subsection{Quantum $(k,s)$-$\SAT$ is $\QMA_1$ hard}

In quantum $k$-$\SAT$ the only constraint on the interaction graph is the number of qubits involved in each interaction (restricted to be no more than $k$).
But retaining locality in the boundary operator when mapping to the homology problem requires also constraining the number of interactions each qubit is involved in.
When we map to the homology problem we will construct the mapping in terms of `lifting' individual states, so in this section we will consider a rank $m$ projector as the sum of $m$ rank 1 projectors.To this end, we define a new problem: \\

\mypromprob{{\sc Quantum $(k,s)$-$\SAT$}}
{A set of $k$-local rank-1 projectors $\Pi_S$ where $S$ are possible subsets of $\{1,\ldots,n\}$ of cardinality $k$, such that each integer in the set $\{1,\cdots,n\}$ appears in no more than $s$ projectors.}{Either there exists a state $\ket{\psi}$ such that $\Pi_S\ket{\psi}=0$ for all $S$ or otherwise $\sum_{S}\bra{\psi}\Pi_{S}\ket{\psi}\geq \epsilon$ for all $\ket{\psi}$, with $\epsilon>\frac{1}{\poly(n)}$. }{Output $\yes$ if the former and $\no$ if the latter.\\}

\noindent We will show that quantum $(k,s)$-$\SAT$ with constant $k,s$ is hard for $\QMA_1$ using a technique from \cite{Aharanov_2004,Oliveira_2005,Zhou_2021}.
The idea is to show that any $\QMA_1$-verification circuit can be efficiently mapped to a circuit where each qubit is involved in a constant number of gates.  The mapping from circuit to Hamiltonian then preserves this property, so that each qubit is acted on by a fixed number of projectors, regardless of system size. 

The utility of this technique will depend on the choice of universal gate set used to define $\QMA_1$. 
It will work for any gate set where the $\SWAP$ gate can be implemented exactly.
It is possible to implement the $\SWAP$ gate exactly using three $\CNOT$ gates, so crucially the hardness proof applies to our choice of gate set $\mathcal{G}$:

\begin{lemma} \label{ksQSAT}
Quantum $(k,s)$-$\SAT$ is $\QMA_1$-hard for $k\geq 4$ and $s\geq 28$.
\end{lemma}

\begin{proof}
Consider a $\QMA_1$-verification circuit consisting of $n$ qubits which are acted on by a circuit $U = \Pi_{t=1}^T U_t$ (for $T = \poly(n)$) consisting of gates taken from the universal gate set $\mathcal{G}$.

We start by showing that this circuit can be mapped to one where each qubit is acted on by at most 3 gates, via a procedure initially used in \cite{Aharanov_2004,Oliveira_2005,Zhou_2021} (see \cref{fig:circuit_sparsification} for a graphical illustration of the process):
\begin{enumerate}
    \item Place all the $n$ qubits on a line, in some arbitrary ordering.
    \item For each gate $U_i \in \{U_1,\cdots,U_T\}$, if $U_i$ acts on a single qubit, or on a pair of nearest neighbour qubits do nothing. If $U_i$ acts on a pair of qubits that are not nearest neighbour then add a sequence of $\SWAP$ gates (implemented as three $\CNOT$ gates) before $U_i$, then reverse the $\SWAP$ gates after $U_i$. For each of the $T$ gates in the circuit, this process adds at most $3n$ gates, so after this step the circuit consists of $T' = O(nT)$ nearest neighbour gates acting on a line of $n$ qubits. We will denote this modified circuit by $U=\Pi_{t=1}^{T'}U'_t$
    \item Take a $n \times T'$ grid of qubits. In each of the $i^{th}$ column of the grid we will implement only the $U'_i$ gate (which acts on nearest neighbour qubits) with all other qubits acted on by the identity gate. After implementing the gate $U'_i$ on the $i^{th}$ column we implement a row of $\SWAP$ gates between the $i^{th}$ column and the $(i+1)^{th}$ column, so that the state of the qubits is mapped along the grid. 
    
    In order to ensure that at each time step a constant number of qubits are acted on for the first time we implement the gates in the following order: in column $i$, start at the top and implement a series of identity gates until reaching the qubit(s) that $U'_i$ acts on, then implement $U'_i$, then continue implementing identity gates on each qubit until we reach the bottom of the column. Then implement $\SWAP$ gates (via a sequence of three $\CNOT$ gates) between the qubits in column $i$ and the corresponding qubits in column $i+1$ starting at the bottom of the column, and working up to the top. 
    
    This circuit acts on $n' = nT' = \poly(n)$ qubits, and consists of $T_\text{fin} = 4(nT')^2 = \poly(n)$ gates (including identity gates), and each qubit in the circuit is involved in at most 7 gates (up to two $\SWAP$ gates each involving three $\CNOT$ gates, and one identity or non-trivial gate).
\end{enumerate}

We will now use the clock construction from \cite{bravyi2011efficient} to construct a mapping from our modified verification circuit to quantum $(k,s)$-$\SAT$. 

First consider the Hamiltonian $\Hprop$.
Each computational qubit is acted on by one projector of the form $\Hpropt$ for each gate it is involved in. 
These projectors are of rank at most 4, and there are at most 7 of them.
So $\Hprop$ can be implemented by acting on each computational qubit with up to 28 rank one projectors.
Each clock qubit is acted on by one projector of the form $\Hpropt$ (rank at most 4) and two projectors of the form $\Hpropt'$ (rank 1). So $\Hprop$ can be implemented by acting on each clock qubit with up to 6 rank one projectors.

The clock Hamiltonian only acts on the clock qubits.
Each clock qubit is acted on by up to three rank 1 projectors, two rank 3 projectors, and one rank 4 projector. 
So $\Hclock$ can be implemented by acting on each clock qubit with up to 13 rank one projectors.

$\Hout$ acts on the final clock qubit and a single computational qubit in the last column of the grid with a single rank one projector.

The final piece of the puzzle is $\Hin$.
As set out in \cite{bravyi2011efficient} $\Hin$ is a sum of rank one projectors, each one acting on the first clock qubit, and the non-witness computational qubits, i.e. all $q \in R_{\text{in}}$. 
However, implementing it in this fashion would lead to the first clock qubit being involved in a number of projectors which scales with the size of the problem since the number of $q \in R_\text{in}$ can scale with $n$.
Following the same procedure as \cite{Oliveira_2005} we can get around this issue.
At each time step in our modified circuit we only act with a single gate, so we don't need to initialise all the input qubits to zero when the clock is at zero.
Instead we can initialise the input qubits to zero immediately before they are acted upon for the first time (possibly by an identity gate).
So we replace the $\Hin$ from \cite{bravyi2011efficient} by:
\equ{
\Hin = \sum_{q\in R_{\text{in}}, q \leq n} \ket{1}\bra{1}_q\otimes \ket{a_1}\bra{a_1}_{c_{t_q}}
}
where $c_{t_q}$ is the time at which qubit $q$ is first acted on by a (possibly trivial) gate.
We only need to add these terms for the first column of qubits, since the qubits in subsequent columns are `dummy' qubits, and they are initialised by the $\SWAP$ gates. 

The only computational qubits that will be acted on by $\Hin$ and $\Hout$ are in the first and last rows of the grid, so are acted on by at most 4 gates (since they are acted on by a single $\SWAP$ gate). 
So the qubits that are involved in the most projectors are computational qubits which are not in the first or last column of the grid. 
Therefore we can map our modified verification circuit to quantum $(k,s)$-$\SAT$ where $k=4$ and $s=28$.

\end{proof}

For the universal gate set $\mathcal{G}$ used in this paper this proof only demonstrates hardness, since we have not shown that quantum $k-\SAT$ is contained within $\QMA_1$ with this gate set.
However, the proof of hardness could also be applied to the gate set $\mathcal{G}'=\{\CNOT,\hat{H},T\}$.
In this case it is straightforward to see that all the projectors needed for the hardness construction have matrix elements in the computational basis of the form $\frac{1}{4}(a+ib+\sqrt{2}c+i\sqrt{2}d)$.
It is shown in \cite[Appendix A]{Gosset_2013} that eigenvalues of projectors of this form can be measured exactly using the gate set $\mathcal{G}'$.
Therefore, quantum $(k,s)-\SAT$ with universal gate set $\mathcal{G}'$ where the set of allowed projectors $\mathcal{P}$ are all 4-local projectors with matrix elements in the computational basis of the form $\frac{1}{4}(a+ib+\sqrt{2}c+i\sqrt{2}d)$ is $\QMA_1$-complete for $k\geq 4$ and $s\geq 28$.

\begin{figure}[]
    \begin{subfigure}[t]{0.45\textwidth}
    \centering
    \includegraphics[width =0.75 \linewidth]{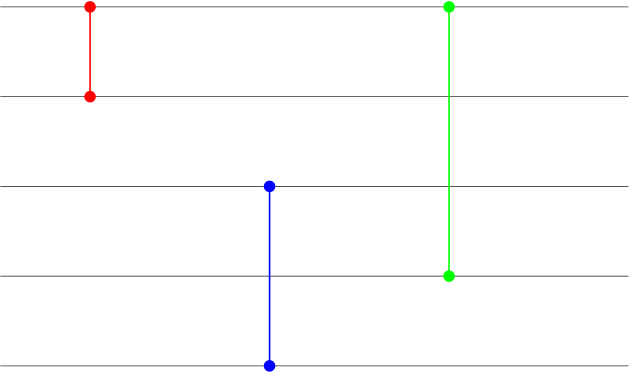}
    \caption{The initial circuit we want to `sparsify' - it consists of three two qubit gates.}
    \label{fig:a}
    \end{subfigure}
        \begin{subfigure}[t]{0.45\textwidth}
        \centering
    \includegraphics[width = 0.75\linewidth]{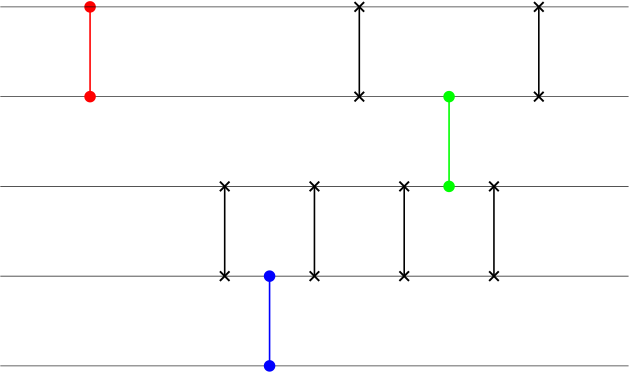}
    \caption{The new circuit after the first step of the sparsification process - the $\SWAP$ gates ensure that all gates are nearest neighbour.}
    \end{subfigure}
    \centering
        \begin{subfigure}[t]{0.95\textwidth}
    \includegraphics[width = 0.9\linewidth]{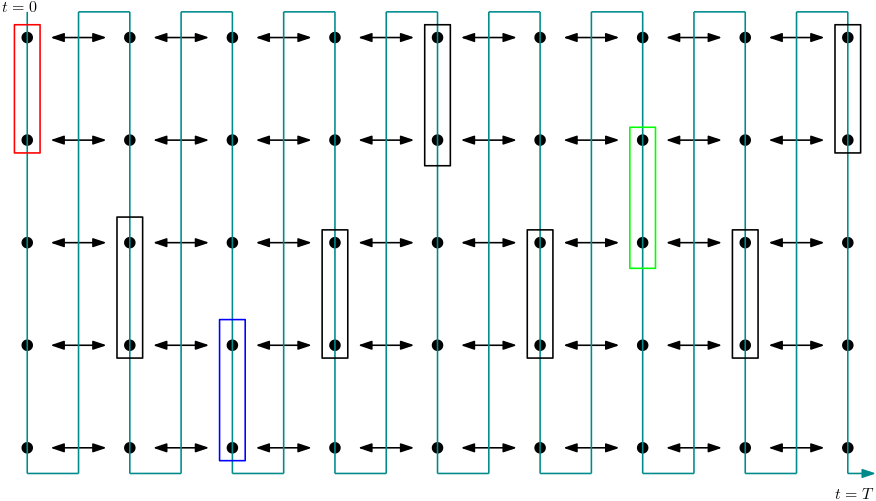}
    \caption{The final step in the sparsification process. We now have a $7\cross 5$ grid of qubits, and the passage of time is denoted by the turquoise line, which starts at the top left and `snakes' through the qubits. At each time step we either apply an identity gate to a particular qubit, a $\SWAP$ gate between qubits in different columns (denoted by an arrow), a $\SWAP$ gate between qubits in the same column (denoted by a black rectangle) or one of the original two qubit gates from \cref{fig:a} (denoted by red, blue and green boxes). In our construction the $\SWAP$ gates within the columns would actually be implemented as three $\CNOT$ gates acting in adjacent columns, however that detail is omitted from the figure for clarity.}
    \end{subfigure}
    \caption{A graphical representation of the `sparsification' process applied in the proof of \cref{ksQSAT}. This sparsification process has previously been used in \cite{Aharanov_2004,Oliveira_2005,Zhou_2021}.}
    \label{fig:circuit_sparsification}
\end{figure}

\subsection{{\sc Homology} with a local boundary operator is $\QMA_1$-complete}

In this section we outline how it is possible to modify the reduction of \cref{lem:ind} so that it always results in an independence complex $I(G)$ with a local boundary operator.

The locality of the boundary operator $\partial$ of $I(G)$ is determined by the maximum degree of $G$ (since the boundary operator can be written as a sum over the vertices of $G$ where at each vertex we construct an operator acting only on adjacent vertices). 
So we need to demonstrate that we can construct the reduction of \cref{lem:ind} in such a way that the resulting graph $G$ has constant degree.

The only modification is that we reduce from an arbitrary problem in $\QMA_1$ via the family of Hamiltonians used to show $\QMA_1$-hardness of quantum $(4,28)$-$\SAT$ in the previous section.
The projectors we need to implement are unchanged, so the proof follows by exactly the same logic as in \cref{lem:ind}.

The vertices of maximal degree in our construction will be those corresponding to mediators from a $\Hpropt(\CNOT)$ gadget acting on qubits which have a $\Hpropt(U_{Pyth.})$ and six $\Hpropt(\CNOT)$ projectors acting on them.
The mediators in the $\Hpropt(\CNOT)$ gadget have maximum degree $11$.
When they are combined with the other gadgets acting on the same qubit, connections are added between these mediators and the $177$ mediators from the two Pythagorean gadgets, and the $110$ mediators from the other five $\CNOT$ gadgets, giving a total degree of 298.
Therefore for the input graphs we are considering the degree is a constant $\delta = 298$.
This results in an independence complex with $O(1)$-local boundary operator as claimed.

\bibliography{References.bib} 
\bibliographystyle{utphys}

\end{document}